\pdfoutput=1
\documentclass[tbtags,authorcolumns,numberwithinsect]{no-lipics-v2022}

\usepackage{amssymb}
\usepackage{mathtools}
\usepackage{bm}
\usepackage{stmaryrd}
\usepackage[draft]{fixme}
\usepackage{booktabs}

\usetikzlibrary{calc}

\SetKwComment{Comment}{$\triangleright$\ \rm}{}

\usetikzlibrary{arrows,arrows.meta,shapes.misc, decorations.pathreplacing,decorations.pathmorphing,calligraphy, patterns.meta}
\usetikzlibrary{calc}
\tikzset{dotmark/.style={circle,fill,inner sep=1.5pt}}
\tikzset{emptymark/.style={circle,draw,fill=white,inner sep=1.5pt}}
\tikzset{crossmark/.style={thick,inner sep=1.5pt}}

\newcommand{\eps}{\varepsilon}

\SetKwInput{KwInput}{Input}
\SetKwInput{KwOutput}{Output}

\def\ShowAuthNotes{1}
\ifnum\ShowAuthNotes=1
\newcommand{\authnote}[3]{\textcolor{#3}{[{\bf #1:} { {#2}}]}}
\else
\newcommand{\authnote}[3]{}
\fi

\newcommand{\unbalanced}[1]{#1}

\newtheorem{assumption}{Assumption}

\newcommand{\Oh}{\mathcal{O}}
\newcommand{\Ohtilde}{\tilde{\Oh}}

\def\fragmentco#1#2{\bm{[}\,#1\,\bm{.\,.}\,#2\,\bm{)}}
\def\fragmentoc#1#2{\bm{(}\,#1\,\bm{.\,.}\,#2\,\bm{]}}
\def\fragmentoo#1#2{\bm{(}\,#1\,\bm{.\,.}\,#2\,\bm{)}}
\def\fragment#1#2{\bm{[}\,#1\,\bm{.\,.}\,#2\,\bm{]}}
\def\position#1{\bm{[}\,#1\,\bm{]}}

\def\twothirds{{}^2{\mskip -3.5mu/\mskip -3mu}_3\,}

\newcommand{\ed}{\mathsf{ed}}

\newcommand{\rev}{\mathsf{rev}}

\newcommand{\dist}{\mathsf{dist}}
\newcommand{\similarity}{\mathsf{sim}}
\newcommand{\sub}{\mathsf{sub}}

\newcommand{\Left}{\mathsf{l}}
\newcommand{\Right}{\mathsf{r}}

\newcommand{\Root}{\mathsf{root}}

\newcommand{\MUL}{\textsf{MUL}\xspace}
\newcommand{\MonMUL}{\textsf{MonMUL}\xspace}
\newcommand{\TED}{\textsf{TED}\xspace}
\newcommand{\SED}{\textsf{SED}\xspace}
\newcommand{\SASED}{\textsf{SASED}\xspace}
\newcommand{\DISED}{\textsf{DISED}\xspace}
\newcommand{\UDISED}{\textsf{UDISED}\xspace}

\newcommand{\ED}{\textsf{ED}\xspace}
\newcommand{\APSP}{\textsf{APSP}\xspace}
\newcommand{\FED}{\textsf{FED}\xspace}
\newcommand{\UFED}{\textsf{UFED}\xspace}
\newcommand{\ETED}{\textsf{All-Subtrees-TED}\xspace}
\newcommand{\BBD}{\textsf{BBD}\xspace}
\newcommand{\LRBBD}{\textsf{LRBBD}\xspace}

\newcommand{\OMv}{\textsf{OMv}\xspace}
\newcommand{\threeSUM}{\textsf{3SUM}\xspace}
\newcommand{\SETH}{\textsf{SETH}\xspace}
\newcommand{\NCSETH}{\textsf{NC-SETH}\xspace}

\newcommand{\mA}{\mathcal{A}}

\newcommand{\sA}{\mathsf{A}}
\newcommand{\sB}{\mathsf{B}}
\newcommand{\sT}{\mathsf{T}}

\newcommand{\bF}{\mathbf{F}}
\newcommand{\bS}{\mathbf{S}}
\newcommand{\bG}{\mathbf{G}}
\newcommand{\bH}{\mathbf{H}}

\newcommand{\bL}{\mathbf{L}}
\newcommand{\bR}{\mathbf{R}}

\newcommand{\bP}{\mathbf{P}}
\newcommand{\bT}{\mathbf{T}}

\newcommand{\w}{\mathrm{w}}
\newcommand{\D}{\mathrm{D}}

\renewcommand{\epsilon}{\varepsilon}

\DeclarePairedDelimiter\abs{\lvert}{\rvert}
\DeclarePairedDelimiter\norm{\lVert}{\rVert}

\makeatletter
\let\oldabs\abs
\def\abs{\@ifstar{\oldabs}{\oldabs*}}
\let\oldnorm\norm
\def\norm{\@ifstar{\oldnorm}{\oldnorm*}}
\makeatother

\def\problembox#1{%
    \vspace{2mm}%
    \noindent\fbox{%
    \begin{minipage}{.985\linewidth}%
        #1
    \end{minipage}%
    }%
    \vspace{2mm}%
}
\newcommand{\defproblem}[3]{%
    \problembox{%
        \textbf{#1}\\
        {\bf{Input:}} #2 ~\\
        {\bf{Output:}} #3
    }%
}

\newcommand{\defproblemwlist}[3]{%
    \problembox{%
        \textbf{#1}\\
        {\bf{Input:}} #2
        {\bf{Output:}} #3
    }%
}

\makeatletter
\renewenvironment{cases}{%
  \matrix@check\cases\env@cases
}{%
  \endarray\right.%
}
\def\env@cases{%
  \let\@ifnextchar\new@ifnextchar
  \left\lbrace
  \def\arraystretch{1.1}%
  \array{@{\;}c@{\quad}l@{}}%
}
\makeatother

\def\mid{\ensuremath :}

\def\emptyset{\varnothing}

\newcommand\numberthis{\addtocounter{equation}{1}\tag{\theequation}}

\newcommand\thefont{\expandafter\string\the\font}

\title{Faster Weighted and Unweighted Tree Edit Distance
\texorpdfstring{\\}{}
and APSP Equivalence}

\author{Jakob Nogler}{ETH Zurich,\\Zurich, Switzerland}{jnogler@ethz.ch}{https://orcid.org/0009-0002-7028-2595}{}
\author{Adam Polak}{Bocconi University,\\Milan, Italy}{adam.polak@unibocconi.it}{https://orcid.org/0000-0003-4925-774X}{}
\author{Barna Saha}{University of California, San Diego\\La Jolla, United States}{barnas@ucsd.edu}{https://orcid.org/0000-0002-6494-3839}{}
\author{Virginia Vassilevska Williams}{Massachusetts Institute of Technology\\Cambridge, United States}{virgi@mit.edu}{https://orcid.org/0000-0003-4844-2863}{}
\author{Yinzhan Xu}{University of California, San Diego\\La Jolla, United States}{xyzhan@ucsd.edu}{https://orcid.org/0009-0002-6809-2514}{}
\author{Christopher Ye}{University of California, San Diego\\La Jolla, United States}{czye@ucsd.edu}{https://orcid.org/0009-0004-0528-5639}{}

\authorrunning{J. Nogler, A. Polak, B. Saha, V. Vassilevska Williams, Y. Xu, and C. Ye}
\titlerunning{Faster Weighted and Unweighted Tree Edit Distance and $\APSP$ Equivalence}

\nolinenumbers
\begin{document}
\maketitle
\begin{abstract}
The tree edit distance ({\sf TED}) between two rooted ordered trees with $n$ nodes labeled from an alphabet $\Sigma$ is the minimum cost of transforming one tree into the other by a sequence of valid operations consisting of insertions, deletions and relabeling of nodes. The tree edit distance is a well-known generalization of string edit distance and has been studied since the 1970s. Its running time has seen steady improvements starting with an $\mathcal{O}(n^6)$ algorithm [Tai, J.ACM 1979], improved to $\mathcal{O}(n^4)$ [Shasha, Zhang, SICOMP 1989] and to $\mathcal{O}(n^3\log{n})$ [Klein, ESA 1998], and culminating in an $\mathcal{O}(n^3)$ algorithm [Demaine, Mozes, Rossman, Weimann, ACM TALG 2010]. The latter is known to be optimal for any dynamic programming based algorithm that falls under a certain decomposition framework that captures all known sub-$n^4$ time algorithms. Fine-grained complexity casts further light onto this hardness showing that a truly subcubic time algorithm for {\sf TED} implies a truly subcubic time algorithm for All-Pairs Shortest Paths ({\sf APSP}) [Bringmann, Gawrychowski, Mozes, Weimann, ACM TALG 2020]. Therefore, under the popular {\sf APSP} hypothesis, a truly subcubic time algorithm for {\sf TED} cannot exist. However, unlike many problems in fine-grained complexity for which conditional hardness based on {\sf APSP} also comes with  equivalence to {\sf APSP}, whether {\sf TED} can be reduced to {\sf APSP} has remained unknown. 

In this paper, we resolve this. Not only we show that {\sf TED} is fine-grained {\em equivalent} to {\sf APSP}, our reduction is tight enough, so that combined with the fastest {\sf APSP} algorithm to-date [Williams, SICOMP 2018] it gives the first ever subcubic time algorithm for {\sf TED} running in $n^3/2^{\Omega(\sqrt{\log{n}})}$ time.

We also consider the unweighted tree edit distance problem in which the cost of each edit (insertion, deletion, and relabeling) is one. For unweighted {\sf TED}, a truly subcubic algorithm is known due to Mao [Mao, FOCS 2022], and later improved slightly by D\"{u}rr [D\"{u}rr, IPL 2023] to run in $\mathcal{O}(n^{2.9148})$ time. Since their algorithm uses bounded monotone min-plus product as a crucial subroutine,  and the best running time for this product is $\tilde{\mathcal{O}}(n^{\frac{3+\omega}{2}})\leq \mathcal{O}(n^{2.6857})$ (where $\omega$ is the exponent of fast matrix multiplication), the much higher running time of unweighted {\sf TED} remained unsatisfactory. In this work, we close this gap and give an algorithm for unweighted {\sf TED} that runs in $\tilde{\mathcal{O}}(n^{\frac{3+\omega}{2}})$ time.
\end{abstract}

\newpage

\section{Introduction}

First introduced by Selkow in the late 1970s \cite{Selkow77}
as a generalization of the more than classical (String) Edit Distance Problem (\ED), the
Tree Edit Distance Problem
(\TED) is a problem of significant interest with applications
spanning computational biology
\cite{gusfield_1997,10.1093/bioinformatics/6.4.309,HochsmannTGK03,waterman1995introduction},
structured data analysis \cite{KochBG03,Chawathe99,FerraginaLMM09},
image processing \cite{BellandoK99,KleinTSK00,KleinSK01,SebastianKK04},
 compiler optimization \cite{DMRW10} and more.

In the classical formulation of \TED, two rooted trees, \( \bT \) and \( \bT' \), are given, 
with nodes arranged in a left-to-right order and labeled from a set \( \Sigma \).
The goal is to compute the \emph{tree edit distance} between \( \bT \) and \( \bT' \),
denoted by \( \ed(\bT, \bT') \),
 defined as the minimum cost required to transform \( \bT \) into \( \bT' \)
using a sequence of valid operations, which can be of three types:
changing a label \( \ell \) to \( \ell' \) at a cost \( \delta(\ell,\ell') \);
removing a vertex with label \( \ell \) at a cost \( \delta(\ell, \epsilon) \), while reattaching its children to its parent in the original order; or
inserting a vertex with label \( \ell \) at a cost \( \delta(\epsilon, \ell) \)
between an existing node and a subsequence of consecutive children of that node.
In the \emph{unweighted tree edit distance} problem, we specify all operations to have cost $1$.

Over the past three decades, the algorithms for \TED have
 been progressively improved, culminating in the current best-known $\Oh(n^3)$ time algorithm by Demaine, Mozes, Rossman, and Weinmann \cite{DMRW10} for two trees on $n$ nodes
(see \cref{tab:complexities} for a summary). 

For the unweighted \TED, a slightly subcubic running time $\Oh(n^{2.9546})$ was recently shown by Mao~\cite{M22} and later improved by D{\"{u}}rr \cite{Durr23} to $\Oh(n^{2.9148})$. 
Nevertheless,
even the $\Oh(n^3)$ time algorithm for the more general weighted \TED was not easy to obtain. The previous best \cite{Klein98} ran in $\Oh(n^3\log n)$ time, and it is unclear whether further logs can be shaved. 

\begin{center}
{\em Question 1: Is there an $o(n^3)$ time algorithm for \TED?}
\end{center}

To address this, \cite{BGMW20} use fine-grained complexity. They show that a truly subcubic time ($\Oh(n^{3-\eps})$ for constant $\eps>0$) algorithm for \TED would imply a truly subcubic time algorithm for the All-Pairs Shortest Paths (\APSP) problem, thus refuting the popular \APSP hypothesis of fine-grained complexity (see the survey \cite{vsurvey}) which states that $n^{3-o(1)}$ is needed for \APSP on $n$-node graphs in the word-RAM model of computation.

This reduction explained why the exponent of the running time is stuck at $3$ but did not address whether any tiny sub-polynomial improvement can be obtained over $n^3$. More frustratingly, no reduction {\em from \textnormal{\TED} to \textnormal{\APSP}} is known to exist. This makes \TED seem {\em different} from the other problems for which \APSP-based conditional lower bounds have been proven (e.g. graph radius, Wiener index, dynamic maximum matching etc., see \cite{vsurvey,VW18,popular14,HenzingerKNS15}) which are all either known to be fine-grained {\em equivalent} to \APSP or whose hardness can be based on even harder problems, such as \OMv \cite{HenzingerKNS15}, or the Exact Triangle problem (which is known to be at least as hard as both \threeSUM and \APSP, see \cite{vsurvey}). 
A tantalizing open question is thus:

\begin{center}
{\em Question 2: Is \TED fine-grained equivalent to \APSP, or can its hardness be based on a harder problem such as Exact Triangle?}
\end{center}

\TED was originally conceived as a generalization of string Edit Distance, and the latter problem has been shown to be very hard within fine-grained complexity: Edit Distance requires $n^{2-o(1)}$ not only under the Orthogonal Vectors conjecture and the Strong Exponential Time Hypothesis (\SETH) \cite{LBStringED15} but also under even more believable hypotheses  such as \NCSETH, and even $\Oh(n^2/\log^c(n))$ time algorithms for Edit Distance for large enough $c$ would imply so-far unattainable Circuit Lower Bounds \cite{LB3StringED15}.

Because of all this, it is conceivable that \TED is similarly difficult  and that (w.r.t. Question 1) only small polylogarithmic improvements are potentially attainable and (w.r.t. Question 2) it is truly harder than \APSP.

\begin{table}[t]
   \centering
   \begin{tabular}{lll} %
      \toprule
      \textbf{Work} & \textbf{Setting} & \textbf{Complexity} \\
      \midrule
      Tai \cite{Tai79} & weighted & $\Oh(n^6)$ \\
      Shasha, Zhang \cite{ShashaZhang89} & weighted & $\Oh(n^4)$ \\
      Klein \cite{Klein98} & weighted & $\Oh(n^3 \log n)$ \\
      Demaine, Mozes, Rossman, Weimann \cite{DMRW10} & weighted & $\Oh(n^3)$ \\
      Bringmann, Gawrychowski, Mozes, Weinmann \cite{BGMW20} & weighted & no $\Oh(n^{3 - \epsilon})$ algorithm under \APSP \\
      Mao \cite{M22} & unweighted & $\Ohtilde(n^{(4 \omega + 12)/(\omega + 5)}) = \Oh(n^{2.9148})$ \\
      \textbf{This work} & \textbf{weighted} & $\bm{n^3/2^{\Omega(\sqrt{\log n})}}$ \\
      \textbf{This work} & \textbf{unweighted} & $\bm{\Ohtilde(n^{(3 + \omega)/2})} = \bm{\Oh(n^{2.6857})}$ \\
      \bottomrule
   \end{tabular}
   \medskip
   \caption{Computational bounds of \TED across different works. 
   In the unweighted setting, complexities are computed using the best known bound on the matrix multiplication exponent $\omega \leq 2.371339$ from \cite{ADVXXZ24}.
   The stated bound for \cite{M22} additionally uses results on Rectangular Bounded Monotone Min-plus Product by D\"{u}rr \cite{Durr23}.
   }
   \label{tab:complexities}
\end{table}

A third question concerns the unweighted \TED problem. Mao \cite{M22} showed that a truly subcubic running time is possible for the problem. The current record for the running time is by D\"{u}rr \cite{Durr23} and is $\Ohtilde(n^{(4\omega + 12)/(\omega + 5}) = \Oh(n^{2.9148})$.\footnote{$\Ohtilde$ hides poly-logarithmic factors. }
Here, $\omega$ is the $n\times n$ matrix multiplication exponent, where one can multiply two $n \times n$ matrices in $\Oh(n^{\omega+\eps})$ time for all $\eps>0$.\footnote{Slightly abusing notation, we omit this $\eps$ from the rest of the paper, which is in line with the literature. The reader should be aware that in many running times in the literature that one states in terms of $\omega$, a secret $\eps$ is always hiding. The running times with a numerical-valued exponent use a strict upper bound on $\omega$ (such as the $\Oh(n^{2.9546})$ running time for unweighted \TED of \cite{M22}), so this $\eps$ does not appear. }
The current bound on $\omega$, due to \cite{ADVXXZ24}, is $\omega \leq 2.371339$.

There are many structured variants of problems related to \APSP
and for a very large number of them, shortly after a truly subcubic time algorithm was found, an $\Ohtilde(n^{(3+\omega)/2})$ time (or better) algorithm was also found.
There are many such examples, we present a few: 
\begin{itemize}
\item The All Pairs Bottleneck Paths \cite{apbp} and All pairs nondecreasing paths \cite{nondecpaths} problems were first shown to have truly subcubic time algorithms in the first decade of the century and are now both known to be solvable in $\Ohtilde(n^{(3+\omega)/2})$ time \cite{DuanJW19,DuanP09}. 

\item The Min-Plus product of $n\times n$ matrices (given $A, B$, compute $C$ with $C[i,j]=\min_k (A[i,k]+B[k,j]$) is known to be equivalent to \APSP in $n$-node graphs \cite{fischermeyer} so that under the \APSP Hypothesis it requires $n^{3-o(1)}$ time. 
Various structured variants of the Min-Plus matrix product have been shown to have truly subcubic time algorithms, e.g. when the matrices have ``bounded differences'' \cite{BGSV19} or whose entries are bounded by $\Oh(n)$ and are monotone (non-decreasing in the rows or columns) \cite{DBLP:conf/soda/WilliamsX20, DBLP:conf/icalp/Gu0WX21}. Later, \cite{CDXZ22} showed that these variants can be solved in $\Ohtilde(n^{(3+\omega)/2})$ time. These structured Min-Plus products have many applications for a variety of fundamental problems, e.g. to many sequence similarity problems such as Language Edit Distance (aka Scored Parsing \cite{aho1972minimum, s17, BGHS19}), RNA Folding \cite{NJ80, akutsu1999RNA, zakov2011RNA, venkatachalam2014RNA, s17} and Dyck Edit Distance \cite{s14,BGHS19,das2021improved,fried2024improved}. All of these problems now have $\Ohtilde(n^{(3+\omega)/2})$ time algorithms.
\end{itemize}

As \TED seems related to \APSP and since its unweighted version is now known to have a truly subcubic algorithm, a natural question is:
\begin{center}{\em Question 3: Can unweighted \TED be solved in $\Ohtilde(n^{(3+\omega)/2})$ time?}\end{center}

Mao's truly subcubic time algorithm for unweighted \TED \cite{M22} can be viewed as a reduction to  Bounded Monotone Min-Plus product which can be solved in $\Ohtilde(n^{(3+\omega)/2})$ time. The reduction, however, is not efficient enough to achieve the same running time for unweighted \TED. One way to resolve Question 3 is by presenting a tight reduction. Is such a reduction possible?

\paragraph*{Our Results.}
We resolve all three questions above.
Similarly to the prior work on \TED, we focus on solving the more general edit distance problem on ordered {\em forests} (rather than just single trees): collections of rooted trees with a left-to-right ordering. Our first main theorem is a fine-grained reduction from (forest) \TED to the Min-Plus product problem mentioned earlier:

\defproblem
{Min-Plus Matrix Multiplication (\MUL)}
{Two $m \times m$ matrices $A = (a_{i,j})$ and $B = (b_{i,j})$.}
{The distance matrix $C = A \star B$, where $C = (c_{i,j})$ is defined as $c_{i,j} = \min_{k \in \fragment{1}{m}} \{a_{i,k} + b_{k,j}\}$.}

Let $\sT_{\MUL}(N)$ denote the running time for solving \MUL on two $N\times N$ matrices, or equivalently (\cite{fischermeyer}) for solving \APSP in $N$ node graphs.

The formal theorem for our reduction is as follows:
\begin{restatable}{mtheorem}{TEDtheorem}
    \label{thm:ted}
    Let $\bF, \bF'$ be forests of size $n = \abs{\bF}, m = \abs{\bF'}$ with $n \geq m$.
    Then, there is an algorithm computing Tree Edit Distance between $\bF, \bF'$ in time $\Ohtilde\left((n/m)^{1+o(1)} \cdot \left(\sT_{\MUL}(m) + m^{2 + o(1)}\right)\right)$.
\end{restatable}

If $m=n$, the theorem states that \TED on $n$-node forests can be solved in $\Ohtilde(\sT_{\MUL}(n))$.
We thus complete the missing direction in establishing the equivalence between \TED and \APSP, resolving Question 2. 

In fact, because our reduction is very efficient and only adds polylogarithmic factors over the \APSP running time, we are also able to resolve Question 1 using Williams' \cite{DBLP:journals/siamcomp/Williams18} \( m^3 / 2^{\Omega(\sqrt{\log m})} \) running time for \APSP.

\begin{corollary}\label{cor:ted}
    Let $\bF, \bF'$ be forests.
    Then, there is an algorithm for \TED running in time $n^3/2^{\Omega(\sqrt{\log n})}$, where $n = \max(|\bF|,|\bF'|)$. \lipicsEnd
\end{corollary}

Beyond providing a faster algorithm for \TED, \cref{cor:ted} underscores once again
the difference in nature between \ED and \TED by demonstrating that,
while we (probably) cannot eliminate an arbitrary number of logarithmic factors for the former,
we can do so for the latter.

The last several \TED algorithms (prior to ours) all use the same dynamic programming approach which was
formalized as a decomposition framework for solving \TED \cite{DT03,DT05}. Formalizing the framework allowed for proving lower bounds on the running time of any algorithm that falls into that framework. The first such lower bound was by \cite{DT03} who showed such algorithms must take $\Omega(n^2\log^2 n)$ time. The approach culminated in obtaining an $\Omega(n^3)$ time lower bound \cite{DMRW10}  for any algorithm that falls within the framework.
Since our algorithm runs faster than cubic time, it falls outside the framework.

We also resolve Question 3 by providing an $\Ohtilde(n^{(3+\omega)/2})$ time algorithm for unweighted \TED. This is a significant improvement over Mao's result \cite{M22}.

We achieve our result via a tight reduction to the Bounded Monotone Min-Plus product problem.

\defproblem
{Bounded Monotone Min-Plus Matrix Multiplication (\MonMUL)}
{An $m\times n$ matrix  $A = (a_{i,j})$ and an $n\times \ell$ matrix $B = (b_{i,j})$ such that either for all $i\in [m], j\in [n]$, $a_{i,j}\leq a_{i,j+1}$ (``row-monotone'') or for all $i\in [m], j\in [n]$, $a_{i,j}\leq a_{i+1,j}$ (``column-monotone'').}
{The distance matrix $C = A \star B$, where $C = (c_{i,j})$ is defined as $c_{i,j} = \min_{k \in \fragment{1}{n}} \{a_{i,k} + b_{k,j}\}$.}

When $m=n=\ell$ and $D=\Oh(n)$, $\MonMUL$ can be solved in $\Ohtilde(n^{(3 + \omega)/2})$ \cite{CDXZ22}. 

We denote by $\sT_{\MonMUL}(m,n,\ell,D)$ the running time for computing $\MonMUL$. When $m=n=\ell=D$, we simply write $\sT_{\MonMUL}(m)$ for the running time of $\MonMUL$.

\begin{restatable}{mtheorem}{unweightedTED}
    \label{thm:unweighted-ted}
    Let $\bF, \bF'$ be two forests of size $n = \abs{\bF}, m = \abs{\bF'}$ with $n \geq m$. Suppose that $\sT_{\MonMUL}(N, N, N, \D) = \Oh(f(N) g(\D))$.
    Then, there is an algorithm computing Unweighted Tree Edit Distance between $\bF, \bF'$ in time $\Ohtilde\left((n/m)^{1+o(1)} \cdot \left( \sT_{\MonMUL}(m) + m^{2 + o(1)} g(m) \right)\right)$.
\end{restatable}

The best known bound for the $\MonMUL$ running time \cite{CDXZ22,Durr23} is $\sT_{\MonMUL}(N, N, N, D) = \Ohtilde(N^{(2+\omega)/2} D^{1/2})$. Hence, \cref{thm:unweighted-ted} implies that given two forests of equal size, we can compute Unweighted Tree Edit Distance in $\Ohtilde(m^{(3+\omega)/2})$ time, matching the complexity of bounded monotone min-plus product~\cite{CDXZ22} and resolving Question 3.

More generally, we obtain the following result.

\begin{corollary}
    \label{cor:unweighted-ted}
   There is an algorithm for Unweighted \TED running in time $n^{1+o(1)} m^{(1 + \omega)/2}$, where $n = \max(|\bF|,|\bF'|), m = \min(|\bF|, |\bF'|)$. \lipicsEnd
\end{corollary}

\subparagraph*{At the core of our results: \TED alignment graphs.}

The edit distance between two length-$n$ strings
can be represented by examining the classical dynamic computation on an $\fragment{1}{(n+1)}\times\fragment{1}{(n+1)}$ grid.
The solution is then traced by following the shortest path from $(1,1)$ to $(n+1,n+1)$.
Such graph-based representation is commonly referred to as an \emph{alignment graph}.

The power of this representation is evident in the numerous results
for string \ED that have either emerged directly from this visualization or can be clearly illustrated through it (including but not limited to \cite{LV88,T06,CKW23,GKS19,CKW23, KNW24, GJKT24,CKM20, Kociumaka23, GK24}). 
Consequently, it is not surprising that specific properties of alignment graphs, such as border-to-border distances, play a central role. For string \ED, these distances can be computed in $\Oh(n^{2})$ time \cite{S95, T06, ACS08, K05}.

Alignment graphs for \TED were introduced as a counterpart to those for \ED and serve as a visualization tool for \TED solutions. However, they appear predominantly in less recent works \cite{T05, BCHMRWZ07, MTWZ09} and have played a less central role than in \ED.

In this work, we reestablish the alignment graph for \TED as a central tool and demonstrate that it provides a more powerful language and visualization framework for capturing the combinatorial structure of solutions than previously recognized. As a technical contribution, we show that in the alignment graph for \TED, border-to-border distances can be computed in \APSP time, while for unweighted \TED, they can be determined in $\Ohtilde(n^{(3 + \omega)/2})$.

\paragraph*{Prior attempts to reduce \TED to \APSP and \MUL.}

At least two prior papers have attempted to leverage min-plus products for computing \TED. 
The first of these works was by Chen \cite{CHEN01},
who introduced a dynamic programming scheme formulated using applications of \MUL.
However, this did not result in a true reduction, as the algorithm ran in $\Oh(n^4)$ time (Chen initially claimed a running time of $\Oh(n^{3.5})$,
but this was later corrected in \cite{SPA17}).
The second work, by \cite{M22},
uses bounded difference min-plus products (a special case of bounded monotone min-plus multiplication)
to achieve truly subcubic time for the unweighted (and very small weight) tree edit distance problem. However, Mao's scheme was not powerful enough to achieve a fine-grained reduction from the general weight case of \TED to \MUL.

\subparagraph*{Other related works.}

Recently, there has been significant interest in \TED approximation algorithms \cite{BGHS19, Seddighin22}
as well as \TED when the distance is bounded \cite{AkmalJin21, Kociumaka22, Kociumaka23}.

\section{Techical Overview}
\label{sec:to}

\subsection{Similarity and alignment graphs}

We examine \TED through its \emph{mapping formulation} (as in~\cite{M22}).
This means that rather than thinking of \TED as finding a least-cost transformation via insertions, deletions, and matches,
we shift our perspective towards seeking a maximum weight mapping between two forests.

\subparagraph*{String similarity.}

To ease into this formulation,
we first discuss the mapping formulation of (weighted) string edit distance.
In the (String) Edit Distance Problem (\ED), we are given two strings
$A = a_1 a_2 \cdots a_n$ and $B = b_1 b_2 \cdots b_m$,
and a cost function $\delta$.
The goal is to find the least cost needed to transform $A$ into $B$,
by substituting, inserting or deleting characters.
The costs $\delta(a_i, b_j)$, $\delta(a_i, \epsilon)$, and $\delta(\epsilon, b_j)$
describe the cost of substituting $a_i$ with $b_j$, deleting $a_i$, and inserting $b_j$, respectively.

We define $\eta(a_i, b_j) \coloneqq \delta(a_i, \epsilon) + \delta(\epsilon, b_j) - \delta(a_i, b_j)$
to be \emph{the weight of matching $a_i$ with $b_j$}.
Determining the string edit distance between $A$ and $B$ translates
into calculating the \emph{similarity between $A$ and $B$},
defined as
\[
   \similarity(A, B) =
   \max\nolimits_{\substack{i_1 < \cdots < i_k \in \fragment{1}{n}\\j_1 < \cdots < j_k \in \fragment{1}{m}}}
   \Big\{ \
      \eta(a_{i_1}, b_{j_1}) + \eta(a_{i_2}, b_{j_2}) + \cdots + \eta(a_{i_k}, b_{j_k})
   \ \Big\}.
\]
Thereby, we obtain $\similarity(A, B) = \sum_i \delta(a_i, \epsilon) + \sum_j \delta(\epsilon, b_j) - \ed(A, B)$.

\subparagraph*{String alignment graphs.}

The value $\similarity(A, B)$ can be computed using the following recurrence
\begin{align*}%
   \similarity(A, B) =
   \begin{cases}
       0, & \text{if $|A| = 0$ or $|B| = 0$,} \\
       \max
       \left\{\begin{aligned}
           &\similarity(A\fragment{2}{n}, B),\\
           &\similarity(A, B\fragment{2}{m}), \\
           &\similarity(A\fragment{2}{n}, B\fragment{2}{m}) + \eta(a_1, b_1)
       \end{aligned}\right\},
       & \text{otherwise}.
   \end{cases}
\end{align*}
These computations can be reformulated as finding a longest path on a directed weighted acyclic graph,
commonly referred to as the \emph{alignment graph}.
This graph has as the vertex set the grid $\fragment{1}{(n+1)} \times \fragment{1}{(m+1)}$,
where each node $(i, j)$ is connected with an edge of weight zero to its right and upper neighbors, i.e., $(i+1, j)$ and $(i, j+1)$, if they are within the borders.
Additionally, from each node $(i, j) \in \fragment{1}{n} \times \fragment{1}{m}$, there is an edge to $(i+1, j+1)$ with weight $\eta(a_i, b_j)$.
Computing $\similarity(A, B)$ translates to finding a longest path between $(1, 1)$ and $(n+1, m+1)$ in this graph.
Each time the longest path traverses an edge from $(i, j)$ to $(i+1, j+1)$ we map $a_i$ to $b_j$.

It is worth noting that in the literature several studies~\cite{S95, T06, ACS08, K05} have focused on computing all longest distances from the lower-left border to the upper-right border in an alignment graph.
For both weighted and unweighted edit distance, this task can be accomplished in time $\mathcal{O}((n + m)^2)$.

\subparagraph*{Tree similarity.}

Similarly, given two forests $\bF$ and $\bF'$,
define $\eta(v, v') \coloneqq \delta(v, \epsilon) + \delta(\epsilon, v') - \delta(v, v')$.
The mapping we consider for \TED corresponds to two sequences of distinct nodes $v_1, \ldots v_k \in \bF$ and $v_1', \ldots v_k' \in \bF'$ such that for all $1 \leq i < j \leq k$:
\begin{itemize}
   \item $v_i$ is an ancestor of $v_j$ in $\bF$ if and only if $v_i'$ is an ancestor of $v_j'$ in $\bF'$,
   \item $v_j$ is an ancestor of $v_i$ in $\bF$ if and only if $v_j'$ is an ancestor of $v_i'$ in $\bF'$, and
   \item if neither $v_i$ nor $v_j$ is the ancestor of the other, $v_i$ comes before $v_j$ in the pre-order traversal of $\bF$
   if and only if $v_i'$ comes before $v_j'$ in the pre-order traversal of $\bF$.
\end{itemize}
The \emph{similarity between $\bF$ and $\bF'$}, denoted as $\similarity(\bF, \bF')$, maximizes $\sum_{1 \le i \le k}{\eta(v_i, v_i')}$,
where the maximum is taken over all such mappings.

As for strings, $\similarity(\bF, \bF') = \sum_{v\in \bF} \delta(v, \epsilon) + \sum_{v'\in\bF'} \delta(\epsilon, v') - \ed(\bF, \bF')$.

\subparagraph*{Forest alignment graphs.}
Suppose we are given $\similarity(\sub(v), \sub(v'))$ for all $v \in \bF$ and $v' \in \bF'$,
where $\sub(v)$ indicates the subtree rooted at node $v$.
Then, the value $\similarity(\bF, \bF')$ can be computed using Shasha and Zhang's recurrence scheme \cite{SZ89}.
Given forests $\bF$ and $\bF'$ with pre-order $v_1, \ldots, v_{|\bF|}$ and $v_1', \ldots, v_{|\bF'|}'$, they compute:

\begin{align*}%
   \similarity(\bF, \bF') =
   \begin{cases}
       0, & \text{if $\bF = \emptyset$ or $\bF' = \emptyset$,} \\
       \max
       \left\{\begin{aligned}
           &\similarity(\bF \setminus v_1, \bF'),\\
           &\similarity(\bF, \bF' \setminus v_1'), \\
           &\similarity(\bF \setminus \sub(v_1), \bF' \setminus \sub(v_1')) + \similarity(\sub(v_1), \sub(v_1'))
       \end{aligned}\right\},
       & \text{otherwise}.
   \end{cases}
\end{align*}

Once again,
these computations can be rephrased as finding the longest path in a directed acyclic graph with a grid as the vertex set,
but only under the condition that we have the values $\similarity(\sub(v), \sub(v'))$ for all $v \in \bF$ and $v' \in \bF'$.
With this condition in place, we can construct a grid $\fragment{1}{(\abs{\bF}+1)} \times \fragment{1}{(\abs{\bF'}+1)}$.
Similar to the case of strings, each node $(i, i')$ is connected to its right and upper neighbors with an edge of weight zero,
i.e., $(i+1, i')$ and $(i, i'+1)$ (if they are within the borders).
Additionally, from each node $(i, i') \in \fragment{1}{\abs{\bF}} \times \fragment{1}{\abs{\bF'}}$,
there is an edge to $(j, j')$ with weight $\similarity(\sub(v_i), \sub(v_{i'}'))$.
Here, $j$ and $j'$ are the smallest integers $j \geq i$ and $j' \geq i'$ such that $v_j \notin \sub(v_i)$ and $v_{j'}' \notin \sub(v_{i'}')$.

For an illustration (and a more formal definition) of such an alignment graph,
refer to \cref{fig:align_graph} in \cref{sec:fed},
where we reduce to \APSP the problem of computing all longest distances from the lower-left border to the upper-right in a forest alignment graph.

\subsection{Warm-up: \TED on caterpillar trees.}

To better understand our algorithm for \TED,
let us first examine a more restricted yet significant case.

\subparagraph*{Caterpillar trees.}
\emph{Caterpillar trees} consist of a central path with nodes labeled $c_1, c_2, \ldots, c_n$, where $c_1$ is the root.
Additional, each central node $c_i$ has both a left child $l_i$ and a right child $r_i$.\footnote{
In standard literature, caterpillar trees can take a more general form. Here, we focus on a specific type of caterpillar tree, though we continue to refer to them simply as caterpillar trees.}

\begin{figure}[htbp]
   \centering
   \usetikzlibrary{matrix}

\begin{tikzpicture}[>=stealth, scale=0.8]
    \foreach \x in {0,...,6}
        \foreach \y in {0,...,6} {
            \node[circle, draw, fill, scale=0.5] (\x\y) at (\x,\y) {};
            \pgfmathtruncatemacro{\lbx}{\x + 1}
            \pgfmathtruncatemacro{\lby}{\y + 1}

            \ifnum\y=0
                \node[below=10pt] at (\x\y) {\lbx};
            \fi
            \ifnum\x=0
                \node[left=10pt] at (\x\y) {\lby};
            \fi
        }

    \foreach \x in {0,...,5}
        \foreach \y in {0,...,6} {
            \pgfmathtruncatemacro{\nextx}{\x + 1}
            \draw[->] (\x\y) -- (\nextx\y);
        }
    \foreach \x in {0,...,6}
        \foreach \y in {0,...,5} {
            \pgfmathtruncatemacro{\nexty}{\y + 1}
            \draw[->] (\x\y) -- (\x\nexty);
        }

    \foreach \x in {0,...,5}
        \foreach \y in {0,...,5} {
            \pgfmathtruncatemacro{\nextx}{\x + 1}
            \pgfmathtruncatemacro{\nexty}{\y + 1}
            \draw[->] (\x\y) -- (\nextx\nexty);
        }

    \draw[dashed] (-0.35, -0.35) -- (-0.35, 1.35) -- (1.35, 1.35) -- (1.35, -0.35) -- (-0.35,-0.35);
    \draw[dashed] (-6.5, -0.5) -- (-3, -0.5) -- (-3, 4) -- (-6.5, 4) -- (-6.5,-0.5);
    \draw[dashed] (-3, -0.5) -- (-0.35, -0.35);
    \draw[dashed] (-3, 4) -- (-0.35, 1.35);

    \draw[blue, opacity=0.7, line width=1.5pt, shift={(-0.07,0.07)}] (0,0) -- (0,1) -- (1,2) -- (1, 3) -- (2, 3) -- (3,3) -- (4,4) -- (4,5) -- (5,5) -- (5,6) -- (6, 6);
    \draw[teal, opacity=0.7, line width=1.5pt, shift={(0.07,-0.07)}] (0,0) -- (0,1) -- (1,2) -- (2, 2) -- (3,3) -- (3,4) -- (4,5) -- (5,5) -- (6, 6);

    \node[circle, draw, red, thick] (gc1) at (0, 1) {};
    \node[circle, draw, red, thick] (gc2) at (3, 3) {};
    \node[circle, draw, red, thick] (gc3) at (4, 5) {};

    \begin{scope}[shift={(-12, 5)}]
        \foreach \x in {0,...,5} {
            \node[circle, draw, fill, scale=0.5] (c\x) at (0,-\x) {};
            \node[circle, draw, fill, scale=0.5] (l\x) at (-0.5,-\x-0.5) {};
            \node[circle, draw, fill, scale=0.5] (r\x) at (+0.5,-\x-0.5) {};
            \draw[->] (c\x) -- (l\x);
            \draw[->] (c\x) -- (r\x);
        }
        \foreach \x in {0,...,4} {
            \pgfmathtruncatemacro{\nextx}{\x + 1}
            \draw[->] (c\x) -- (c\nextx);
        }
        \node[circle, draw, red] (mc1) at (0, 0) {};
        \node[circle, draw, red] (mc2) at (0, -3) {};
        \node[circle, draw, red] (mc3) at (0, -4) {};

        \draw[blue, opacity=0.3, line width=4pt] (-0.5,0) -- (-0.5,-6);
        \draw[teal, opacity=0.3, line width=4pt] (+0.5,0) -- (+0.5,-6);

        \node at (-1.2, 0) {$\bT$};
    \end{scope}

    \begin{scope}[shift={(-9, 5)}]
        \foreach \x in {0,...,5} {
            \node[circle, draw, fill, scale=0.5] (cp\x) at (0,-\x) {};
            \node[circle, draw, fill, scale=0.5] (lp\x) at (-0.5,-\x-0.5) {};
            \node[circle, draw, fill, scale=0.5] (rp\x) at (+0.5,-\x-0.5) {};
            \draw[->] (cp\x) -- (lp\x);
            \draw[->] (cp\x) -- (rp\x);
        }
        \foreach \x in {0,...,4} {
            \pgfmathtruncatemacro{\nextx}{\x + 1}
            \draw[->] (cp\x) -- (cp\nextx);
        }
        \node[circle, draw, red] (mc1p) at (0, -1) {};
        \node[circle, draw, red] (mc2p) at (0, -3) {};
        \node[circle, draw, red] (mc3p) at (0, -5) {};

        \draw[blue, opacity=0.3, line width=4pt] (-0.5,0) -- (-0.5,-6);
        \draw[teal, opacity=0.3, line width=4pt] (+0.5,0) -- (+0.5,-6);

        \node at (1.2, 0) {$\bT'$};
    \end{scope}

    \begin{scope}[shift={(-1, 0.2)}]
        \draw[red] (mc1) -- (mc1p);
        \draw[red] (mc2) -- (mc2p);
        \draw[red] (mc3) -- (mc3p);

        \node[circle, draw, fill, scale=0.5, color=blue] (l11) at (-5,0) {};
        \node[circle, draw, fill, scale=0.5, color=blue] (l21) at (-4,0) {};
        \node[circle, draw, fill, scale=0.5, color=blue] (l12) at (-5,1) {};
        \node[circle, draw, fill, scale=0.5, color=blue] (l22) at (-4,1) {};
        \draw[->, blue] (l11) to node[below] {$0$} (l21);
        \draw[->, blue] (l11) to node[left] {$0$} (l12);
        \draw[->, blue] (l11) to node[right] {$\eta(l_1, l_1')$} (l22);
    \end{scope}

    \begin{scope}[shift={(-0.5, 2)}]
        \node[circle, draw, fill, scale=0.5, color=teal] (r11) at (-5,0) {};
        \node[circle, draw, fill, scale=0.5, color=teal] (r21) at (-4,0) {};
        \node[circle, draw, fill, scale=0.5, color=teal] (r12) at (-5,1) {};
        \node[circle, draw, fill, scale=0.5, color=teal] (r22) at (-4,1) {};
        \draw[->, teal] (r11) to node[below] {$0$} (r21);
        \draw[->, teal] (r11) to node[left] {$0$} (r12);
        \draw[->, teal] (r11) to node[right] {$\eta(r_1, r_1')$} (r22);
    \end{scope}

    \draw[dotted] (l11) -- (r11);
    \draw[dotted] (l12) -- (r12);
    \draw[dotted] (l21) -- (r21);
    \draw[dotted] (l22) -- (r22);

\end{tikzpicture}
   \caption{Overlaying the alignment graphs for string edit distance for $(L, L')$ and $(R, R')$
   leads to an intuitive visualization of \TED on caterpillar trees under the assumption that central nodes are only mapped to central nodes, left children only to left ones, and right children only to right ones (\cref{assm:cated_mapping}).
   The problem can be visualized as two paths in the two graphs, which, whenever they intersect, allow mapping of central nodes to central nodes.}
   \label{fig:cat_trees}
\end{figure}
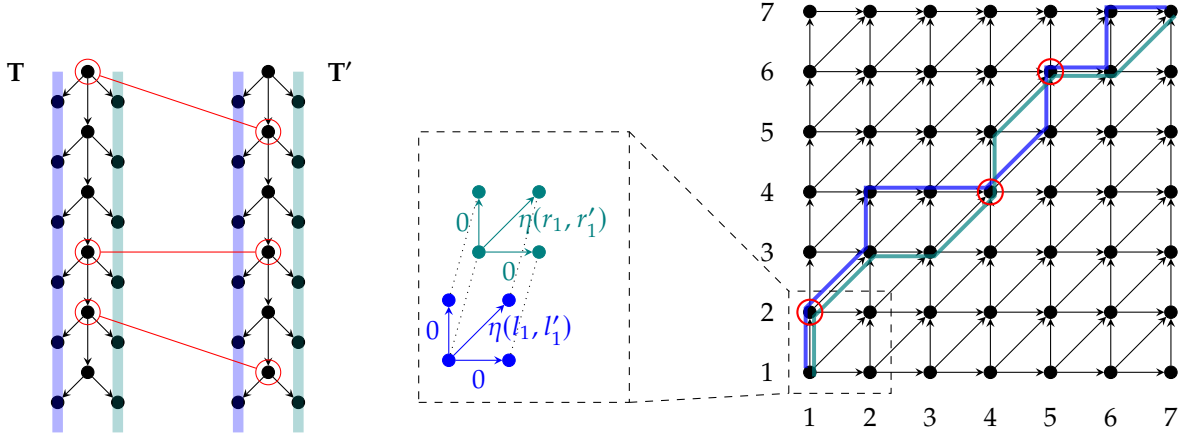

Let us examine the similarity mapping for two such caterpillar trees $\bT$ and $\bT'$,
of size $\abs{\bT} = 3n, \abs{\bT'} = 3n'$ with node sets $\{l_i, c_i, r_i\}_{i\in \fragment{1}{n}}$ and $\{l_i', c_i', r_i'\}_{i\in \fragment{1}{n'}}$.
We examine such a mapping under an additional simplifying assumption: nodes from one side of the caterpillar $\bT$ are always mapped to nodes of the same side of the caterpillar $\bT'$.

\begin{assumption}\label{assm:cated_mapping}
The nodes $\{c_i\}_i$, $\{l_i\}_i$ and $\{r_i\}_i$ are always only mapped by $\similarity(\bT, \bT')$ to nodes in  $\{c_i'\}_i$, $\{l_i'\}_i$ and $\{r_i'\}_i$, respectively.
\lipicsEnd
\end{assumption}

Under \cref{assm:cated_mapping}, suppose $\similarity(\bT, \bT')$ maps $c_{i}$ to $c_{i'}'$ and $c_j$ to $c_{j'}'$ for $i < j$, and no further node $c_{i+1}, \ldots, c_{j-1}$ is mapped.
Then, the nodes from $l_{i}, \ldots, l_{j-1}$ are mapped to nodes from $l_{i'}', \ldots, l_{j'-1'}$
as in $\similarity(L\fragmentco{i}{j}, L'\fragmentco{i'}{j'})$,
where $L = l_{1} \cdots l_{n}, L' = l_{1'}' \cdots l_{n'}'$ are strings built from the left children of the central nodes.
Similarly, the nodes from $r_{i}, \ldots, r_{j-1}$ are mapped to nodes from $r_{i'}', \ldots, r_{j'-1'}$
as in $\similarity(R\fragmentco{i}{j}, R'\fragmentco{i'}{j'})$,
where $R = r_{1} \cdots r_{n}, R' = r_{1'}' \cdots r_{n'}'$ are strings built from the right children of the central nodes.

This brings us to the visualization of $\similarity(\bT, \bT')$ (under \cref{assm:cated_mapping}) illustrated in \cref{fig:cat_trees}.
Consider a $\fragment{1}{(n+1)} \times \fragment{1}{(n'+1)}$ grid, and overlay on it the
alignment graphs for string edit distance for $(L, L')$ and $(R, R')$.
Then, $\similarity(\bT, \bT')$ can be visualized as two paths from $(1,1)$ to $(n+1,n'+1)$ in the two respective alignment graphs.
Whenever these paths intersect at a grid point $(i,i')$, $\similarity(\bT, \bT')$ has the opportunity to map $c_i$ to $c_{i'}'$, provided neither $c_i$ nor $c_{i'}'$ has been mapped previously.

Henceforth, we denote by $\similarity((x, x'), (y,y'))$ the maximum value achievable by a sum of three terms:
\begin{enumerate}[(1)]
    \item the weight of a path from $(x, x')$ to $(n+1, n'+1)$ in the alignment graph of $\similarity(L, L')$;
    \item the weight of a path from $(y, y')$ to $(n+1, n'+1)$ in the alignment graph of $\similarity(R, R')$; and
    \item a sum of values of the form $\eta(c_i, c_{i'}')$ for $(i, i')$ where the two paths intersect,
provided each $c_i$ and $c_{i'}'$ appears at most once.
\end{enumerate}
By the previous discussion, $\similarity(\bT, \bT') = \similarity((1, 1), (1,1))$ under \cref{assm:cated_mapping}.

\subparagraph*{Reduction of caterpillar \TED (under \cref{assm:cated_mapping}) to \APSP.}
To compute $\similarity((1, 1), (1,1))$,
we utilize a divide-et-impera scheme.
Given a rectangle in the grid defined by the lower-left corner $(a,a')$ and the upper-right corner $(b,b')$,
along with $\similarity((x, x'), (y,y'))$ for all $(x,x'),(y,y')$ on the upper-right border of the rectangle,
our task is to determine $\similarity((x, x'), (y,y'))$ for all $(x,x'),(y,y')$ on the lower-left border of the rectangle.
Note that computing the inputs becomes trivial when the rectangle we are considering is the whole grid.

\begin{figure}[htbp]
   \centering
   \usetikzlibrary{decorations.pathmorphing}

\tikzset{snake it/.style={decorate, decoration=snake}}

\begin{tikzpicture}
    \draw[thick] (0,0) node [left] {$(1, 1)$} rectangle (10,6) node [right] {$(n+1, n'+1)$};
    \draw[thick] (2,1) node [below left] {$(a, a')$} rectangle (7,4) node [above right] {$(b, b')$};

    \draw[dashed, thick] (4.5,1) node [below] {$(r, a')$} -- (4.5,4) node [above] {$(r,b')$};

    \draw [rounded corners, color=teal] (4.5, 2) node [left] {$(r,y')$} -- (6, 2) -- (6, 2.5) -- (7, 2.5) node[below right] {$(w, w')$} -- (7.5, 2.5) -- (7.5, 3) -- (8, 3) -- (8, 5.25) -- (8.5, 5.25) -- (8.5, 5.5) -- (9.5, 5.5) -- (9.5, 6) -- (10, 6){};
    \draw [rounded corners, color=blue] (3, 4) node [below] {$(x,b')$} -- (3, 4.5) -- (4, 4.5) -- (4, 5) -- (8, 5) -- (9, 5) -- (9, 5.75) -- (10, 5.75) -- (10, 6){};

    \node[circle, fill, draw, red, thick, scale=0.4] (gc1) at (8, 5) {};
    \node[circle, fill, draw, red, thick, scale=0.4] (gc1) at (9, 5.5) {};

    \draw[orange, line width=4pt, opacity=0.5] (2,4) -- (7,4) -- (7,1);
    \draw[purple, line width=4pt, opacity=0.5] (2,4) -- (2,1) -- (7,1);

\end{tikzpicture}
   \caption{
      The figure illustrates an instance of the recursive scheme used to solve \TED on caterpillars.
      The inputs can be visualized as fixing the starting points of the two paths on the upper right border (in orange) of the rectangle,
      and we are determining the maximum value achievable from there onwards.
      We are required to compute the same values for the lower left border (in purple).
      The divide-et-impera scheme divides the rectangle vertically into two smaller ones,
      parameterized by the corners $(a,a'), (r,b')$ and $(r,a'), (b,b')$.
      The figure also demonstrates how to compute the inputs for the scheme on the left subrectangle for one specific case:
      one path leaves the upper border, and the other exits from the right border.
   }
   \label{fig:rec_scheme}
\end{figure}
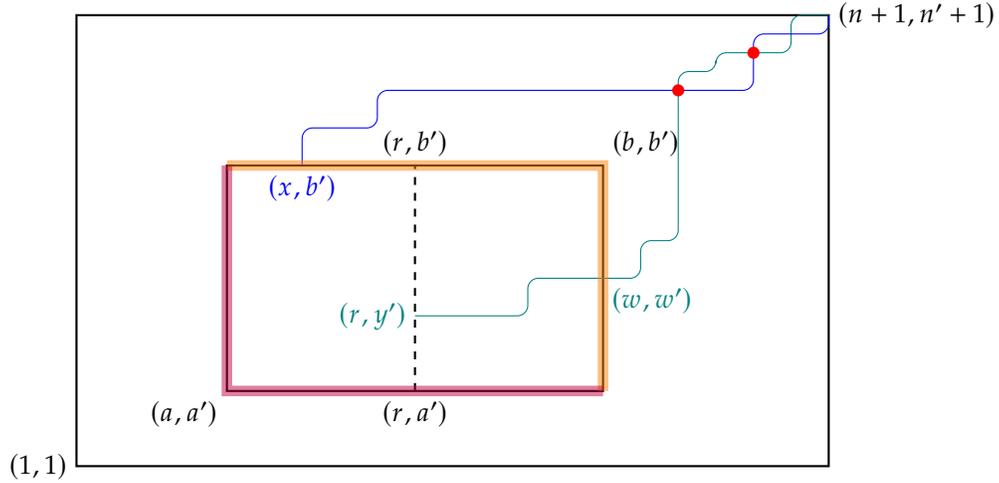

Let $a < r < b$. We split the rectangle vertically into two smaller rectangles (see \cref{fig:rec_scheme}):
one with corners $(a,a')$ and $(r,b')$ and the other with corners $(r,a')$ and $(b,b')$.
We aim to recurse on those two subrectangles.

For the right subrectangle, we can directly recurse since its inputs are a subset of those for the full rectangle.

For the left subrectangle, we must compute its input before recursing on it. First, let us consider $\similarity((x,x'), (y,y'))$ for all $(x,x')$, $(y,y') \in \fragment{a}{r} \times b'$. These values corresponds to the case where both paths cross the upper border $\fragment{a}{r} \times b'$ in \cref{fig:rec_scheme}. Note that these values are part of the input for the entire rectangle.

Similarly, when both paths cross the right border of the left subrectangle, i.e. $r \times \fragment{a'}{b'}$, no work needs to be done, as we can obtain the inputs from the right subrectangle's output.

Next, suppose the path in $\similarity(L, L')$ crosses the upper border at $(x, b')$ for some $x \in \fragment{a}{r}$,
and the path in $\similarity(R, R')$ crosses the right border at $(r, y')$ for some $y' \in \fragment{a'}{b'}$.
Then, we can decompose the path in $\similarity(R, R')$ into two parts (see \cref{fig:rec_scheme}), as
\begin{align} 
    \label{eq:max-plus}
   \similarity((x, b'), (r, y')) = \max_{(w,w') \in (\fragment{r}{b} \times b') \cup (b \times \fragment{a'}{b'})} \Big \{ \ \similarity(R\fragmentco{r}{w}, R'\fragmentco{y'}{w'}) + \similarity((x, b'), (w, w')) \ \Big\}.
\end{align}
In this maximization, the latter summands are part of the input for the full rectangle,
and the former summands are border-to-border paths in an alignment graph which can be computed in time $\Oh(m^2)$,
where $m = \max(b - a, b' - a')$.
\Cref{eq:max-plus} can be computed concurrently for all such $x$ and $y'$
via a max-plus product\footnote{In max-plus products, the minimum operator is replaced by a maximum. Note that min and max products are equivalent, as one can be transformed into the other by appropriately reversing the signs of the matrices. For the remainder of this paper, we treat them as equivalent.} in time $\Oh(\sT_{\MUL}(m))$.

The case where the path in $\similarity(L, L')$ crosses the right border of the left subrectangle,
and the the path in $\similarity(R, R')$ crosses the upper border of the left subrectangle can be handled symmetrically.
This allows us now to recurse on the left subrectangle as well.

It remains to discuss how to patch together the outputs of the two subrectangles
to get the outputs of the full rectangle.
For the sake of brevity, we omit discussing this here, but it is not difficult to see that
by employing similar calculations to those before (border-to-border paths in an alignment graph
combined with the outputs of the two subrectangles via max-plus products), we can calculate
$\similarity((x, x'), (y,y'))$ for all $(x,x'),(y,y')$ on the lower-left border of the rectangle,
ignoring the contribution arising from $c_a$ or $c_{a'}'$ being mapped to other central nodes.
With some additional computations
(which requires redefining $\similarity((x, x'), (y,y'))$ to compute $2^4$ values instead of one, depending on whether central nodes $c_a,c_b, c_{a'}',c_{b'}'$ were already mapped or not), we can factor in these contributions and compute the desired output for the entire rectangle.

This shows how the divide-et-impera scheme can handle vertical cuts.
By symmetry of the problem, the scheme can handle horizontal cuts as well,
allowing us to split a single instance in four roughly equal parts to recurse on.

For a rough analysis, let us focus on the case when $n = n' = 2^k$
for which the scheme always recurses on squares.
This results in the recurrence relation $\sT(n) = 4 \cdot \sT(n/2) + \Oh(\sT_{\MUL}(n))$,
which implies $\sT(n) = \Oh(\sT_{\MUL}(n))$.\footnote{Here, we assume $\sT_{\MUL}(n) = \Omega(n^{2 + \epsilon})$ for some small $\epsilon > 0$, a reasonable assumption given the widely accepted conjecture that no truly subcubic algorithms exist for \MUL.}

\subparagraph*{Getting rid of \cref{assm:cated_mapping}.}

Up to this point, we argued how to determine the similarity between two caterpillar trees, $\bT$ and $\bT'$,
under \cref{assm:cated_mapping}.
In the general case, the mapping of $\similarity(\bT, \bT')$ (when observed from the root downwards) preserves
this assumption up to some point,
i.e., there exist $a,b, a', b'$ such that nodes in $l_{1}, \ldots, l_{a-1}$ are only mapped to nodes in
$l_{1}', \ldots, l_{a'-1}'$, nodes in $r_{1}, \ldots, r_{b-1}$ are only mapped to nodes in
$r_{1}', \ldots, r_{b'-1}'$, nodes in $c_{1}, \ldots, c_{\min(a, b)-1}$ are only mapped to nodes in $c_{1}', \ldots, c_{\min(a', b')-1}'$.
Then, one of two cases can occur:
\begin{enumerate}[(a)]
\item $l_a$ is mapped to $c_{a'}'$, nodes in $l_{a+1}, \ldots, l_{b-1}$ are exclusively mapped to nodes $r_{a'-1}', \ldots, r_{b'+1}'$,
and $c_b$ is mapped to $r_{b'}'$.
The first and third conditions are optional; if they do not hold, we include $l_a$, $r_{a'}'$, and $l_b$, $r_{b'}'$,
respectively, in the middle condition.
\label{it:map:a}
\item$r_a$ is mapped to $c_{a'}'$, nodes in $r_{a-1}, \ldots, r_{b+1}$ are mapped to nodes $l_{a'+1}', \ldots, l_{b'-1}'$,
and $c_b$ is mapped to $l_{b'}'$.
As before, the first and third conditions are optional; if they do not hold, we include $r_a$, $r_{b}$ and $l_{a'}'$, $l_{b'}'$,
respectively, in the middle condition.
\label{it:map:b}
\end{enumerate}
Without loss of generality, we focus on the former case, as the two cases are symmetric.

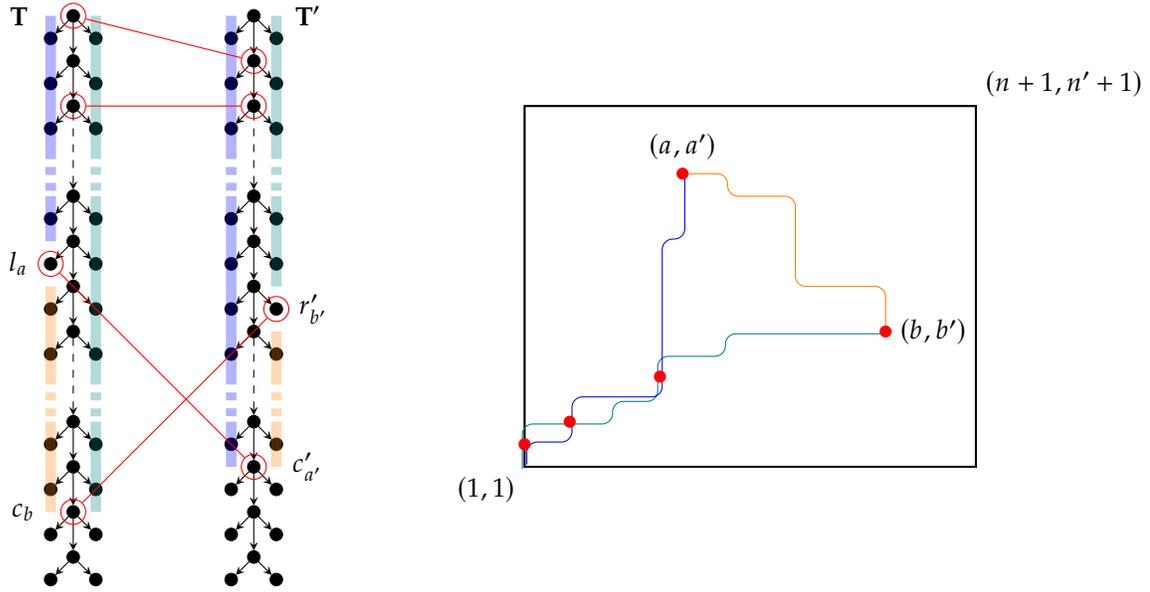
\begin{figure}[htbp]
   \centering
   \usetikzlibrary{matrix}

\begin{tikzpicture}[>=stealth, scale=0.6]

    \draw[thick] (0,-2) node[below left] {$(1, 1)$} rectangle (10,6) node[above right] {$(n+1, n'+1)$};

    \draw[rounded corners, teal, shift={(-0.05,-0.05)}] (0,-2) -- (0,-1) -- (2,-1) -- (2, -0.5) -- (3,-0.5) -- (3, 0.5) -- (4.5, 0.5) -- (4.5, 1) -- (8, 1);
    \draw[rounded corners, blue, shift={(0.05,0.05)}] (0,-2) -- (0,-1.5) -- (1,-1.5) -- (1, -0.5) -- (2, -0.5) -- (3,-0.5) -- (3, 3) -- (3.5,3) -- (3.5, 4.5);
    \draw[rounded corners, orange] (3.5, 4.5) -- (4.5, 4.5) -- (4.5,4)  -- (6, 4) -- (6, 2) -- (8, 2) -- (8, 1);

    \node[circle, fill, draw, red, thick, scale=0.4] (gc1) at (0, -1.5) {};
    \node[circle, fill, draw, red, thick, scale=0.4] (gc2) at (1, -1) {};
    \node[circle, fill, draw, red, thick, scale=0.4] (gc3) at (3, 0) {};

    \node[circle, fill, draw, red, thick, scale=0.4, label={right:$(b,b')$}] (gc4) at (8, 1) {};
    \node[circle, fill, draw, red, thick, scale=0.4, label={above:$(a,a')$}] (gc5) at (3.5, 4.5) {};

    \begin{scope}[shift={(-10, 8)}]
        \foreach \x in {0,...,2} {
            \node[circle, draw, fill, scale=0.5] (c\x) at (0,-\x) {};
            \node[circle, draw, fill, scale=0.5] (l\x) at (-0.5,-\x-0.5) {};
            \node[circle, draw, fill, scale=0.5] (r\x) at (+0.5,-\x-0.5) {};
            \draw[->] (c\x) -- (l\x);
            \draw[->] (c\x) -- (r\x);
        }
        \foreach \x in {0,...,1} {
            \pgfmathtruncatemacro{\nextx}{\x + 1}
            \draw[->] (c\x) -- (c\nextx);
        }

        \foreach \x in {4,...,7} {
            \node[circle, draw, fill, scale=0.5] (c\x) at (0,-\x) {};
            \node[circle, draw, fill, scale=0.5] (l\x) at (-0.5,-\x-0.5) {};
            \node[circle, draw, fill, scale=0.5] (r\x) at (+0.5,-\x-0.5) {};
            \draw[->] (c\x) -- (l\x);
            \draw[->] (c\x) -- (r\x);
        }
        \foreach \x in {4,...,6} {
            \pgfmathtruncatemacro{\nextx}{\x + 1}
            \draw[->] (c\x) -- (c\nextx);
        }

        \foreach \x in {9,...,12} {
            \node[circle, draw, fill, scale=0.5] (c\x) at (0,-\x) {};
            \node[circle, draw, fill, scale=0.5] (l\x) at (-0.5,-\x-0.5) {};
            \node[circle, draw, fill, scale=0.5] (r\x) at (+0.5,-\x-0.5) {};
            \draw[->] (c\x) -- (l\x);
            \draw[->] (c\x) -- (r\x);
        }
        \foreach \x in {9,...,11} {
            \pgfmathtruncatemacro{\nextx}{\x + 1}
            \draw[->] (c\x) -- (c\nextx);
        }

        \draw[dashed, ->] (c2) -- (c4);
        \draw[dashed, ->] (c7) -- (c9);

        \node[circle, draw, red] (mc1) at (0, 0) {};
        \node[circle, draw, red] (mc2) at (0, -2) {};
        \node[circle, draw, red, label={left:$l_a$}] (mc3) at (-0.5, -5.5) {};
        \node[circle, draw, red, label={[label distance=6pt]left:$c_b$}] (mc4) at (0, -11) {};

        \draw[blue, opacity=0.3, line width=4pt] (-0.5,0) -- (-0.5,-3);
        \draw[blue, opacity=0.3, line width=4pt, dashed] (-0.5,-3) -- (-0.5,-4);
        \draw[blue, opacity=0.3, line width=4pt] (-0.5,-4) -- (-0.5,-5);

        \draw[teal, opacity=0.3, line width=4pt] (0.5,0) -- (0.5,-3);
        \draw[teal, opacity=0.3, line width=4pt, dashed] (0.5,-3) -- (0.5,-4);
        \draw[teal, opacity=0.3, line width=4pt] (0.5,-4) -- (0.5,-8);
        \draw[teal, opacity=0.3, line width=4pt, dashed] (0.5,-8) -- (0.5,-9);
        \draw[teal, opacity=0.3, line width=4pt] (0.5,-9) -- (0.5,-11);

        \draw[orange, opacity=0.3, line width=4pt] (-0.5,-6) -- (-0.5,-8);
        \draw[orange, opacity=0.3, line width=4pt, dashed] (-0.5,-8) -- (-0.5,-9);
        \draw[orange, opacity=0.3, line width=4pt] (-0.5,-9) -- (-0.5,-11);

        \node at (-1.2, 0) {$\bT$};
    \end{scope}

    \begin{scope}[shift={(-6, 8)}]
        \foreach \x in {0,...,2} {
            \node[circle, draw, fill, scale=0.5] (cp\x) at (0,-\x) {};
            \node[circle, draw, fill, scale=0.5] (lp\x) at (-0.5,-\x-0.5) {};
            \node[circle, draw, fill, scale=0.5] (rp\x) at (+0.5,-\x-0.5) {};
            \draw[->] (cp\x) -- (lp\x);
            \draw[->] (cp\x) -- (rp\x);
        }
        \foreach \x in {0,...,1} {
            \pgfmathtruncatemacro{\nextx}{\x + 1}
            \draw[->] (cp\x) -- (cp\nextx);
        }

        \foreach \x in {4,...,7} {
            \node[circle, draw, fill, scale=0.5] (cp\x) at (0,-\x) {};
            \node[circle, draw, fill, scale=0.5] (lp\x) at (-0.5,-\x-0.5) {};
            \node[circle, draw, fill, scale=0.5] (rp\x) at (+0.5,-\x-0.5) {};
            \draw[->] (cp\x) -- (lp\x);
            \draw[->] (cp\x) -- (rp\x);
        }
        \foreach \x in {4,...,6} {
            \pgfmathtruncatemacro{\nextx}{\x + 1}
            \draw[->] (cp\x) -- (cp\nextx);
        }

        \foreach \x in {9,...,12} {
            \node[circle, draw, fill, scale=0.5] (cp\x) at (0,-\x) {};
            \node[circle, draw, fill, scale=0.5] (lp\x) at (-0.5,-\x-0.5) {};
            \node[circle, draw, fill, scale=0.5] (rp\x) at (+0.5,-\x-0.5) {};
            \draw[->] (cp\x) -- (lp\x);
            \draw[->] (cp\x) -- (rp\x);
        }
        \foreach \x in {9,...,11} {
            \pgfmathtruncatemacro{\nextx}{\x + 1}
            \draw[->] (cp\x) -- (cp\nextx);
        }

        \draw[dashed, ->] (cp2) -- (cp4);
        \draw[dashed, ->] (cp7) -- (cp9);

        \node[circle, draw, red] (mc1p) at (0, -1) {};
        \node[circle, draw, red] (mc2p) at (0, -2) {};
        \node[circle, draw, red, label={[label distance=6pt]right:$c_{a'}'$}] (mc3p) at (0, -10) {};
        \node[circle, draw, red, label={right:$r_{b'}'$}] (mc4p) at (0.5, -6.5) {};

        \draw[teal, opacity=0.3, line width=4pt] (0.5,0) -- (0.5,-3);
        \draw[teal, opacity=0.3, line width=4pt, dashed] (0.5,-3) -- (0.5,-4);
        \draw[teal, opacity=0.3, line width=4pt] (0.5,-4) -- (0.5,-6);

        \draw[blue, opacity=0.3, line width=4pt] (-0.5,0) -- (-0.5,-3);
        \draw[blue, opacity=0.3, line width=4pt, dashed] (-0.5,-3) -- (-0.5,-4);
        \draw[blue, opacity=0.3, line width=4pt] (-0.5,-4) -- (-0.5,-8);
        \draw[blue, opacity=0.3, line width=4pt, dashed] (-0.5,-8) -- (-0.5,-9);
        \draw[blue, opacity=0.3, line width=4pt] (-0.5,-9) -- (-0.5,-10);

        \draw[orange, opacity=0.3, line width=4pt] (0.5,-7) -- (0.5,-8);
        \draw[orange, opacity=0.3, line width=4pt, dashed] (0.5,-8) -- (0.5,-9);
        \draw[orange, opacity=0.3, line width=4pt] (0.5,-9) -- (0.5,-10);

        \node at (1.2, 0) {$\bT'$};
    \end{scope}

    \draw[red] (mc1) -- (mc1p);
    \draw[red] (mc2) -- (mc2p);
    \draw[red] (mc3) -- (mc3p);
    \draw[red] (mc4) -- (mc4p);

\end{tikzpicture}
   \caption{The general case for similarity mappings on caterpillar trees.}
   \label{fig:cat_map}
\end{figure}

This scenario is depicted in \cref{fig:cat_map}
and can be visualized by overlaying an additional string similarity graph onto the existing ones in the grid (specifically, the one involved in $\similarity(L, \rev(R'))$,
where $\rev(R')$ denotes the reversed string $R'$).
When overlaying this similarity graph, we ensure that the indices of $\rev(R')$ align with the indexing of the grid,
and the paths in this string alignment graph travel from the upper-left border to the lower-right border.
The two paths that were previously observed now extend from $(1,1)$ to $(a,a')$ and $(b,b')$,
where $l_a$ is mapped to $c_{a'}'$ and $c_b$ is mapped to $l_{b'}'$.
Within this new alignment graph, we search for a path connecting $(a+1,a')$ with $(b,b'+1)$.
If $l_a$ is not mapped to $c_{a'}'$ or $c_b$ is not mapped to $l_{b'}'$, then this path starts/ends at $(a,a')$ and $(b,b')$, respectively.

The details necessary for considering this final path can be incorporated into the earlier divide-et-impera framework,
although it necessitates additional inputs and outputs for the scheme. For the sake of brevity, we omit details here.

\subsection{From \TED on caterpillar trees to spine edit distance}

\subparagraph*{Spine edit distance.}

Our algorithm for \TED on caterpillar trees
generalizes to a problem known as \emph{Spine Edit Distance} (\SED),
initially proposed in~\cite{BGHS19}.

In \SED, besides two forests $\bF$ and $\bF'$,
we are provided with a spine for each forest.
A spine $\bS \subseteq \bF$ is any root-to-leaf path within a forest
(in cases where the forest contains multiple trees, the spine can start from any root of a tree contained in the forest).
In \SED we are given the similarity between all pairs of subtrees of $\bF$ and $\bF'$,
provided at least one of the two roots does not lie on a spine.
The task is to compute the similarity for all missing pairs of subtrees.
Put more formally:

\defproblem
{Spine Edit Distance (\SED)}
{Two forests $\bF,\bF'$, two spines $\bS \subseteq \bF, \bS'\subseteq \bF'$, and $\similarity(\sub(v), \sub(v'))$ for all $(v, v') \in (\bF \times \bF') \setminus (\bS \times \bS')$.}
{$\similarity(\sub(v), \sub(v'))$ for all $(v, v') \in \bS \times \bS'$.}

As already noticed in~\cite{BGHS19}, \TED
can be reduced to \SED by employing appropriate tree decompositions.
In \cref{sec:ted}, we will prove the following:

\begin{restatable}{lemma}{sedtoted}\label{lem:sed_to_ted}
Suppose there exists an algorithm for \SED on two forests $\bH,\bH'$
running in time $\sT_{\SED}(m, m') = \Oh(f(m) g(m'))$,
where $m = |\bH|, m' = |\bH'|$ and $f(m) = \Omega(m), g(m') = \Omega(m')$ are some functions.
Then, there is an algorithm for \TED on two forests
$\bF,\bF'$ running in time $\Oh(f(n) g(n') \log^2 \max(n', n))$,
where $n = |\bF|, n' = |\bF'|$.
\end{restatable}

\subparagraph*{Reducing spine edit distance to \APSP.}
In the remaining part of this (sub)section, we describe the main ingredients we need to reduce \SED to \APSP.

We can assume, without loss of generality, that $\bF$ and $\bF'$ are trees (by adding a virtual root and defining weights accordingly to enforce their deletion).

Let us further decompose $\bF$ and $\bF'$ into $\bL, \bS, \bR$ and $\bL', \bS', \bR'$,
where $\bR$ is obtained from $\bF$ by removing all nodes in $\bS$ and those to the left of it.
Similarly, $\bL$ is obtained from $\bF$ by removing all nodes in $\bS$ and those to the right of it.
We define $\bL'$ and $\bR'$ symmetrically.
For simplicity, let us make a similar assumption as in the caterpillar case.

\begin{assumption}\label{assm:sed_mapping}
    For all $(v, v') \in \bS \times \bS'$,
    nodes from $\bL, \bS$, and $\bR$ are always only mapped by $\similarity(\sub(v), \sub(v'))$
    to nodes of $\bL', \bS'$, and $\bR'$, respectively. \lipicsEnd
\end{assumption}

At this point,
we can attempt to approach the problem similarly to how we did for caterpillar trees.
Provided with the values $\similarity(\sub(v), \sub(v'))$ for all $(v, v') \in (\bF \times \bF') \setminus (\bS \times \bS')$ in \SED,
we can construct the forest alignment graph of $\similarity(\bL, \bL')$ and $\similarity(\bR, \bR')$,
aiming to overlay them on top of each other.
Refer to \cref{fig:cat_to_sed} for a visualization and for some more (informal) discussion.

\begin{figure*}[htbp]
   \centering
   \input{tikz/cat_to_sed}
   \caption{
    To compute \(\similarity(\bF, \bF')\) for the two depicted trees \(\bT\) and \(\bT'\) under \cref{assm:sed_mapping}, we can use the same visualization approach as before. Specifically, we overlay two tree alignment graphs corresponding to the concatenated left and right subtrees, tracing two paths that, whenever they intersect at two spine nodes, provide the possibility to map them. These graphs correspond to \(L_1 L_2 \cdots L_m\) vs. \(L_1' L_2' \cdots L'_{m'}\) and \(\rev(R_1 R_2 \cdots R_m)\) vs. \(\rev(R_1' R_2' \cdots R'_{m'})\), shown in blue and teal, respectively.
    Given the similarity values \(\similarity(\sub(v), \sub(v'))\) for all \((v, v') \in (\bF \times \bF') \setminus (\bS \times \bS')\), we can construct these trees. \\ 
    In Spine Edit Distance (\SED), instead of computing only \(\similarity(\bF, \bF')\), we also need to determine \(\similarity(\sub(s), \sub(s'))\) for all \((s, s') \in (\bS \times \bS')\). Fortunately, by employing a similar divide-and-conquer approach as used for caterpillars, computing \(\similarity(\bF, \bF')\) naturally leads to obtaining these values as well.  \\
    When designing such a divide-and-conquer scheme, the subproblems must be indexed by rectangles whose edges align with coordinates corresponding to spine nodes. This ensures that no diagonal edges in the tree alignment graphs ``jump over'' the sides of the rectangle, allowing for a ``clean'' partitioning into subproblems. \\
    As with caterpillar trees, removing \cref{assm:sed_mapping} would introduce a third path between the two existing ones.
    }
   \label{fig:cat_to_sed}
\end{figure*}

However, this approach presents two challenges:
\begin{itemize}
   \item The two forest alignment graphs might have different grid sizes.
   \item Identifying meeting points of paths in the two alignment graphs where nodes from $\bS$ can be mapped to nodes of $\bS'$ is not as straightforward as for \TED on caterpillar trees.
\end{itemize}
In \cref{sec:sed}, despite these challenges, we prove the following result:
\begin{restatable}{mtheorem}{sed}\label{thm:sed}
Suppose there exists an algorithm computing the min-plus product of two $m \times m$ matrices in time $\sT_{\MUL}(m)$.
Then, there is an algorithm for \SED running in time $\Oh(\sT_{\MUL}(n) + n^{2+o(1)})$, where $n = \max(|\bF|,|\bF'|)$. 
\end{restatable}

\noindent
Observe that \cref{thm:sed} together with \cref{lem:sed_to_ted} yields \cref{thm:ted}.
Moreover,
note that we establish not only the equivalence between \TED and \APSP,
but also between \SED and \APSP.

Central to overcome the aforementioned challenges
is the notation introduced by Mao in \cite{M22}.
This notation allows us to more directly formalize the problem in terms of forests,
without relying on paths as an intermediary.
Additionally, it leads to a more concise description of the
divide-et-impera strategy that we employ to solve \SED.
However, this comes at the cost of obscuring the intuitive understanding provided by the two crossing paths,
which remain essential for visualizing the structure of the problem.

\subparagraph*{Forest edit distance.}

Generalizing \TED on caterpillar trees on \SED involves a shift from alignment graphs on strings to alignment graphs on forests.
While there are algorithms finding border-to-border paths in the former in quadratic time,
we still need to come up with an efficient way to compute such distances in the latter.
The objective here is to develop either a truly subcubic algorithm or a reduction to \APSP.
We call the problem of finding border-to-border distances in a forest alignment graph the \emph{Forest Edit Distance} problem (\FED),
and we present it here formulated using Mao's notation (to be introduced in the next section).\footnote{We remark that a different variant of \FED was introduced already in \cite{BGHS19}.
There, it is defined with the same input, but the only required output is $\similarity(\bF, \bF')$.}

Before introducing \FED, let us briefly introduce some notation and the bi-order traversal of a forest, already used in \cite{M22}.
The bi-order traversal of a forest $\bF$ is a sequence of length $2 \abs{\bF}$ generated by starting a depth-first traversal from the virtual root, and adding a node to the sequence whenever we enter or leave a node. 
We then use $\bF\fragmentco{x}{y}$ to denote the induced sub-forest of $\bF$ consisting of nodes that appear twice in the segment $\fragment{x}{y-1}$ in the bi-order traversal of $\bF$.
For example $\bF = \bF\fragmentco{1}{2\abs{\bF} + 1}$.
For a node $v \in \bF$, we let $\Left(v)$ denote the first occurrence of $v$ in the bi-order traversal of $\bF$ and $\Right(v)$ denote one \emph{plus} the last occurrence of $v$.

\defproblem
{Forest Edit Distance (\FED)}
{Two forests $\bF$ and $\bF'$ and $\similarity(\sub(v), \sub(v'))$ for all $(v, v') \in \bF \times \bF'$.}
{The following values:
\begin{itemize}
\item $\similarity(\bF\fragmentco{x}{y}, \bF')$ for all $x,y \in \fragment{1}{(2|\bF|+1)}$,
\item $\similarity(\bF\fragmentco{x}{(2|\bF|+1)}, \bF'\fragmentco{1}{y'})$ for all $x \in \fragment{1}{(2|\bF|+1)}, y' \in \fragment{1}{(2|\bF'|+1)}$,
\item $\similarity(\bF, \bF'\fragmentco{x'}{y'})$ for all $x',y' \in \fragment{1}{(2|\bF'|+1)}$, and
\item $\similarity(\bF\fragmentco{1}{y}, \bF'\fragmentco{x'}{(2|\bF|+1)})$ for all $y \in \fragment{1}{(2|\bF|+1)}, x' \in \fragment{1}{(2|\bF'|+1)}$.
\end{itemize}
}

\medskip

In \cref{sec:fed}, we demonstrate that \FED can indeed be reduced to \APSP, and we prove:

\begin{restatable}{mtheorem}{fed}
    \label{thm:fed}
    Suppose there exists an algorithm computing the min-plus product of two $m \times m$ matrices in time $\sT_{\MUL}(m)$.
    Then, there is an algorithm for \FED running in time $\Oh(\sT_{\MUL}(n) + n^{2+o(1)})$, where $n = \max(|\bF|,|\bF'|)$.
\end{restatable}

To address \FED, we adopt a divide-et-impera approach.
Our recursive scheme presented in \cref{thm:fed} builds upon Mao's decomposition scheme for forests introduced in \cite{M22}.

\subsection{Unweighted Tree Edit Distance}
\label{sec:unweighted-ted-overview}

We now discuss our algorithm for Unweighted Tree Edit Distance.
Consider two forests $\bF, \bF'$ with $n, n'$ nodes respectively.
We can make two simple observations with respect to the matrices involved in max-plus products in our reduction:
\begin{itemize}
    \item The similarity matrices are row-monotone and column-monotone.
    \item The entries of the similarity matrices are bounded by $\Oh(\min(n, n'))$.
\end{itemize}

Note that the first observation holds even in the weighted setting. 
For example, in Equation~\eqref{eq:max-plus}, for a fixed $x$, $\similarity((x, b'), (r, y'))$ must be non-decreasing in $y'$. 
In the unweighted setting, the second observation additionally holds, since the value of any alignment is at most twice the number of nodes.
In particular, in computing \FED (\Cref{thm:fed}) and \SED (\Cref{thm:sed}), we can instead use Bounded Monotone Min-Plus Matrix Multiplication.

Recall that the min-plus (or max-plus) product of an $n \times n$ arbitrary integer matrix and an $n \times n$ bounded monotone matrix can be computed in time $\Ohtilde(n^{(3 + \omega)/2})$ \cite{CDXZ22}, with subsequent generalizations to arbitrary rectangular matrices \cite{Durr23, SY24}.

Directly applying the observation, we have that \SED between two forests of size $m, m'$ can be computed in time $\Ohtilde(\max(m, m')^{(3 + \omega)/2})$ \cite{CDXZ22} which may be improved to $\Ohtilde(\sqrt{\min(m, m')} \cdot \max(m, m')^{(2 + \omega)/2})$ by our bound on entries \cite{Durr23}.

Let us see the result we can obtain via \Cref{lem:sed_to_ted}.

Recall that \Cref{lem:sed_to_ted} states that given an $\Oh(f(m) g(m'))$ algorithm for \SED, there is an $\Ohtilde(f(n) g(n'))$ algorithm for \TED on two forests of size $n, n'$.
Suppose $f(x) = \Oh(x^{a}), g(x) = \Oh(x^{b})$.
Setting $m' = 1$, we have $m^{(2 + \omega)/2} = \Oh(m^{a} m'^{b}) = \Oh(m^{a})$, so that $a \geq (2 + \omega) / 2$. 
Similarly, we obtain $b \geq (2 + \omega) / 2$, so that we only obtain a \TED algorithm running in time $\Oh(n^{2 + \omega})$ for two forests of size $n$.
Even if $\omega = 2$, the running time $\Oh(n^{4})$ is prohibitively expensive.

In general, a running time $\Ohtilde(\min(m, m')^c \max(m, m')^d)$ can be upper bounded by $\Ohtilde(m^{\max(c, d)} m'^{d})$. Hence, in order to keep the total exponents of the two formulations the same, we hope to obtain a running time with $c \ge d$ without increasing $c + d$, i.e., we hope for an algorithm for \SED with running time $\Ohtilde(\min(m, m')^c \max(m, m')^{d})$ for $c \ge d$ and $c + d = (3+\omega)/2$. This is achieved by the following theorem:

\begin{restatable}{theorem}{unweightused}
    \label{thm:unweighted-used}
    There is an $(n/n')^{1+o(1)} \cdot \left(\sT_{\MonMUL}(n') + n'^{2 + o(1)} g(n') \right)$ algorithm for unweighted \SED, where $n = \abs{\bF}, n' = \abs{\bF'}$, $n \geq n'$, and $\sT_{\MonMUL}(n', n', n', \D) = \Oh(f(n') g(\D))$ for some functions $f, g$.
\end{restatable}

We now illustrate the techniques required to obtain the above result.
Suppose we have two trees $\bF, \bF'$ of size $n, n'$ with spines $\bS = \{r, \cdots, q\}, \bS' = \{r', \cdots, q'\}$.
Assume without loss of generality that $n \gg n'$.
For simplicity, we again assume \Cref{assm:sed_mapping}.

We can formulate the \SED problem using a divide-et-impera scheme.
Since $n \gg n'$, we only decompose the larger tree $\bF$.
For a given threshold $\gamma$, we can efficiently find a subsequence of spine vertices $I =\fragment{s_{1}}{s_{d}}$ with $r = s_{1}, q = s_{d}$ such that for all $i < d$, either $\abs{\sub(s_{i}) \setminus \sub(s_{i + 1})} \leq \gamma$ or $s_{i}$ immediately precedes $s_{i + 1}$ in $\bS$.

We can then compute \SED recursively in a bottom-up manner.
We begin by computing $\similarity(\sub(s_{d}), \bF'\fragmentco{x'}{y'})$ for all $x', y' \in \fragmentco{1}{2\abs{\bF'} + 1}$.
Since $\sub(s_{d})$ is a single node, we can compute this efficiently.
Next, for $i < d$, given $\similarity(\sub(s_{i + 1}), \bF'\fragmentco{x'}{y'})$ for all $x', y'$, our hope is to compute $\similarity(\sub(s_{i}), \bF'\fragmentco{x'}{y'})$ for all $x', y'$.
First, suppose $\abs{\sub(s_{i}) \setminus \sub(s_{i + 1})} \leq \gamma$.
In this case we recurse on the smaller problem of size $(\gamma, n')$.
Since there are at most $3n/\gamma$ sub-problems in $I$, the total time on the first case is
\begin{equation*}
    \sT(n, n') = (3n/\gamma) \cdot \sT(\gamma, n')
\end{equation*}
which yields total time $(n/n')^{1 + o(1)} \cdot n'^{(3 + \omega)/2}$, since when $\gamma = \Oh(n')$, we apply the recursive scheme employed in \Cref{thm:sed} in time $n'^{(3 + \omega)/2}$.

Thus, it remains to handle the case where $s_{i}$ immediately precedes $s_{i + 1}$.
Under \Cref{assm:sed_mapping}, we can decompose the similarity when $s_{i}$ is not aligned as follows:
\begin{align*}
    \similarity(\sub(s_{i}), \bF'\fragmentco{x'}{y'})
    = \max\nolimits_{\substack{
        w' \in \fragment{\Left(r')}{\Left(q')} \\
        z' \in \fragment{\Right(q')}{\Right(r')}}} \Big\{ \
    &\similarity(\bF\fragmentco{\Left(s_{i})}{\Left(s_{i + 1})}, \bF'\fragmentco{x'}{w'}) \\
    &+ \similarity(\sub(s_{i + 1}), \bF'\fragmentco{w'}{z'}) \\
    &+ \similarity(\bF\fragmentco{\Right(s_{i + 1})}{\Right(s_{i})}, \bF'\fragmentco{z'}{y'}) \ \Big\}. \numberthis \label{eq:unweighted_assm_b}
\end{align*}
To handle the case where $s_{i}$ is aligned, say to vertex $v$, we can add $\eta(s_{i}, v)$ to the value of $\similarity(\sub(s_{i}), \bF'\fragmentco{\Left(v) + 1}{\Right(v)})$.
For completeness, we also ensure monotonicity after updating the values in a post-processing step (i.e. $\similarity(\sub(s_i), \bF'\fragmentco{x'}{y'}) \geq \sub(s_{i}, \bF'\fragmentco{\Left(v)}{\Right(v)})$ if $x' \preceq \Left(v) \prec \Right(v) \preceq y'$).
The computation can be written as a max-plus product of three $\Oh(n') \times \Oh(n')$ matrices with entries bounded by $\Oh(n')$ since each entry is a similarity measure between two forests, one of size at most $\Oh(n')$.
Since the matrices are also monotone, the summands can be combined in time $\sT_{\MonMUL}(n') = \Ohtilde(n'^{(3 + \omega)/2})$.

\begin{figure}[htbp]
   \centering
   \usetikzlibrary{matrix}

\begin{tikzpicture}[>=stealth, scale=0.6]

    \draw[thick] (2,-2) rectangle (12,6) ;

    \draw[rounded corners, orange, shift={(-0.05,-0.05)}] (2,-1) -- (3,-1) -- (3, -0.5) -- (5.5, 1) -- (7, 1) -- (7, 3.5) -- (9, 3.5) -- (10, 5) -- (10.5, 5) -- (10.5, 5.5) -- (11.5, 5.5) -- (12, 5.5);
    \draw[rounded corners, blue, shift={(0.05,0.05)}] (2,0) -- (3, 0) -- (5, 1) -- (5, 2) -- (7, 2) -- (8, 3) -- (8.5, 3) -- (10, 3) -- (10, 4.5) -- (12, 4.5);

    \node[circle, fill, draw, red, thick, scale=0.4, label={left:$(\Left(s_{i}), x')$}] (gc4) at (2, -1) {};
    \node[circle, fill, draw, red, thick, scale=0.4, label={left:$(\Right(s_{i}, y')$}] (gc5) at (2, 0) {};
    
    \node[circle, fill, draw, red, thick, scale=0.4, label={right:$(\Left(s_{i + 1}, w')$}] (gc4) at (12, 5.5) {};
    \node[circle, fill, draw, red, thick, scale=0.4, label={right:$(\Right(s_{i + 1}, z')$}] (gc5) at (12, 4.5) {};

    \begin{scope}[shift={(-10, 0)}]
        \draw (0,-1) -- (4,-1) -- (2, 5) -- (0, -1);
        \draw (2, 5) -- (2, -1);
        \draw[fill, opacity=0.5, orange] (0.5, -1) -- (2, 3.5) -- (2, 0.5) -- (1.5, -1) -- (0.5, -1);
        \draw[fill, opacity=0.5, blue] (2, 3.5) -- (3.5,-1) -- (2.5, -1) -- (2, 0.5) -- (2, 3.5);
        \draw[dotted] (2, 3.5) -- (1, 3.5) node[left] {$s_{i}$};
        \draw[dotted] (2, 0.5) -- (0, 0.5) node[left] {$s_{i + 1}$};
    \end{scope}

    \begin{scope}[shift={(-5, 0)}]
        \draw (0,-1) -- (4,-1) -- (2, 5) -- (0, -1);
        \draw (2, 5) -- (2, -1);
        \draw[fill, opacity=0.5, orange] (0.25, -1) -- (2, 4) -- (2, 1) -- (1.25, -1) -- (0.25, -1);
        \draw[fill, opacity=0.5, blue] (2, 3.5) -- (3.75,-1) -- (3, -1) -- (2, 0.5) -- (2, 3.5);
        \draw[dotted] (2, 5) -- (1, 5) node[left] {$r'$};
        \draw[dotted] (2, -1) -- (2, -1.25) node[below] {$q'$};
        \draw[dotted] (0.25, -1) -- (0.25, -1.25) node[below] {$x'$};
        \draw[dotted] (1.25, -1) -- (1.25, -1.25) node[below] {$w'$};
        \draw[dotted] (3, -1) -- (3, -1.25) node[below] {$z'$};
        \draw[dotted] (3.75, -1) -- (3.75, -1.25) node[below] {$y'$};
    \end{scope}
    
\end{tikzpicture}
   \caption{The alignment whose value is computed in \Cref{eq:unweighted_assm_b} visualized as a Border-to-Border (\BBD) distance computation. 
   The value of the orange path corresponds to the first summand, or the alignment between $\bF\fragmentco{\Left(s_{i})}{\Left(s_{i + 1}}$ and $\bF'\fragmentco{x'}{w'}$.
   The value of the blue path corresponds to the last summand, or the alignment between $\bF\fragmentco{\Right(s_{i + 1})}{\Right(s_{i}}$ and $\bF'\fragmentco{z'}{y'}$.
   For every pair of points on the right border, we have the optimal alignment between $\sub(s_{i + 1})$ and $\bF'\fragmentco{w'}{z'}$.}
   \label{fig:u_sed_overview}
\end{figure}

The second summand is available to us as part of the bottom-up recursion. 
It remains to compute the first and last summand.
Consider the first summand.
Observe that it is exactly the third output of \FED between $\bF\fragmentco{\Left(s_{i})}{\Left(s_{i + 1})}$ and $\bF'\fragmentco{\Left(r')}{\Left(q')}$.
Similarly, the third summand is exactly the third output of \FED between $\bF\fragmentco{\Right(s_{i + 1})}{\Right(s_{i})}$ and $\bF'\fragmentco{\Right(q')}{\Right(r')}$.
Since none of the forests contain any spine nodes $\bS, \bS'$, all inputs to the \FED instances are given by the inputs of the \SED instance.

Our goal is now to compute these two \FED instances.
We can bound the size of $\bF\fragmentco{\Left(s_{i})}{\Left(s_{i + 1})}$ and $\bF\fragmentco{\Right(s_{i + 1})}{\Right(s_{i})}$ by $\abs{\sub(s_{i}) \setminus \sub(s_{i + 1})}$.
However, it is possible that $\abs{\sub(s_{i}) \setminus \sub(s_{i + 1})} = \Theta(n)$, in which case directly applying \Cref{thm:fed} on the \FED instance already takes $\Oh(\sT_{\MonMUL}(n))$ time in this one step alone. Even if the entries are bounded by $\Oh(n')$, the output of \FED instance is already size $\Omega(n^2)$.
How can we reduce the dependence on $n$?
In our \SED application, we have only used the third output of \FED. 
In fact, we show that all outputs except the first output (which has size $\Omega(n^2)$) can be computed efficiently.
Formally, we define the Unbalanced Forest Edit Distance (\UFED) problem as computing all but the first output of \FED and show the following theorem.
\sloppy
\begin{restatable}{theorem}{UnweightFEDUnbalance}
    \label{thm:fed-unbalance-unweighted}
    \unbalanced{There is an algorithm for unweighted \UFED running in time $(n/n')^{1+o(1)} \cdot \left( \sT_{\MonMUL}(n') + n'^{2 + o(1)} g(n') \right)$ where $n = \abs{\bF}$, $n' = \abs{\bF'}$, $n \geq n'$ and $\sT_{\MonMUL}(n', n', n', \D) = \Oh(f(n') g(\D))$  for some functions $f, g$.} 
\end{restatable}

In the remainder of the technical overview of the unweighted algorithm, we give an overview of the proof of \Cref{thm:fed-unbalance-unweighted}. 
Further simplifying under \Cref{assm:sed_mapping}, our previous discussion shows that it is enough to only compute the third \FED output.
To this end, consider the alignment graph $\bar{\bG}$ of two forests $\bF, \bF'$ with vertices $\bF \times \bF'$ arranged in a grid where both trees are ordered in pre-order traversal.

The grid contains horizontal and vertical edges of weight $0$ (corresponding to vertex deletions) and edges from $(v, v')$ to $(w, w')$ of weight $\similarity(\sub(v), \sub(v'))$ where $w$ (resp.\ $w'$) is the first vertex after $\sub(v)$ (resp.\ $\sub(v')$) in pre-order traversal i.e.\ the first vertices that can be aligned after aligning $\sub(v)$ and $\sub(v')$.
The desired output of \FED then corresponds to computing all distances from the left border to the right border of the grid.
In contrast, the general \FED problem asks to compute all distances from the left and bottom borders to the right and top borders.

The alignment graph $\bG$ has dimension $n \times n'$, so ideally we hope to be able to breaking the it into $\Oh(n/n')$ subgraphs, each of size $\Oh(n') \times \Oh(n')$. Then, we can compute \FED on each subgraph and combine the results using $\Oh(n/n')$ bounded monotone max-plus products, thus computing the desired \FED output in total time $\Oh((n/n') \cdot \sT_{\MonMUL}(n'))$. However, in reality, it is not simple to decompose $\bG$ to $\Oh(n/n')$ subgraphs and combine the results, as there can be edges connecting two nodes that are far away in the alignment graph. Hence, in our actual algorithm, we apply a divide-et-impera scheme utilizing Mao's tree decomposition~\cite{M22}, which incurs an additional $(n/n')^{o(1)}$ factor. Still, we obtain the desired almost-linear dependence on the size of the larger forest $n = \abs{\bF}$.
In particular, at least under \Cref{assm:sed_mapping}, both \FED and \SED can be computed in time $\Oh(\min(m, m')^{(1 + \omega)/2} \max(m, m')^{1+o(1)})$.

To remove \Cref{assm:sed_mapping}, we must be able to additionally compute distances from the left border to the top border in \BBD instances, corresponding to the second and fourth outputs of the \FED problem (see \Cref{sec:fed-unbalance}).
In the \SED instance we then proceed by careful case analysis to ensure that we consider all possible alignments between the two forests (see \Cref{sec:sed-unbalance}).
Applying \Cref{lem:sed_to_ted} then yields an $n^{1 + o(1)} n'^{(1 + \omega)/2}$-time algorithm for unweighted tree edit distance between forests of size $n, n'$ and $n \geq n'$, proving \Cref{cor:unweighted-ted}.

\subsection{Organization of the paper}

The structure of this paper is as follows. 
In \cref{sec:prelims}, we introduce the remaining notation not covered in this section. 
Next, in \cref{sec:fed}, we present an algorithm for computing border-to-border distances in forest alignment graphs within \APSP time.
With this algorithm in hand, we extend the caterpillar algorithm to the \SED problem in \cref{sec:sed}, where we shift from a path-based approach to more formal and concise notation, offering a clearer treatment of the problem. 
Both \cref{sec:fed} and \cref{sec:sed} are structured to first discuss complexity in terms of the larger of the two input forests which is optimal when forests are similar in size. 
In the second part of these sections, we provide a more detailed analysis suited for unbalanced cases by parameterizing the complexity by both forest sizes.
In these sections, we provide algorithms for both weighted and unweighted \TED, in particular noting when the unweighted \TED problem can be computed more efficiently.
In \cref{sec:ted}, we show how to solve \TED using \SED, providing a generic reduction for both weighted and unweighted \TED, thus completing the proofs in this paper.

\section{Preliminaries}
\label{sec:prelims}

\subparagraph{Set notation.}

For integers \(i, j \in \mathbb{Z}\), we use the notation \(\fragment{i}{j}\) to represent the set \(\{i, \dots, j\}\), and \(\fragmentco{i}{j}\) to denote the set \(\{i, \dots, j - 1\}\).
We define \(\fragmentoc{i}{j}\) and \(\fragmentoo{i}{j}\) similarly.

Consider the infinite grid $\mathbb{Z}\times\mathbb{Z}$,
and consider a subset of the form of a rectangle $\fragment{a}{b} \times \fragment{a'}{b'}$
for some $a,b,a',b'$.
We define $\sB^{\bot}(a,a',b,b')$ and $\sB^{\top}(a,a',b,b')$ to be
set of grid points located at the bottom-left
and upper-right border of such rectangle $\fragment{a}{b} \times \fragment{a'}{b'}$.
Put more formally:

\begin{definition}
    For integers $a,a',b,b'$ we define
    \begin{align*}
        \sB^{\top}(a,a',b,b') &\coloneqq (\{b\} \times \fragment{a'}{b'}) \cup (\fragment{a}{b} \times \{b'\}),
        \text{ and } \\
        \sB^{\bot}(a,a',b,b') &\coloneqq (\{a\} \times \fragment{a'}{b'}) \cup (\fragment{a}{b} \times \{a'\}). \qedhere
    \end{align*}
\end{definition}

\subparagraph{Notation on ordered trees.}
For consistency, we adopt most of the notations for \(\TED\) from \cite{M22}.

In \TED we work with ordered trees.
For a tree \(\bT\), we denote its root as \(\Root(\bT)\).
We also treat a forest \(\bF\) as ordered, meaning that the sequence of trees within the forest is relevant.
In this context, we can view a forest \(\bF\) as a tree with a \emph{virtual root}, whose children are the roots of the trees in \(\bF\),
ordered from left to right as they appear in the forest.

For a tree \(\bT \in \bF\), let \(\bF \setminus \bT\) represent the forest obtained by removing \(\bT\) from \(\bF\)
while preserving the order of the remaining trees.
We use \(\abs{\bF}\) to denote the number of nodes in \(\bF\).
Occasionally, we may also abuse notation by using \(\bF\) to refer to the set of nodes in the forest,
though the context will make this clear.
For two forests \(\bF\) and \(\bF'\),
we denote their concatenation (from left to right) as \(\bF + \bF'\).
An empty forest is denoted by \(\emptyset\).
For a node \(v \in \bF\), we let \(\sub(v)\) represent the subtree rooted at \(v\).

Finally, we write $\rev(\bF)$ for the \emph{reversed forests of $\bF$},
obtained by taking each node of $\bF$ (including its virtual root), and reversing the order of the children.

\subparagraph{Bi-order traversal.}
To compute the similarity between forests,
we use the notation introduced by \cite{M22},
which allows us to conveniently index subforests of \(\bF\) and express calculations in a way that is well-suited to max-plus products.

\begin{definition}[{\cite[Definition~2.3]{M22}}]
Consider the depth-first traversal of a forest $\bF$ starting from the virtual root,
with subtrees recursively traversed from left to right.
From that we can generate an \emph{bi-order traversal sequence} of length $2\abs{\bF}$,
where each node appears twice, in the following way:
\begin{itemize}
    \item Start from the empty sequence.
    \item Every time we enter or leave a node,
    we attach the node to the end of the sequence (do not attach the virtual root).
\end{itemize}
We use $\bF\position{i}$ to denote the $i$-th node in such sequence.
\end{definition}

Observe that bi-order traversal of $\rev(\bF)$ corresponds to the reversed bi-order traversal of $\bF$.

\begin{definition}[{\cite[Definition~2.4]{M22}}]
For $1 \le \ell \le r \le 2\abs{\bF} + 1$, consider bi-order traversal $\bF\position{1}, \bF\position{2}, \ldots, \bF\position{2\abs{\bF}}$ of $\bF$. We use $\bF\fragmentco{\ell}{r}$
to denote the forest obtained by removing from $\bF$ all nodes that appear at least once in
$\bF\position{1}, \bF\position{2}, \ldots, \bF\position{\ell - 1}$ or $\bF\position{r}, \bF\position{r + 1}, \ldots, \bF\position{2\abs{\bF}}$,
and we call such forest a \emph{subforest of $\bF$}.
\end{definition}

Consider the bi-order traversal sequence of a forest $\bF$.
For a node \(v \in \bF\), we define \(\Left(v)\) as the first index where \(v\) appears in the bi-order traversal sequence,
and \(\Right(v)\) as one plus the second index where \(v\) appears.
By these definitions, \(\bF\fragmentco{\ell}{r}\) contains \(v\) if and only if \(\ell \leq \Left(v)\) and \(\Right(v) \leq r\).
Additionally, \(\bF\fragmentco{\Left(v)}{\Right(v)}\) corresponds to \(\sub(v)\).
Note that \(\bF\fragmentco{\ell}{r}\) and \(\bF\fragmentco{\ell'}{r'}\) may represent the same forest for distinct pairs \((\ell, r)\) and \((\ell', r')\).

\begin{definition}[{\cite[Definition~2.5]{M22}}] \label{def:syn}
    For a forest $\bF$, we say subforest $\bF'$ of $\bF$ is a \emph{synchronous subforest of $\bF$} if there exists a node $v \in \bF$
    that is either a node in $\bF$ or the virtual root of $\bF$,
    and such that $\bF'$ is the union of subtrees of subsequent children of $v$.
\end{definition}

\subparagraph{Anchors and anchor sets.} 
A key point in working with a similarity of the form $\similarity(\bF\fragmentco{x}{y}, \bF'\fragmentco{x'}{y'})$ between two subforests $\bF\fragmentco{x}{y}$ and $\bF'\fragmentco{x'}{y'}$ in bi-order traversal is realizing that the following inequality holds for all $(z, z') \in \fragment{x}{y} \times \fragment{x'}{y'}$:
\[
    \similarity(\bF\fragmentco{x}{y}, \bF'\fragmentco{x'}{y'}) \geq \similarity(\bF\fragmentco{x}{z}, \bF'\fragmentco{x'}{z'}) +  \similarity(\bF\fragmentco{z}{y}, \bF'\fragmentco{z'}{y'}).
\]

An \emph{anchor}, is a pair $(z,z')$ where such last inequality holds with equality.

\begin{definition}[Anchors and anchor sets]
Let $\bF,\bF'$ be two forests, and let $\bF\fragmentco{x}{y},\bF'\fragmentco{x'}{y'}$ be two subforests.
We say that \emph{$(z,z') \in \fragment{1}{(2|\bF|+1)} \times \fragment{1}{(2|\bF'|+1)}$ is an anchor of $\similarity(\bF\fragmentco{x}{y}, \bF'\fragmentco{x'}{y'})$} if $x \leq z \leq y$, and $x' \leq z' \leq y'$ and
\[
    \similarity(\bF\fragmentco{x}{y}, \bF'\fragmentco{x'}{y'}) = \similarity(\bF\fragmentco{x}{z}, \bF'\fragmentco{x}{z'})
                        + \similarity(\bF\fragmentco{z}{y}, \bF'\fragmentco{z'}{y}).
\]

Further, we say \emph{$\sB \subseteq \fragment{1}{(2|\bF|+1)} \times \fragment{1}{(2|\bF'|+1)}$ is an anchor set of $\similarity(\bF\fragmentco{x}{y}, \bF'\fragmentco{x'}{y'})$} if there exists $(z,z') \in \sB$ such that $(z,z')$ is an anchor of $\similarity(\bF\fragmentco{x}{y}, \bF'\fragmentco{x'}{y'})$. 
\end{definition}

We note that we use the concept of anchor sets, defined as subsets of \(\fragment{1}{(2|\bF|+1)} \times \fragment{1}{(2|\bF'|+1)}\), not only for similarities of the form \(\similarity(\bF\fragmentco{x}{y}, \bF'\fragmentco{x'}{y'})\) but also for other subforests where the same bi-order indexing applies, such as \(\similarity(\bF\fragmentco{x}{z} + \bF\fragmentco{z}{y}, \bF'\fragmentco{x'}{y'})\). 

We extend the concept of anchors to \emph{paired anchors}.

\begin{definition}[Paired anchors and anchor sets]
    Let $\bF,\bF'$ be two forests, and let $\bF\fragmentco{x}{y},\bF'\fragmentco{x'}{y'}$ be two subforests.
    We say that \emph{$(z,z'),(w,w') \in \fragment{1}{(2|\bF|+1)} \times \fragment{1}{(2|\bF'|+1)}$ are paired anchors of $\similarity(\bF\fragmentco{x}{y}, \bF'\fragmentco{x'}{y'})$} if $x \leq z \leq w \leq y$, and $x' \leq z' \leq w' \leq y'$ and
    \begin{align*}
    \similarity(\bF\fragmentco{x}{y}, \bF'\fragmentco{x'}{y'}) &= \similarity(\bF\fragmentco{x}{z}, \bF'\fragmentco{x}{z'}) \\
                        &+ \similarity(\bF\fragmentco{z}{w}, \bF'\fragmentco{z'}{w'})
                        + \similarity(\bF\fragmentco{w}{y}, \bF'\fragmentco{w'}{y'}).
    \end{align*}
    
    Further, we say \emph{$\sB, \sB' \subseteq \fragment{1}{(2|\bF|+1)} \times \fragment{1}{(2|\bF'|+1)}$ are paired anchor sets of $\similarity(\bF\fragmentco{x}{y}, \bF'\fragmentco{x'}{y'})$} if there are $(z,z) \in \sB$ and $(w,w') \in \sB'$ such that $(z,z')$ and $(w,w')$ are paired anchors.
\end{definition}

Observe that a set $\sB$ is an anchor set of $\similarity(\bF\fragmentco{x}{y}, \bF'\fragmentco{x'}{y'})$ if and only if
\begin{equation*}
    \similarity(\bF\fragmentco{x}{y}, \bF'\fragmentco{x'}{y'}) = \max_{(z,z') \in \sB \mid x \leq z \leq y, x' \leq z' \leq y'} \Big\{ \ \similarity(\bF\fragmentco{x}{z}, \bF'\fragmentco{x}{z'})
                        \ + \ \similarity(\bF\fragmentco{z}{y}, \bF'\fragmentco{z'}{y'}) \ \Big\}.
\end{equation*}
Moreover, for sets $\sB$ and $\sB'$ are paired anchor sets holds
\begin{align*}
    \similarity(\bF\fragmentco{x}{y}, \bF'\fragmentco{x'}{y'}) = \max\nolimits_{\substack{(z,z') \in \sB \mid x \leq z \leq y, x' \leq z' \leq y' \\ (w,w') \in \sB' \mid z \leq w \leq y, z' \leq y' \leq y' }} 
    & \Big\{ \ \similarity(\bF\fragmentco{x}{z}, \bF'\fragmentco{x'}{z'}) \\
                        & + \ \similarity(\bF\fragmentco{z}{w}, \bF'\fragmentco{z'}{w'}) \\
                        & + \similarity(\bF\fragmentco{w}{y}, \bF'\fragmentco{w'}{y'}) \ \Big\}.
\end{align*}

In previous works, such as {\cite[Equation~(9)]{M22}},
it was observed that when \(\bF\) consists of more than one tree,
there exists an anchor set with a simple representation.
More specifically, if we can write $\bF = \bF_{\Left} + \bF_{\Right}$ for two forests $\bF_{\Left}, \bF_{\Right}$, then $(2|\bF_{\Left}|+1) \times \fragment{1}{(2|\bF'|+1)}$
is an anchor set  of $\similarity(\bF, \bF')$.
We rephrase this observation for subforests as follows.

\begin{proposition} \label{rmk:mao}
    Suppose that for two subforests $\bF\fragmentco{x}{y}, \bF'\fragmentco{x'}{y'}$
    and some $z \in \fragment{x}{y}$,
    we have that $\similarity(\bF\fragmentco{x}{y}, \bF'\fragmentco{x'}{y'}) = \similarity(\bF\fragmentco{x}{z} + \bF\fragmentco{z}{y} , \bF'\fragmentco{x'}{y'})$ holds.
    Then, $z \times \fragment{1}{(2\abs{\bF'}+1)}$ is an anchor set of $\similarity(\bF\fragmentco{x}{y}, \bF'\fragmentco{x'}{y'})$.
    \lipicsEnd
\end{proposition}

Next, we prove \cref{rmk:anchor_transform} that allows us to transform anchor sets.

\begin{proposition} \label{rmk:anchor_transform}
    Let $(z,z')$ be an anchor of $\similarity(\bF\fragmentco{x}{y}, \bF'\fragmentco{x'}{y'})$.
    Suppose $\sB$ is an anchor set of $\similarity(\bF\fragmentco{z}{y}, \bF'\fragmentco{z'}{y'})$.
    Then, the two following hold:
    \begin{enumerate}[(i)]
        \item For any anchor set $\sA$ such that $(z,z') \in \sA$, we have that $\sA$ and $\sB$ are paired anchor set.
        \label{it:anchor_transform:i}
        \item $\sB$ is also anchor set of $\similarity(\bF\fragmentco{x}{y}, \bF'\fragmentco{x'}{y'})$.
        \label{it:anchor_transform:ii}
    \end{enumerate}
\end{proposition}

\begin{proof}
    From the assumptions on $\sB$, there exists $(w,w') \in \sB$
    such that $\similarity(\bF\fragmentco{z}{y}, \bF'\fragmentco{z'}{y'}) = \similarity(\bF\fragmentco{z}{w}, \bF'\fragmentco{z'}{w'}) + \similarity(\bF\fragmentco{w}{y}, \bF'\fragmentco{w'}{y'})$.
    Combining this with the anchor assumption on $(z, z')$, we have 
    \begin{align*} 
    \MoveEqLeft \similarity(\bF\fragmentco{x}{y}, \bF'\fragmentco{x'}{y'}) \\
    &= \similarity(\bF\fragmentco{x}{z}, \bF'\fragmentco{x'}{z'}) + \similarity(\bF\fragmentco{z}{y}, \bF'\fragmentco{z'}{y'})\\ 
    &= \similarity(\bF\fragmentco{x}{z}, \bF'\fragmentco{x'}{z'}) + \similarity(\bF\fragmentco{z}{w}, \bF'\fragmentco{z'}{w'}) + \similarity(\bF\fragmentco{w}{y}, \bF'\fragmentco{w'}{y'}). 
    \end{align*}
    This shows \eqref{it:anchor_transform:i}.
    To prove \eqref{it:anchor_transform:ii}, it suffices to show that
    $\similarity(\bF\fragmentco{x}{w}, \bF'\fragmentco{x'}{w'}) = \similarity(\bF\fragmentco{x}{z}, \bF'\fragmentco{x'}{z'}) + \similarity(\bF\fragmentco{z}{w}, \bF'\fragmentco{z'}{w'})$,
    which would imply that $(w,w')$ is an anchor of $\similarity(\bF\fragmentco{x}{y}, \bF'\fragmentco{x'}{y'})$.
    To this end, observe that
    $\similarity(\bF\fragmentco{x}{w}, \bF'\fragmentco{x'}{w'}) > \similarity(\bF\fragmentco{x}{z}, \bF'\fragmentco{x'}{z'}) + \similarity(\bF\fragmentco{z}{w}, \bF'\fragmentco{z'}{w'})$
    would imply
    \begin{align*}
        \MoveEqLeft \similarity(\bF\fragmentco{x}{y}, \bF'\fragmentco{x'}{y'}) \\
        &= \similarity(\bF\fragmentco{x}{z}, \bF'\fragmentco{x'}{z'})
        + \similarity(\bF\fragmentco{z}{w}, \bF'\fragmentco{z'}{w'})
        + \similarity(\bF\fragmentco{w}{y}, \bF'\fragmentco{w'}{y'})\\
        &< \similarity(\bF\fragmentco{x}{w}, \bF'\fragmentco{x'}{w'}) + \similarity(\bF\fragmentco{w}{y}, \bF'\fragmentco{w'}{y'}) \\
        &\leq \similarity(\bF\fragmentco{x}{y}, \bF'\fragmentco{x'}{y'}),
    \end{align*}
    a contradiction.
\end{proof}

Using \cref{rmk:anchor_transform} we can transform the anchor set from \cref{rmk:mao} as follows.

\begin{proposition} \label{rmk:anchor_corner}
    Consider two subforests $\bF\fragmentco{x}{y}, \bF'\fragmentco{x'}{y'}$,
    and suppose there exist $v \in \bF\fragmentco{x}{y}$ and $v' \in \bF'\fragmentco{x'}{y'}$
    such that 
    \[
        \similarity(\bF\fragmentco{x}{y}, \bF'\fragmentco{x'}{y'}) = \similarity(\bF\fragmentco{x}{\Left(v)} + \bF\fragmentco{\Left(v)}{y}, \bF'\fragmentco{x'}{\Left(v')}) + \bF'\fragmentco{\Left(v')}{y'}). 
    \]
    Then, we have that 
    \[
        \sB^{\top}(1, 1, \Left(v), \Left(v')) \ = \ (\Left(v) \times \fragment{1}{\Left(v')}) \ \cup \  (\fragment{1}{\Left(v)} \times \Left(v'))
    \]
    is an anchor set of $\similarity(\bF\fragmentco{x}{y}, \bF'\fragmentco{x'}{y'})$.
\end{proposition}

\begin{proof}
    By \cref{rmk:mao}, the set \(\sB \coloneqq \Left(v) \times \fragment{1}{(2\abs{\bF'}+1)}\) serves as an anchor set for \(\similarity(\bF\fragmentco{x}{\Left(v)} + \bF\fragmentco{\Left(v)}{y}, \bF'\fragmentco{x'}{\Left(v')} + \bF'\fragmentco{\Left(v')}{y'})\). 
    
    Suppose there exists an anchor \((z, z')\) contained in the subset \(\sB' \coloneqq \Left(v) \times \fragment{\Left(v')}{(2\abs{\bF'}+1)} \subseteq \sB\). For any such \((z, z')\), note that, by \cref{rmk:mao}, the set \(\sB'' \coloneqq \fragment{1}{\Left(v)} \times \Left(v')\) is an anchor set for \(\similarity(\bF\fragmentco{x}{z}, \bF'\fragmentco{x'}{\Left(v')} + \bF'\fragmentco{\Left(v')}{z'})\).
    By \cref{rmk:anchor_transform}\eqref{it:anchor_transform:ii}, we get that \(\sB''\) is also an anchor set for \(\similarity(\bF\fragmentco{x}{\Left(v)} + \bF\fragmentco{\Left(v)}{y}, \bF'\fragmentco{x'}{\Left(v')} + \bF'\fragmentco{\Left(v')}{y'})\). 
    
    Now, we perform a case distinction on whether there is an anchor $(z,z') \in \sB$ of $\similarity(\bF\fragmentco{x}{\Left(v)} + \bF\fragmentco{\Left(v)}{y}, \bF'\fragmentco{x'}{\Left(v')}) + \bF'\fragmentco{\Left(v')}{y'})$ contained in the subset $\sB'$ or not. In any case, $(\sB \setminus \sB') \cup \sB'' = (\Left(v) \times \fragment{1}{\Left(v')}) \cup (\fragment{1}{\Left(v)} \times \Left(v'))$ is an anchor set of $\similarity(\bF\fragmentco{x}{\Left(v)} + \bF\fragmentco{\Left(v)}{y}, \bF'\fragmentco{x'}{\Left(v')}) + \bF'\fragmentco{\Left(v')}{y'})$.
    By the assumption, $(\sB \setminus \sB') \cup \sB''$
    is also an anchor set of $\similarity(\bF\fragmentco{x}{y}, \bF'\fragmentco{x'}{y'})$.
\end{proof}

For certain of forms of similiarities 
we can also find paired anchor sets.

\begin{proposition}  \label{rmk:paired_anchors}
    Let $\bF,\bF'$ be two forests, and let $\bF\fragmentco{x}{y},\bF'\fragmentco{x'}{y'}$ be two subforests. Then, the two following hold:
    \begin{enumerate}[(i)]
        \item For any $v \in \bF\fragmentco{x}{y}$
    such that $\similarity(\bF\fragmentco{x}{y},\bF'\fragmentco{x'}{y'}) = \similarity(\bF\fragmentco{x}{\Left(v)} + \bF\fragmentco{\Left(v)}{y}, \bF'\fragmentco{x'}{y'})$, we have that $\Left(v) \times \fragment{1}{(2|\bF'|+1)}$ and $\Right(v) \times \fragment{1}{(2|\bF'|+1)}$ are paired anchors sets of $\similarity(\bF\fragmentco{x}{y}, \bF'\fragmentco{x'}{y'})$.
        \label{it:paired_anchors:a}
        \item For any $v \in \bF\fragmentco{x}{y}$ and $v' \in \bF'\fragmentco{x'}{y'}$ such that 
       $\similarity(\bF\fragmentco{x}{y}, \bF'\fragmentco{x'}{y'}) = \similarity(\bF\fragmentco{x}{\Left(v)} + \bF\fragmentco{\Left(v)}{y}, \bF'\fragmentco{x'}{\Left(v')}) + \bF'\fragmentco{\Left(v')}{y'})$, we have that
        $\sB^{\top}(1, 1, \Left(v), \Left(v'))$ and $\sB^{\bot}(\Right(v), \Right(v'), 2\abs{\bF}+1, 2\abs{\bF'}+1)$ are paired anchors sets of $\similarity(\bF\fragmentco{x}{y}, \bF'\fragmentco{x'}{y'})$.
        \label{it:paired_anchors:b}
    \end{enumerate}
\end{proposition}

\begin{proof}
    We first prove \eqref{it:paired_anchors:a}.
    By \cref{rmk:mao},
    we have that $\Left(v) \times \fragment{1}{(2|\bF'|+1)}$ is an anchor set of $\similarity(\bF\fragmentco{x}{y},\bF'\fragmentco{x'}{y'})$.
    Observe that $\bF\fragmentco{\Left(v)}{y} = \sub(v) + \bF\fragmentco{\Right(v)}{y}$.
    Thus, by \cref{rmk:mao}, for every $(z,z') \in \Left(v) \times \fragment{1}{(2|\bF'|+1)}$, we have that $\Right(v) \times \fragment{1}{(2|\bF'|+1)}$ is an anchor set of 
    $\similarity(\bF\fragmentco{z}{y},\bF'\fragmentco{z'}{y'})$.
    To conclude the proof of \eqref{it:paired_anchors:a} we use \cref{rmk:anchor_transform}\eqref{it:anchor_transform:i}. 

    We prove \eqref{it:paired_anchors:b} in a similar manner. By \cref{rmk:anchor_corner}, the set \(\sB^{\top}(1, 1, \Left(v), \Left(v'))\) serves as an anchor set for \(\similarity(\bF\fragmentco{x}{y}, \bF'\fragmentco{x'}{y'})\). 
    For each \((z, z') \in \sB^{\top}(1, 1, \Left(v), \Left(v'))\), we have \(\bF\fragmentco{z}{y} = \bF\fragmentco{z}{\Right(v)} + \bF\fragmentco{\Right(v)}{y}\) and \(\bF\fragmentco{z'}{y'} = \bF\fragmentco{z'}{\Right(v')} + \bF\fragmentco{\Right(v')}{y'}\). 
    Thus, we can once again apply \cref{rmk:anchor_corner}, using a symmetric equivalent that uses \(\Right(v)\) instead of \(\Left(v)\), and subsequently apply \cref{rmk:anchor_transform}\eqref{it:anchor_transform:i}. 
\end{proof}

\subparagraph{Aligned Subtrees.} 
Next, we introduced what means for two subtrees to be \emph{aligned} by a similarity.

\begin{definition}
    Given two forests $\bF,\bF'$ and nodes $v \in \bF$, $v' \in \bF'$, we say $\similarity(\bF, \bF')$ aligns $\sub(v)$ to $\sub(v')$ if 
    \begin{align*}
        \similarity(\bF, \bF') &= \similarity(\bF\fragmentco{1}{\Left(v)}, \bF'\fragmentco{1}{\Left(v')}) + \similarity(\sub(v), \sub(v'))\\
        &+ \similarity(\bF\fragmentco{\Right(v)}{(2\abs{\bF}+1)}, \bF'\fragmentco{\Right(v')}{(2\abs{\bF'}+1)}).
        \qedhere
    \end{align*}
\end{definition}

In other words, $\similarity(\bF, \bF')$ aligns $\sub(v)$ to $\sub(v')$, if $\{(\Left(v),\Left(v))\}, \{(\Right(v),\Right(v'))\}$ are paired anchor sets.

\begin{proposition}\label{prp:align_or_remove}
    Given two forests $\bF,\bF'$ and a node $v \in \bF$, let $\bP$ be the set of nodes contained in the path going from the virtual root of $\bF$ to $v$ (virtual root and $v$ excluded). 
    Then, for each pair of subforests $\bF\fragmentco{x}{y}$ and $\bF'\fragmentco{x'}{y'}$ such that 
    $\sub(v) \subseteq \bF\fragmentco{x}{y}$ at least one of the two following cases holds:
    \begin{enumerate}[(a)]
        \item There exists a node $u \in \bP$ such that $\similarity(\bF\fragmentco{x}{y}, \bF'\fragmentco{x'}{y'})$ aligns $\sub(u)$ to a subtree of $\bF'$.
        \label{prp:align_or_remove:a}
        \item $\similarity(\bF\fragmentco{x}{y}, \bF'\fragmentco{x'}{y'}) = \similarity(\bF\fragmentco{x}{\Left(v)} + \bF\fragmentco{\Left(v)}{y}, \bF')$.
        \label{prp:align_or_remove:b}
    \end{enumerate}
\end{proposition}

\begin{proof}
    We show that whenever \eqref{prp:align_or_remove:a} does not hold, then \eqref{prp:align_or_remove:b} does.
    To this end, note that $\bF\fragmentco{x}{\Left(v)} + \bF\fragmentco{\Left(v)}{y}$ correponds to $\bF\fragmentco{x}{y}$
    with all the nodes from $\bP$ taken out. Moreover, the pre-order of $\bF\fragmentco{x}{\Left(v)} + \bF\fragmentco{\Left(v)}{y}$ equals to the pre-order of $\bF\fragmentco{x}{y}$ with the nodes of $\bP$ taken out.

    Next, let us compare the computations done by Shasha and Zhang's recurrence scheme \cite{SZ89} for
    the two similarities $ \similarity(\bF\fragmentco{x}{y}, \bF'\fragmentco{x'}{y'}) $ and $\similarity(\bF\fragmentco{x}{\Left(v)} + \bF\fragmentco{\Left(v)}{y}, \bF')$.
    When computing $\similarity(\bF\fragmentco{x}{y}, \bF'\fragmentco{x'}{y'})$ note that each time the recurrence scheme arrives to consider as next node of $\bF\fragmentco{x}{y}$ some $u \in \bP$, we can ignore the case where $\sub(u)$ is aligned to the subtree of the current node in $\bF'\fragmentco{x'}{y'}$. As a consequence, we can assume that $u$ is always deleted. We conclude that we obtain the same computations as in the recurrence scheme for $\similarity(\bF\fragmentco{x}{\Left(v)} + \bF\fragmentco{\Left(v)}{y}, \bF')$.
\end{proof}

\subparagraph{Spines.}
Let $\bF$ be a forest.
We define a \emph{spine} as a path that starts at the root of one of the trees in $\bF$ and ends in a leaf within that tree.

Now, consider a spine $\bS \subseteq \bF$.
For $s,q \in \bS$, we write $s \prec q$ if $s$ appears before $q$ in the root-to-leaf path, and we say $s$ \emph{precedes} $q$.
Furthermore, we say that $s$ \emph{immediately precedes} $q$ if there is no $r \in \bS$ such that $s \prec r \prec q$.

Lastly, suppose a spine $\bS \subseteq \bF$ is a path of the form $s_1, \ldots, s_{m}$.
We observe that for all $i \in \fragmentco{1}{m}$, a bi-order traversal enters $s_i$ before it enters $s_{i+1}$,
and leaves $s_{i+1}$ before it leaves $s_{i}$. Consequently,
\[
    \Left(s_1) < \Left(s_2) < \cdots < \Left(s_m) < \Right(s_m) < \cdots < \Right(s_2) < \Right(s_1).
\]
Moreover, $\Left(s_m) = \Right(s_m) - 2$,
since $s_m$ is a leaf and the bi-order traversal leaves $s_m$ right after it enters it.

\subsection{Min-Plus Products and Structured Instances}

While Min-plus Matrix Multiplication (\MUL) of $n \times n$ matrices is commonly conjectured to require $n^{3 - o(1)}$ time, there have been several structured instances in which faster algorithms have been obtained, e.g. bounded entries \cite{AGM97}, (row/column) bounded-difference matrices \cite{BGSV19}, and (row/column) monotone matrices \cite{DBLP:conf/soda/WilliamsX20, DBLP:conf/icalp/Gu0WX21, CDXZ22} .
Specifically, we focus on the class of monotone matrices.

\begin{definition}
    An $n \times n$ matrix is row-monotone if all entries are non-negative integers bounded by $\Oh(n)$ and each row of this matrix is non-decreasing.
    Similarly, we may define a column-monotone matrix.
\end{definition}

When at least one of the matrices is row-monotone or column-monotone, there are sub-cubic algorithms for min-plus matrix multiplication \cite{DBLP:conf/soda/WilliamsX20, DBLP:conf/icalp/Gu0WX21, CDXZ22}.

\begin{theorem}[\cite{CDXZ22}]
    There is a (randomized) algorithm that computes the min-plus product $A * B$ in expected running time $\Ohtilde(n^{(3 + \omega)/2})$, where $A$ is an $n \times n$ integer matrix and $B$ is a $n \times n$ row-monotone matrix.
    The result holds also when $B$ is a $n \times n$ column-monotone matrix. \lipicsEnd
\end{theorem}

The above result holds also when the rows or columns of the matrix are non-increasing.

We let $\sT_{\MonMUL}(n)$ denote the time required to multiply an $n \times n$ integer matrix with a $n \times n$ monotone matrix.
Furthermore, let $\sT_{\MonMUL}(a, b, c, d)$ be the time required to multiply an $a \times b$ integer matrix with a $b \times c$ monotone matrix with all entries  bounded by $d$. 

In our work, we will specifically work with rectangular matrices, and thus require the following simple lemma. 
We use $\sT_{\MUL}(a, b, c)$ to denote the time required to compute the $(\min, +)$ product between arbitrary $a \times b$ integer matrices and $b \times c$ integer matrices.

\begin{lemma}
    Suppose $n \geq m$.
    Then,
    \begin{align*}
        &\Oh(\sT_{\MUL}(n, m, m)) = \Oh(\sT_{\MUL}(m, n, m)) = \Oh(\sT_{\MUL}(m, m, n)) = \Oh(n/m \cdot \sT_{\MUL}(m)) \quad \text{and} \quad \\
        &\Oh(\sT_{\MonMUL}(n, m, m, m)) = \Oh(\sT_{\MonMUL}(m, n, m, m)) = \Oh(\sT_{\MonMUL}(m, m, n, m))= \Oh(n/m \cdot \sT_{\MonMUL}(m)).
    \end{align*}
    \lipicsEnd
\end{lemma}
We require the following bound on $(\min, +)$-product between monotone matrices where the entries are bounded by $\D$.

\begin{lemma}[\cite{Durr23}]
    $\sT_{\MonMUL}(n, n, n, \D) = \Ohtilde(\sqrt{\D} n^{(2 + \omega)/2})$.
     \lipicsEnd
\end{lemma}

\section{Reduction from Forest Edit Distance to \APSP}
\label{sec:fed}

In this section, we focus on proving \cref{thm:fed} which we restate here and its unweighted version \Cref{thm:unweighted-fed-balance}.

\fed*

\begin{theorem}
    \label{thm:unweighted-fed-balance}
    There is an $\Oh(\sT_{\MonMUL}(n) + n^{2 + o(1)} g(n))$ time algorithm for unweighted \FED where $n = \max(|\bF|,|\bF'|)$ and $\sT_{\MonMUL}(n, n, n, \D) = \Oh(f(n) g(\D))$ for some functions $f, g$. \lipicsEnd
\end{theorem}

To approach \FED, we adopt a conceptual shift, viewing it as a graph problem.
Our perspective extends the computation of border-to-border distances in the alignment graph for string edit distance~\cite{LMS98, T06, ACS08, K05}.
These distances precisely capture the edit distance between the prefix of one string and the suffix of the other,
as well as between the entire string and all infixes of the second.

In this section, we denote by $\dist_{\bG}(u, v)$ the longest distance (rather than shortest distance)
from $u \in \bG$ to $v \in \bG$ in a directed weighted acyclic graph $\bG$.
If $v$ is not reachable from $u$, then $\dist_{\bG}(u, v) = - \infty$.
This shift in definition is merely a notational change,
given that we are working with a directed acyclic graph (by simply inverting the sign of the weights, one can switch between shortest and longest distances).

Throughout this section, we will give algorithms for the general setting of weighted tree edit distance, adding where appropriate optimizations that can be made for the unweighted tree edit distance problem.
In \Cref{sec:fed-unbalance}, we describe how to compute \FED efficiently on unbalanced instances, i.e.\ when one forest is significantly larger than the other. 
While \Cref{sec:fed-unbalance} is necessary to obtain \Cref{thm:unweighted-ted} for unweighted \TED, readers interested only in \Cref{thm:fed} may skip it.

\subsection{Generalizing alignment graphs}

In this (sub)section, we extend the concept of the alignment graph for two strings
to a broader framework capable of representing \FED.
Similar to the string case,
the alignment graph we examine has a grid as its vertex set.
This grid includes directed edges of zero weight,
connnecting every grid point to its upper and right neighbour.
Moreover, from each grid point \((i,j)\), there is an additional edge connecting to a subgrid, whose lower-left corner is \((i+1,j+1)\) and whose top-right corner coincides with that of the entire grid.

\begin{definition}\label{def:map_pi}
    Given a forest $\bF$ with pre-order $v_1, \ldots, v_{|\bF|}$, we define $\pi_{\bF} : \fragment{1}{|\bF|} \rightarrow \fragment{1}{|\bF|+1}$
    as $\pi_{\bF}(i) = \min \{j : j \geq i \text{ and } v_j \notin \sub(v_i)\}$. If no such $j$ exists, then $\pi_{\bF}(i) = |\bF|+1$.
\end{definition}

In other words, $\pi_{\bF}(i)$ maps each node $v_i$ to the next node that pre-order visits after leaving the subtree $\sub(v_i)$ (see \cref{fig:align_graph} for an example).

\begin{definition}[See also \cite{T05, BCHMRWZ07, MTWZ09}]\label{def:alg_graph}
Given two forests $\bF$, $\bF'$ with pre-orders $v_1, \ldots, v_{\abs{\bF}}$ and $v_1', \ldots, v_{\abs{\bF'}}'$,
and given a weight function $\w : D \rightarrow \mathbb{R}$ such that $\bF \times \bF' \subseteq D$,
we define the \emph{alignment graph $\bar{\bG} = (\bar{V}, \bar{E})$ w.r.t.~$(\w, \bF, \bF')$} as the directed weighted acyclic graph with
\begin{itemize}
    \item Vertex set $\bar{V} = \fragment{1}{(\abs{\bF}+1)} \times \fragment{1}{(\abs{\bF'}+1)}$;
    \item Edge set $\bar{E}$ such that $\bar{e} = ((u, u'), (w, w')) \in \bar{E}$, if exactly one of the following holds:
    \begin{enumerate}[(a)]
        \item $(w, w') = (u, u'+1)$,
        \label{it:delFp}
        \item $(w, w') = (u+1, u')$, or
        \label{it:delF}
        \item $(w, w') = (\pi_{\bF}(u), \pi_{\bF'}(u'))$;
        \label{it:match}
    \end{enumerate}
    \item Weight function $\bar{\w}$ such that
    $\bar{\w}(\bar{e}) = 0$ if $\bar{e}$ has the form~\eqref{it:delFp} or~\eqref{it:delF}, and
    $\bar{\w}(\bar{e}) = \similarity(\sub(v_u), \sub(v_{u'}'))$ if it has the form~\eqref{it:match}.
    \qedhere
\end{itemize}
\end{definition}

\begin{figure}[htbp]
    \centering
    \usetikzlibrary{matrix}

\def\fx{{5, 2, 5, 4, 5}}
\def\gy{{3, 2, 3, 5, 5}}

\def\lx{{"2","3-4"}}
\def\ly{{"2","3","4-5"}}

\begin{tikzpicture}[>=stealth, scale=0.8]

     \node at (2.5, 6) {$\bar{\bG}$};
    
    \foreach \x in {0,...,5}
        \foreach \y in {0,...,5} {
            \node[circle, draw, fill, scale=0.5] (\x\y) at (\x,\y) {};
            \pgfmathtruncatemacro{\lbx}{\x + 1}
            \pgfmathtruncatemacro{\lby}{\y + 1}

            \ifnum\y=0
                \node[below=10pt] at (\x\y) {\lbx};
            \fi
            \ifnum\x=0
                \node[left=10pt] at (\x\y) {\lby};
            \fi
        }

    \foreach \x in {0,...,4}
        \foreach \y in {0,...,5} {
            \pgfmathtruncatemacro{\nextx}{\x + 1}
            \draw[->] (\x\y) -- (\nextx\y);
        }
    \foreach \x in {0,...,5}
        \foreach \y in {0,...,4} {
            \pgfmathtruncatemacro{\nexty}{\y + 1}
            \draw[->] (\x\y) -- (\x\nexty);
        }

    \foreach \x in {0,...,4}
        \foreach \y in {0,...,4} {
            \pgfmathparse{\fx[\x]}
            \let\fvalue\pgfmathresult
            \pgfmathparse{\gy[\y]}
            \let\gvalue\pgfmathresult
            \draw[->] (\x\y) -- (\fvalue\gvalue);
        }

    \draw[red, thick] (0.5,0.5) rectangle (3.5,4.5);

    \begin{scope}[shift={(7,1)}]
        \draw[red, thick] (-0.3,-0.3) rectangle (1.3,2.3);

        \foreach \x in {0,...,1}
            \foreach \y in {0,...,2} {
                \pgfmathparse{\lx[\x]}
                \let\xlabel\pgfmathresult
                \pgfmathparse{\ly[\y]}
                \let\ylabel\pgfmathresult

                \node[circle, draw, fill, scale=0.5] (s\x\y) at (\x,\y) {};
                \ifnum\y=0
                    \node[below=10pt] at (s\x\y) {\xlabel};
                \fi
                \ifnum\x=0
                    \node[left=10pt] at (s\x\y) {\ylabel};
                \fi
            }

        \node at (0.5, 3) {$\bar{\bG}\fragment{2}{4}\fragment{2}{5}$};

        \foreach \x in {0,...,0}
            \foreach \y in {0,...,2} {
                \pgfmathtruncatemacro{\nextx}{\x + 1}
                \draw[->] (s\x\y) -- (s\nextx\y);
            }
        \foreach \x in {0,...,1}
            \foreach \y in {0,...,1} {
                \pgfmathtruncatemacro{\nexty}{\y + 1}
                \draw[->] (s\x\y) -- (s\x\nexty);
            }

        \draw[->] (s00) -- (s11);
        \draw[->] (s01) -- (s12);

    \end{scope}

    \draw[red, thick] (3.5,0.5) -- (6.7,0.7);
    \draw[red, thick] (3.5,4.5) -- (6.7,3.3);

    \node at (-10.5, 5) {$\bF$};

    \node[circle, draw, scale=0.8] (v1) at (-9,5) {$v_1$};
    \node[circle, draw, scale=0.8] (v2) at (-9.5,4) {$v_2$};
    \node[circle, draw, scale=0.8] (v3) at (-8.5,4) {$v_3$};
    \node[circle, draw, scale=0.8] (v4) at (-9,3) {$v_4$};
    \node[circle, draw, scale=0.8] (v5) at (-8,3) {$v_5$};
    \draw (v1) -- (v2);
    \draw (v1) -- (v3);
    \draw (v3) -- (v4);
    \draw (v3) -- (v5);

    \node at (-4,4) {
            \begin{tabular}{c|c|c|c|c|c}
                $\pi_{\bF}$ & 1 & 2 & 3 & 4 & 5 \\
                \hline
                & 6 & 3 & 6 & 5 & 6 \\
            \end{tabular}
        };

    \node at (-10.5, 2.5) {$\bF'$};

    \node[circle, draw, scale=0.8] (v1p) at (-10,1.5) {$v_1'$};
    \node[circle, draw, scale=0.8] (v2p) at (-10.5,0.5) {$v_2'$};
    \node[circle, draw, scale=0.8] (v3p) at (-9.5,0.5) {$v_3'$};
    \node[circle, draw, scale=0.8] (v4p) at (-8,1.5) {$v_4'$};
    \node[circle, draw, scale=0.8] (v5p) at (-8.5,0.5) {$v_5'$};
    \draw (v1p) -- (v2p);
    \draw (v1p) -- (v3p);
    \draw (v4p) -- (v5p);

    \node at (-4,1) {
        \begin{tabular}{c|c|c|c|c|c}
            $\pi_{\bF'}$ & 1 & 2 & 3 & 4 & 5 \\
            \hline
            & 4 & 3 & 4 & 6 & 6 \\
        \end{tabular}
    };

\end{tikzpicture}
    \caption{The figure displays two forests $\bF, \bF'$, together with the two 
    mappings \(\pi_{\bF} : \fragment{1}{5} \rightarrow \fragment{1}{6}\) and \(\pi_{\bF'}: \fragment{1}{5} \rightarrow \fragment{1}{6}\) from \cref{def:map_pi}.
    On the right side the alignment graph $\bar{\bG}$ w.r.t.~\((\w, \bF, \bF')\) is depicted (weights are omitted) from \cref{def:alg_graph}.
    The figure also illustrates \(\bar{\bG}\fragment{2}{4}\fragment{2}{5}\) from \cref{def:alg_sub_graph}.}
    \label{fig:align_graph}
\end{figure}

\cref{fig:align_graph} illustrates an example of an alignment graph.
Notably, if $\bF$ and $\bF'$ are forests containing single node trees,
one obtains the alignment graph for string edit distance.

In the rest of \cref{sec:fed}, we slightly abuse notation,
and for a forest \(\bF\) with pre-order \(v_1, \ldots, v_{|\bF|}\), we write \(\Left(v_{|\bF|+1}) = 2|\bF|+1\).

\begin{definition}\label{def:alg_sub_graph}
    Let $\bar{\bG}$ be an alignment graph w.r.t ~$(\w, \bF, \bF')$ where
    $\bF$, $\bF'$ are two forests with pre-orders $v_1, \ldots, v_{\abs{\bF}}$ and $v_1', \ldots, v_{\abs{\bF'}}'$.
    Given intervals $\fragment{i}{j} \subseteq \fragment{1}{(\abs{\bF}+1)}, \fragment{i'}{j'} \subseteq \fragment{1}{(\abs{\bF'}+1)}$,
    we denote with $\bar{\bG}\fragment{i}{j}\fragment{i'}{j'}$ the alignment graph w.r.t.~$(\w, \bF\fragmentco{\Left(v_i)}{\Left(v_{j})}, \bF'\fragmentco{\Left(v_{i'}')}{\Left(v_{j'}')})$.
\end{definition}

Note that $\bar{\bG}\fragment{i}{j}\fragment{i'}{j'}$ does not necessarily correspond
to the subgraph of $\bar{\bG}$ induced by the vertex set $\fragment{i}{j} \times \fragment{i'}{j'}$,
as there might exist $z \in \fragmentco{i}{j}$ such that $v_z \notin \bF\fragmentco{\Left(v_i)}{\Left(v_{j})}$.
However, we can still derive the distances between any two nodes in the former graph from the distances in the latter graph, and vice versa.

To see this, observe that $\bar{\bG}\fragment{i}{j}\fragment{i'}{j'}$ can be obtained from
the subgraph of $\bar{\bG}$ induced by the vertex set $\fragment{i}{j} \times \fragment{i'}{j'}$
by performing the following steps iteratively for each $z \in \fragmentco{i}{j}$ where $v_z \notin \bF\fragmentco{\Left(v_i)}{\Left(v_{j})}$
(and a similar process for each $z'\in \fragmentco{i'}{j'}$ where $v_{z'}' \notin \bF\fragmentco{\Left(v_{i'}')}{\Left(v_{j'}')}$):
identify the vertex $(z, z')$ with vertex $(z+1,z')$ for all $z' \in \fragment{i'}{j'}$.
Throughout these iterations, $(i,j)$ denotes the vertex where the vertex $(i,j)$ of the subgraph of $\bar{\bG}$ formed by the vertex set $\fragment{i}{j} \times \fragment{i'}{j'}$ might have potentially merged.

Let $\bar{\bG}^{+z}$ be the graph before applying the step for some $z$,
and $\bar{\bG}^{-z}$ denote the graph obtained after this step.
Since $v_z \notin \bF\fragmentco{\Left(v_i)}{\Left(v_{j})}$, it follows that $\pi_{\bF}(z) > j$.
Therefore, no vertex in $z \times \fragment{i'}{j'}$ has an outgoing edge of the form \eqref{it:match} in $\bar{\bG}^{+z}$.
Moreover, $\pi_{\bF}^{-1}(z+1) = \emptyset$, as otherwise it would contradict $\pi_{\bF}(z) > j$.
Thus, no vertex in $(z+1) \times \fragment{i'}{j'}$ has an incoming edge of the form \eqref{it:match} in $\bar{\bG}^{+z}$.
It is not difficult to see that for any pair $(\bar{u}, \bar{w})$ where at most one of $\bar{u}, \bar{w}$ is contained in $\{z,z+1\} \times \fragment{i'}{j'}$,
we have $\dist_{\bar{\bG}^{+z}}(\bar{u}, \bar{w}) = \dist_{\bar{\bG}^{-z}}(\bar{u}, \bar{w})$.
Now, it is not difficult to see that if $\bar{u}, \bar{w} \in \{z,z+1\} \times \fragment{i'}{j'}$,
then the distance between $\bar{u}$ and $\bar{w}'$ in $\bar{\bG}^{+z}$ and $\bar{\bG}^{-z}$ is zero
if $\bar{w}$ is reachable from $\bar{w}$ in $\bar{\bG}^{+z}$ and $\bar{\bG}^{-z}$, respectively,
and $-\infty$ otherwise.

\subsection{\FED as border-to-border distances in an alignment graph}

Our focus lies in computing paths from the lower-left to the upper-right border of an alignment graph.

\defproblem{Border-to-border Distances (\BBD)}
{two forests $\bF$, $\bF'$ of size $\abs{\bF}=n, \abs{\bF'}=n'$, and a weight function $\w : D \rightarrow \mathbb{R}$ such that $D \subseteq \bF \times \bF'$.}
{$\dist_{\bar{\bG}}(\bar{u}, \bar{w})$ for all $\bar{u} \in \sB^{\bot}(1, 1, n+1, n'+1)$ and $\bar{w} \in \sB^{\top}(1, 1, n+1, n'+1)$,
where $\bar{\bG}$ is the alignment graph w.r.t.~$(\w, \bF, \bF')$.}

We write $(\w, \bF, \bF')$-\BBD for the \BBD instance with input $\bF, \bF'$ and $\w$.
Moreover, given a $(\w, \bF, \bF')$-\BBD instance we say $\max(n, n')$ is the \emph{size} of the instance.
Note, given a \BBD instance of size $m$, the output of such instance is of size $\Oh(m^2)$.
Finally, we say that a $(\w, \bF, \bF')$-\BBD instance has \emph{diameter} $\D$ if $\dist_{\bar{\bG}}(\bar{u}, \bar{w}) \leq \D$ for all $\bar{u} \in \sB^{\bot}(1, 1, n+1, n'+1)$ and $\bar{w} \in \sB^{\top}(1, 1, n+1, n'+1)$ i.e.\ if the alignment graph $\bar{\bG}$ has diameter $\D$. 

\begin{lemma} \label{lem:fed_bbsp}
    Let $\bF$, $\bF'$ be forests with pre-orders $v_1, \ldots, v_{\abs{\bF}}$
    and $v_1', \ldots, v_{\abs{\bF'}}'$. Then, \FED on $\bF$ and $\bF'$
    can be reduced to \BBD on $(\w_{\bF, \bF'}, \bF, \bF')$,
    where $\w_{\bF, \bF'}(i, i') = \similarity(\sub(v_i), \sub(v_{i'}'))$ for $(i, i') \in \fragment{1}{\abs{\bF}} \times \fragment{1}{\abs{\bF'}}$.
\end{lemma}

\begin{proof}
    Recall Shasha and Zhang's dynamic programming algorithm \cite{SZ89}
    for computing the similarity between two forests $\bF$ and $\bF'$
    with pre-order $v_1, \ldots, v_{|\bF|}$ and $v_1', \ldots, v_{|\bF'|}'$, described by the following recursive formula:
    \begin{align*}
        \similarity(\bF, \bF) =
        \begin{cases}
            0, & \text{if $\bF = \emptyset$ or $\bF' = \emptyset$,} \\
            \max
            \left\{\begin{aligned}
                &\similarity(\bF \setminus v_1, \bF'),\\
                &\similarity(\bF, \bF' \setminus v_1'), \\
                &\similarity(\bF \setminus \sub(v_1), \bF' \setminus \sub(v_1')) + \similarity(\sub(v_1), \sub(v_1'))
            \end{aligned}\right\},
            & \text{otherwise}.
      \end{cases}
    \end{align*}
    Consider arbitrary $\fragment{i}{j} \subseteq \fragment{1}{(\abs{\bF}+1)}$ and $\fragment{i'}{j'} \subseteq \fragment{1}{(\abs{\bF'}+1)}$.
    Then, the dynamic programming computation
    of the longest path in $\bar{\bG}\fragment{i}{j}\fragment{i'}{j'}$ between $(i,i')$ and $(j,j')$
    aligns with 
    the computation done by Shasha and Zhang's scheme
    for $\similarity(\bF\fragmentco{\Left(v_i)}{\Left(v_j)}, \bF'\fragmentco{\Left(v_{i'}')}{\Left(u_{j'}')})$.
    This allows us, given $x,y \in \fragment{1}{(2\abs{\bF}+1)}$ such that $x \leq y$, to find
    two nodes in the grid whose longest distance equals $\similarity(\bF\fragmentco{x}{y}, \bF'\fragmentco{x'}{y'})$.
    To this end, it suffices to define $i,j$ as the smallest $i$ such that $x \leq \Left(v_i)$ and the smallest $j$ such that $y \leq \Left(v_j)$,
    obtaining $\bF\fragmentco{x}{y} = \bF\fragmentco{\Left(v_i)}{\Left(v_j)}$.
    Symmetrically, the same holds for $\bF'$.

    To conclude, note that whole-versus-infix similarities correspond to paths from the lower border to the upper border
    and from the left border to the right border,
    while prefix-versus-suffix similarities correspond to paths from the lower border to the right border and from the left border to the upper border.
\end{proof}

To prove \cref{thm:fed} and \Cref{thm:unweighted-fed-balance}, we proceed to develop an algorithm for \BBD instances.
The weight function $\w$ in the algorithm we present can be an arbitrary function.

\subsection{A decomposition scheme for forests}

To solve such \BBD instances,
we employ a divide-and-conquer approach.
Given a $(\w, \bF, {\bF'})$-\BBD instance,
we break it into several $(\w, \bH, {\bH'})$-\BBD instances,
for smaller forests $\bH$ and $\bH'$.
We then merge these solutions to obtain the final solution for $\bF$ and $\bF'$.
To derive these smaller instances,
we utilize the decomposition scheme introduced by Mao \cite{M22} in his subcubic algorithm for unweighted tree edit distance (though with a different parameter choice).
This scheme, governed by a threshold parameter $\Delta$,
uses transitions of two types:
\begin{enumerate}[I.]
    \item transition from two synchronous forests $\bF_{1}$, $\bF_{2}$,
    both of size no less than $\Delta/3$,
    to the synchronous subforest $\bF = \bF_{1} + \bF_{2}$;
    \label{it:decomp:1}
    \item transition from a synchronous forest $\bF_{s} \subset \bF$ such that $|\bF \setminus \bF_s| \leq \Delta$
    to a synchronous forest $\bF$.
    \label{it:decomp:2}
\end{enumerate}

The subsequent two lemmas,
the proofs of which we defer to \cref{sec:puttogether},
show how transitions in Mao's decomposition scheme can be adapted to our scenario.

\begin{restatable}{lemma}{fedpatchtwo}
    \label{lem:fed_patch_two}
    Suppose we are given a $(\w, \bF, {\bF'})$-\BBD
    instance of size $m$ such that $\bF = \bF_1 + \bF_2$.

    Then,
    given the outputs of the $(\w, {\bF_1}, {\bF'})$-\BBD instance
    and the $(\w, {\bF_2}, {\bF'})$-\BBD instance,
    we can solve the $(\w, \bF, {\bF'})$-\BBD instance
    in time $\Oh(\sT_{\MUL}(m))$.

    \unbalanced{Furthermore, if the \BBD instance has diameter $\D$, we can solve the $(\w, \bF, {\bF'})$-\BBD instance in time $\Oh(\sT_{\MonMUL}(m, m, m, \D))$.}
\end{restatable}

\begin{restatable}{lemma}{fedpatchthree}
    \label{lem:fed_patch_three}
    Suppose we are given a $(\w, \bF, {\bF'})$-\BBD
    instance of size $m$ and a synchronous subforest $\bF_s \subseteq \bF$
    of size $\abs{\bF_s} = n_s$.
    Further, let $v_1, \ldots, v_{\abs{\bF}}$ be the pre-order of $\bF$,
    and let $i_s$ be such that $v_{i_s}, \ldots, v_{i_s+n_s}$
    is the pre-order of $\bF_s$.
    Define $\bF_{\ell} = \bF\fragmentco{1}{\Left(v_{i_s})}$, $\bF_s = \bF\fragmentco{\Left(v_{i_s})}{\Left(v_{i_s+n_s+1})}$,
    and $\bF_{r} = \bF\fragmentco{\Left(v_{i_s+n_s+1})}{2\abs{\bF}+1}$.

    Then,
    given the outputs of the four \BBD instances $(\w, \bF_{\ell}, {\bF'})$-\BBD,
    $(\w, {\bF_s}, {\bF'})$-\BBD, $(\w, {\bF_{r}}, {\bF'})$-\BBD,
    and $(\w, {\bF \setminus \bF_{s}}, {\bF'})$-\BBD,
    we can solve the $(\w, \bF, {\bF'})$-\BBD instance
    in time $\Oh(\sT_{\MUL}(m))$.

    \unbalanced{Furthermore, if the \BBD instance has diameter $\D$, we can solve the $(\w, \bF, {\bF'})$-\BBD instance in time $\Oh(\sT_{\MonMUL}(m, m, m, \D))$.}
\end{restatable}

The decomposition scheme, together with \cref{lem:fed_patch_two} and \cref{lem:fed_patch_three},
allows us to establish the following reduction from a
$(\w, \bF, {\bF'})$-\BBD instance of size $m$
to $\mathcal{O}(m^2/\Delta^2)$ \BBD instances, each of size at most $\Delta$.

\begin{lemma}
    \label{lem:fed_decomp}
    Suppose, we are given a $(\w, \bF, {\bF'})$-\BBD instance of size $m$,
    and a threshold $\Delta$.

    Then, in time $\Oh(m^2)$ we can find forests $\bF_{1}, \ldots, \bF_{k}$ and $\bF_{1}', \ldots, \bF_{k}'$,
    all of size at most $\Delta$, such that $k = \Oh(m^2/\Delta^2)$ and such that,
    given the output of the $(\w, {\bF_i}, {\bF_{i}'})$-\BBD instance
    for all $i \in \fragment{1}{k}$,
    we can solve the $(\w, \bF, {\bF'})$-\BBD instance
    in time $\Oh(m^2/\Delta^2 \cdot \sT_{\MUL}(m))$.
    
    \unbalanced{Furthermore, if the \BBD instance has diameter $\D$, we can solve the $(\w, \bF, {\bF'})$-\BBD instance in time $\Oh(m^2 / \Delta^2 \cdot \sT_{\MonMUL}(m, m, m, \D))$.}
\end{lemma}

\begin{proof}
    First, we want to argue that it is sufficient to demonstrate the following simplified version of the lemma:
    In time $\Oh(m)$ we can find $\bF_{1}, \ldots, \bF_{d}$, all of size at most $\Delta$,
    such that $d = \Oh(m/\Delta)$ and such that,
    given the output of the $(\w, {\bF_i}, {\bF'})$-\BBD instance
    for all $i \in \fragment{1}{d}$,
    we can solve the $(\w, \bF, {\bF'})$-\BBD instance
    in time $\Oh(m/\Delta \cdot \sT_{\MUL}(m))$.
    This simplification is sufficient to prove the lemma, as we can apply this lemma to each \((\w, {\bF_i}, {\bF'})\text{-}\BBD\) instance
    for $i \in \fragment{1}{d}$, this time decomposing \(\bF'\)
    (\BBD is symmetrical w.r.t.~swapping the role of $\bF$ and $\bF'$).

    \begin{algorithm}[!t]
        \KwInput{a forest $\bH$.}
        \KwOutput{answer to the $(\w, {\bH}, {\bF'})$-\BBD instance.}
        \CommentSty{ // Given a forest $\bH$ containing more than one tree,
    we write $\bL_{\bH}$ and $\bR_{\bH}$ for the leftmost and rightmost tree, respectively.}
    
        \CommentSty{ // Given a forest $\bH$ containing one tree, we write $\Root(\bH)$ for the root of $\bH$.}
    
        \If {$\abs{\bH} \leq \Delta$} {
            As $\abs{\bH} \leq \Delta$, retrieve directly the answer to the $(\w, {\bH}, {\bF'})$-\BBD instance, and return it\;
        }
        \ElseIf {$\bH$ contains more than one tree and $\abs{\bL_{\bH}} \ge \Delta/3$ and $\abs{\bR_{\bH}} \ge \Delta/3$} {
            Via \cref{alg:mao_decomp}, get the answer to the $(\w, {\bL_{\bH}}, {\bF'})$-\BBD and $(\w, {\bH \setminus \bL_{\bH}}, {\bF'})$-\BBD instances\; \label{line:two_decomp}
            Via \cref{lem:fed_patch_two} on $\bL_{\bH}$ and $\bH \setminus \bL_{\bH}$, get the answer $(\w, {\bH}, {\bF'})$-\BBD instance, and return it\;

        } \Else {
            $\bH_s \leftarrow \bH$\;
            \While {\textbf{true}} {
                $\bH_{\mathrm{next}} \leftarrow \emptyset$\;
                \If{$\bH_s$ contains only one tree} {
                    $\bH_{\mathrm{next}} \leftarrow \bH_s \setminus \Root(\bH_s)$\;
                } \Else {
                    \If{$\abs{\bL_{\bH_s}} < \abs{\bR_{\bH_s}}$} {
                        $\bH_{\mathrm{next}} \leftarrow \bH_s \setminus \bL_{\bH_s}$\;
                    } \Else {
                        $\bH_{\mathrm{next}} \leftarrow \bH_s \setminus \bR_{\bH_s}$\;
                    }
                }
                \If {$\abs{\bH} - \abs{\bH_{\mathrm{next}}} > \twothirds \Delta$} {
                    \textbf{break}\;
                }
                $\bH_s \leftarrow \bH_{\mathrm{next}}$\;
            }
            Via \cref{alg:mao_decomp}, get the answer to the $(\w, {\bH_s}, {\bF'})$-\BBD instance\;

            As $\abs{\bH \setminus \bH_s} \leq \Delta$, retrieve directly the answer to
            the $(\w, {\bH \setminus \bH_s}, {\bF'})$-\BBD,
            $(\w, {\bH_{\ell}}, {\bF'})$-\BBD,
            $(\w, {\bH_{r}}, {\bF'})$-\BBD instances,
            where $\bH_{\ell}, \bH_{r}$ are defined as in \cref{lem:fed_patch_three}.

            Via \cref{lem:fed_patch_three} on $\bH$ and $\bH_s$, get the answer of the $(\w, {\bH}, {\bF'})$-\BBD instance, and return it\;
        }
        \caption{Adaptation of Algorithm 2 from \cite{M22}.} \label{alg:mao_decomp}
    \end{algorithm}

    To prove the simplified version of the lemma,
    we utilize \cref{alg:mao_decomp} (an adaptation of Algorithm 2 from~\cite{M22})
    on the $(\w, {\bF}, {\bF'})$-\BBD instance.
    \cref{alg:mao_decomp} applies recursively Mao's transitions with the same parameter $\Delta$ on forests $\bH \subseteq \bF$,
    until it is left with forests of size at most $\Delta$.
    \cref{alg:mao_decomp} assumes that for $\bH$ of size at most $\Delta$ the output of the $(\w, {\bH}, {\bF'})$-\BBD instance can be retrieved directly
    (this assumption is equivalent to putting such instances among $\bF_{1}, \ldots, \bF_{d}$).
  
    The correctness of \cref{alg:mao_decomp} follows directly from \cref{lem:fed_patch_two} and \cref{lem:fed_patch_three}.

    What remains to be demonstrated is a bound on $d$ and on the running time.
    For that purpose, notice that $d = \Oh(t_{\mathrm{I}} + t_{\mathrm{II}})$ and
    the running time is bounded by $\Oh((t_{\mathrm{I}} + t_{\mathrm{II}}) \cdot \sT_{\MUL}(m))$,
    where $t_{\mathrm{I}}, t_{\mathrm{II}}$ is the number of times we apply transition of type~\ref{it:decomp:1} and~\ref{it:decomp:2}, respectively.
    \unbalanced{If the diameter of $\bar{\bG}$ is at most $\D$, then we note that the running time is in fact bounded by $\Oh((t_{\mathrm{I}} + t_{\mathrm{II}}) \cdot \sT_{\MonMUL}(m, m, m, \D))$.}
    To bound $t_{\mathrm{I}} + t_{\mathrm{II}}$, we report here the argumentation from Section 4.2.2. of~\cite{M22}.

    For every transition of type \ref{it:decomp:1},
    we merge two forests each of size no less than \( \Delta/3 \),
    yielding \( t_{\text{I}} = \mathcal{O}(|\bF|/\Delta) \).

    Now, for \( t_{\text{II}} \), let us decompose it into \( t_{\text{II}} = t_{\text{II}}^{(1)} + t_{\text{II}}^{(2)} \),
    where \( t_{\text{II}}^{(1)} \) denotes the number of transitions of type \ref{it:decomp:2} where \( \bH_s \) contains fewer than two trees or either \( |\bL_{\bH_s}| \)
    or \( |\bR_{\bH_s}| \) is less than \( \Delta/3 \),
    and \( t_{\text{II}}^{(2)} \) represents the number of transitions of type \ref{it:decomp:2} where \( \bH_s \) contains at least two trees and both \( |\bL_{\bH_s}| \) and \( |\bR_{\bH_s}| \) are no less than \( \Delta/3 \).
    For \( t_{\text{II}}^{(1)} \), we have \( |\bH \setminus \bH_s| > \Delta/3 \), as otherwise, we would have not halted the removal process.
    Since \( \bH \setminus \bH_s \) are disjoint across different type \ref{it:decomp:2} transitions,
    we have \( t_{\text{II}}^{(1)} = \mathcal{O}(|\bF|/\Delta) \).
    Regarding \( t_{\text{II}}^{(2)} \), observe that the subsequent transition will be of type \ref{it:decomp:1},
    so \( t_{\text{II}}^{(2)} \leq t_{\text{I}} = \mathcal{O}(|\bF|/\Delta) \).
\end{proof}

\begin{corollary} 
    \label{cor:bbd_algo}
    There is an algorithm that solves in time $\Oh(\sT_{\MUL}(m)+m^{2+o(1)})$
    a $(\w, {\bF}, {\bF'})$-\BBD instance of size $m$.

    \unbalanced{Furthermore, if the \BBD instance has diameter $\D$, the algorithm runs in time $\Oh(\sT_{\MonMUL}(m, m, m, \D) + m^{2 + o(1)} g(D))$, where $\sT_{\MonMUL}(m, m, m, \D) = \Oh(f(m) g(\D))$ for some functions $f, g$.}
\end{corollary}

\begin{proof}
    We apply \cref{lem:fed_decomp} recursively with threshold $\Delta = m/\alpha$
    for some constant $\alpha \geq 1$ to be determined later.
    For small enough instances any algorithm computing shortest paths in a directed
    acyclic graph will do.
    Thereby, we obtain an algorithm for the \FED Problem where the running time is described by the
    recurrence
    \begin{align*}
        \label{eq:cor_bbd_algo}
        \sT(m) &\leq c_1 \cdot (m/\Delta)^2 \cdot \sT(\Delta) + c_2 \cdot (m/\Delta)^2 \cdot \sT_{\MUL}(m) + c_3 \cdot m^2 \\
        &=  c_1 \alpha^{2} \cdot \sT(m/\alpha) + c_2 \alpha^{2} \cdot \sT_{\MUL}(m) + c_3 \cdot m^2
    \end{align*}
    where $c_1, c_2, c_3$, are the constants hidden in the number of problems we recurse on (parameter $k$ in \cref{lem:fed_decomp}), in the time needed to solve the $(\w, {\bH}, {\bF'})$-\BBD instance, and the time needed to find the $k$ pairs of subforests, respectively. Now, for any $\epsilon > 0$,
    we can choose a sufficiently large constant $\alpha$ such that
    \[
        \log_{\alpha}(c_1 \alpha^{2}) = \log_{\alpha} c_1 + 2 < 2 + \epsilon
    \]
    and still $c_2 \alpha^{2} = \Oh(1)$.
    By applying the Master theorem, we conclude that $\sT(m) = \Oh(\sT_{\MUL}(m)+m^{2+o(1)})$.

    \unbalanced{If the diameter is at most $\D$, we instead have the recurrence,
    \begin{align*}
        \sT(m) &= c_1 \alpha^{2} \cdot \sT(m/\alpha) + c_2 \alpha^{2} \cdot \sT_{\MonMUL}(m, m, m, \D) + c_3 \cdot m^2\\
        & = c_1 \alpha^{2} \cdot \sT(m/\alpha) + O(f(m) g(\D)). 
    \end{align*}
    By the master theorem, we can choose a sufficiently large constant $\alpha$ such that $\log_{\alpha}(c_1 \alpha^{2}) < 2 + \epsilon$ for any $\epsilon$. By applying the Master theorem, we conclude that $\sT(m) = \Oh((f(m)+m^{2+o(1)}) g(D)) = \Oh(\sT_{\MonMUL}(m, m, m, \D) + m^{2 + o(1)} g(D))$.
    }
\end{proof}

\fed

\begin{proof}
    We use \cref{cor:bbd_algo} on the weight function defined in \cref{lem:fed_bbsp}.
    Note that when the $\FED$ instance is unweighted, the \BBD instance has diameter at most $\Oh(n)$.
    In particular, we have $\sT_{\MonMUL}(n, n, n, \Oh(n)) = \Oh(\sT_{\MonMUL}(n))$.
\end{proof}

\subsection{Patching together subproblems}
\label{sec:puttogether}

\fedpatchtwo

\begin{proof}
    Let $n_1 = \abs{\bF_1}$, $n_2 = \abs{\bF_2}$, $n = \abs{\bF} = n_1 + n_2$, and set $i_2 = n_1+1$.
    Consider the pre-order traveral $v_1, \ldots, v_n$ of $\bF$,
    and observe that $v_1, \ldots, v_{n_1}$ corresponds to the pre-order of $\bF_1$
    and that $v_{i_2}, \ldots, v_{n}$ corresponds to the pre-order traversal of $\bF_2$.
    Further, consider the alignment graph $\bar{\bG}$ w.r.t.~$(\w, {\bF}, {\bF'})$.
    The alignment graph \(\bar{\bG}_1\) w.r.t.~\((\w, {\bF_1}, {\bF'})\)
    corresponds to $\bar{\bG}\fragment{1}{i_2}\fragment{1}{(n'+1)}$,
    and the alignment graph \(\bar{\bG}_2\) w.r.t.~\((\w, {\bF_2}, {\bF'})\)
    corresponds to $\bar{\bG}\fragment{i_2}{(n+1)}\fragment{1}{(n'+1)}$.

    \begin{figure}[htbp]
        \centering
        \usetikzlibrary{arrows.meta, bending}

\begin{tikzpicture}
    \draw[thick] (0,0) rectangle (6,2);
    \draw[red, line width=7pt, opacity=0.5] (0,2) --node[left=20pt, thick, opacity=1] {$\sB^{\bot}(1, 1, n+1, n'+1)$} (0,0) -- (6,0);
    \draw[blue, line width=7pt, opacity=0.5] (6, 0) --node[right=14pt, thick, opacity=1] {$\sB^{\top}(1, 1, n+1, n'+1)$} (6,2) -- (0,2);
    \draw[dashed, thick] (3,0) --node[right] {$i_2$} (3,2);

    \draw[Bracket-Parenthesis, thick] (-0.3,2) -- (-0.3,-0.3) --node[below] {$I_1^{\bot}$} (3,-0.3);
    \draw[Bracket-Bracket, thick] (3,-0.3) --node[below right] {$I_2^{\bot}$} (6,-0.3);

    \draw[Bracket-Parenthesis, thick] (0,2.3) --node[above] {$I_1^{\top}$} (3,2.3);
    \draw[Bracket-Bracket, thick] (3,2.3) --node[above] {$I_2^{\top}$} (6.3,2.3) -- (6.3,0);

\end{tikzpicture}
        \caption{A \BBD instance can be thought as a rectangle where we need to compute distances
        from the lower left border to the upper right border.
        In \cref{lem:fed_patch_two} we `cut' the rectangle between $\bF_1$ and $\bF_2$,
        and given the answer for the instances corresponding to the two resulting rectangle halves,
        we show how to patch them together for the full rectangle.
        In order to do so, we split the lower left border $\sB^{\bot}(1, 1, n+1, n'+1)$
        into $I^\bot_1$ and $I^\bot_2$, and the upper right border $\sB^{\top}(1, 1, n+1, n'+1)$ into $I^\top_1$ and $I^\top_2$.}
        \label{fig:patch_two}
    \end{figure}
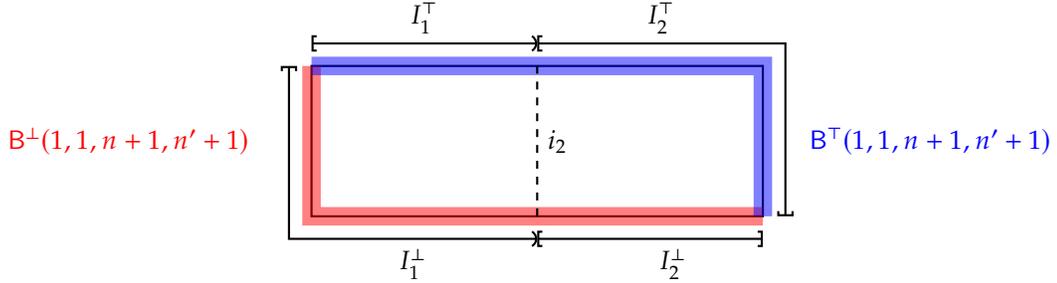

    We rewrite $\sB^{\bot}(1, 1, n+1, n'+1) = I^\bot_1 \cup I^\bot_2$,
    where $I^\bot_1 = \sB^{\bot}(1, 1, i_2-1, n'+1)$ and $I^\bot_2 = \fragment{i_2}{(n+1)} \times 1$.
    Similarly, we rewrite $\sB^{\top}(1, 1, n+1, n'+1) = I^\top_1 \cup I^\top_2$,
    where $I^\top_1 = \fragmentco{1}{i_2} \times (n'+1)$ and $I^\top_2 =  \sB^{\top}(i_2, 1, n+1, n'+1)$.
    See \cref{fig:patch_two} for a visualization of the newly defined index sets.

    We split up the indices for the output of the $(\w, \bF, {\bF'})$-\BBD instance as follows.
    \begin{align}
        \MoveEqLeft \{\ (\bar{u}, \bar{v}) \mid \bar{u} \in \sB^{\bot}(1, 1, n+1, n'+1), \ \bar{v} \in \sB^{\top}(1, 1, n+1, n'+1)\ \} = \nonumber \\[5pt]
        \vspace{10pt}
        & \{\ (\bar{u}, \bar{v})  \mid \bar{u} \in I^\bot_1, \
        \bar{v} \in I^\top_1  \ \} \label{ptwo:eq:output:1}\\
        \cup \quad
        & \{\ (\bar{u}, \bar{v})  \mid \bar{u} \in I^\bot_1, \
        \bar{v} \in I^\top_2  \ \} \label{ptwo:eq:output:2}\\
        \cup \quad
        & \{\ (\bar{u}, \bar{v})  \mid \bar{u} \in I^\bot_2, \
        \bar{v} \in I^\top_1 \ \} \label{ptwo:eq:output:3}\\
        \cup \quad
        & \{\ (\bar{u}, \bar{v})  \mid \bar{u} \in I^\bot_2, \
        \bar{v} \in I^\top_2  \ \}. \label{ptwo:eq:output:4}
    \end{align}
    We show separately how we compute the distances for the index sets~\eqref{ptwo:eq:output:1},~\eqref{ptwo:eq:output:2},~\eqref{ptwo:eq:output:3}, and~\eqref{ptwo:eq:output:4}.
    \begin{itemize}
        \item For~\eqref{ptwo:eq:output:1}, we can get $\dist_{\bar{\bG}}(\bar{u}, \bar{v})$
        using the output of the $(\w, {\bF_1}, {\bF'})$-\BBD instance.
        \item For~\eqref{ptwo:eq:output:2}, we use that for every $i \in \fragmentco{1}{i_2}$ we have $\pi_{\bF}(i) \leq i_2$ as $\bF = \bF_1 + \bF_2$.
            As a consequence, a path goes from $\bar{u}$ to $\bar{v}$ must pass through a node contained in the set $i_2 \times \fragment{1}{(n'+1)}$, and we can write
            \begin{align*}
                \dist_{\bar{\bG}}(\bar{u}, \bar{v})
                &= \max\nolimits_{\bar{w} \in i_2 \times\fragment{1}{(n'+1)}}
                \Big\{ \ \dist_{\bar{\bG}}(\bar{u}, \bar{w}) + \dist_{\bar{\bG}}(\bar{w}, \bar{v}) \ \Big\}.
            \end{align*}
            The two summands that appear in the latter maximization are a subset of the output of the $(\w, {\bF_1}, {\bF'})$-\BBD instance
            and of the $(\w, {\bF_2}, {\bF'})$-\BBD instance.
            Moreover, note that we may rewrite the computation of $\dist_{\bar{\bG}}(\bar{u}, \bar{v})$ for all such $\bar{u}$ and $\bar{v}$
            as a max-plus product $A = B \star C$, where $A$, $B$, and $C$ are matrices of size
            $(n_1+n'+1) \times (n_2+n'+1)$, $(n_1+n'+1) \times n'$, and $n' \times (n_2+n'+1)$, respectively.
            As all dimensions of the matrices $A$, $B$, ans $C$ are upper bounded by $\Oh(m)$, we spend at most $\Oh(\sT_{\MUL}(m))$ time in the computation of these distances.

            \unbalanced{If $\w$ is unweighted, we claim that $C$ is a column-monotone matrix. 
            In particular, whenever $\bar{w} \leq \bar{w'}$, there is a path (of weight $0$) from $\bar{w}$ to $\bar{w'}$ so that the maximum weight path from $\bar{w}$ to $\bar{u}$ is at least the maximum weight path from $\bar{w'}$ to $\bar{u}$.
            Furthermore, since the original alignment graph $\bar{\bG}$ has diameter $\D$, the entries are bounded by $\D$.
            Thus, the matrix product requires time $\Oh(\sT_{\MonMUL}(m, m, m, \D))$.}
        \item For~\eqref{ptwo:eq:output:3}, we have $\dist_{\bar{\bG}}(\bar{u}, \bar{v}) = -\infty$, as any such $\bar{v}$ is not reachable from any such $\bar{u}$ in $\bar{\bG}$.
        \item For~\eqref{ptwo:eq:output:4}, we can get $\dist_{\bar{\bG}}(\bar{u}, \bar{v})$
        using the output of the $(\w, {\bF_2}, {\bF'})$-\BBD instance.
        \qedhere
    \end{itemize}
\end{proof}

\fedpatchthree

\begin{proof}
    Let $\abs{\bF} = n$, and define $i_r = i_s + n_s + 1$.
    Consider the alignment graph $\bar{\bG}$ w.r.t.~$(\w, \bF, {\bF'})$,
    and observe that $\bar{\bG}\fragment{1}{i_s}\fragment{1}{n'}$ corresponds to the alignment graph
    $\bar{\bG}_{\ell}$ w.r.t.~$(\w, {\bF_{\ell}}, {\bF'})$,
    $\bar{\bG}\fragment{i_s}{i_r}\fragment{1}{n'}$ corresponds to the alignment graph
    $\bar{\bG}_{s}$ w.r.t.~$(\w, {\bF_{s}}, {\bF'})$, and
    $\bar{\bG}\fragment{i_r}{n+1}\fragment{1}{n'}$ corresponds to the alignment graph
    $\bar{\bG}_{r}$ w.r.t.~$(\w, {\bF_{r}}, {\bF'})$.
    The alignment graph $\bar{\bG}_{\setminus s}$ w.r.t.~$(\w, {\bF \setminus \bF_s}, {\bF'})$,
    corresponds to the graph obtained by removing from $\bar{\bG}$ all nodes (and their incident edges)
    of the set $\fragmentco{i_s}{i_r} \times \fragment{1}{(n'+1)}$, and by adding an edge from $(x, i_s-1)$ to $(x, i_r)$
    of cost $0$ for all $x \in \fragment{1}{(n'+1)}$.
    We abuse notation, and we index nodes in $\bar{\bG}_{\setminus s}$ w.r.t.~the corresponding positions in $\bar{\bG}$, taking care of never
    using any index from the set $\fragmentco{i_s}{i_r} \times \fragment{1}{(n'+1)}$.

    \begin{figure}[htbp]
        \centering
        \usetikzlibrary{arrows.meta, bending}

\begin{tikzpicture}
    \draw[thick] (0,0) rectangle (8,2);
    \draw[red, line width=7pt, opacity=0.5] (0,2) --node[left=14pt, thick, opacity=1] {$\sB^{\bot}(1, 1, n+1, n'+1)$} (0,0) -- (8,0);
    \draw[blue, line width=7pt, opacity=0.5] (8, 0) --node[right=14pt, thick, opacity=1] {$\sB^{\top}(1, 1, n+1, n'+1)$} (8,2) -- (0,2);
    \draw[dashed, thick] (3,0) --node[right] {$i_{s}$} (3,2);
    \draw[dashed, thick] (5,0) --node[right] {$i_{r}$} (5,2);

    \draw[Bracket-Parenthesis, thick] (-0.3,2) -- (-0.3,-0.3) --node[below] {$I_{\ell}^{\bot}$} (3,-0.3);
    \draw[Bracket-Parenthesis, thick] (3,-0.3) --node[below right] {$I_{s}^{\bot}$} (5,-0.3);
    \draw[Bracket-Bracket, thick] (5,-0.3) --node[below right] {$I_{r}^{\bot}$} (8,-0.3);

    \draw[Bracket-Parenthesis, thick] (0,2.3) --node[above] {$I_{\ell}^{\top}$} (3,2.3);
    \draw[Bracket-Parenthesis, thick] (3,2.3) --node[above] {$I_{s}^{\top}$} (5,2.3);
    \draw[Bracket-Bracket, thick] (5,2.3) --node[above] {$I_{r}^{\top}$} (8.3,2.3) -- (8.3,0);

\end{tikzpicture}
        \caption{In \cref{lem:fed_patch_three}, we `cut' the rectangle in three subrectangles.
            This time, unlike \cref{lem:fed_patch_two}, edges of the form~\eqref{it:match} do not always stay
            in the subrectangle they originate from, as edges of the form~\eqref{it:match} might go from the leftmost
            to the rightmost subrectangle. To cover paths that include such edges, we use outputs from the
            $(\w, {\bF \setminus \bF_{s}}, {\bF'})$-\BBD instance,
            for which we `cut out' the subrectangle in the middle.
        }
        \label{fig:patch_three}
    \end{figure}
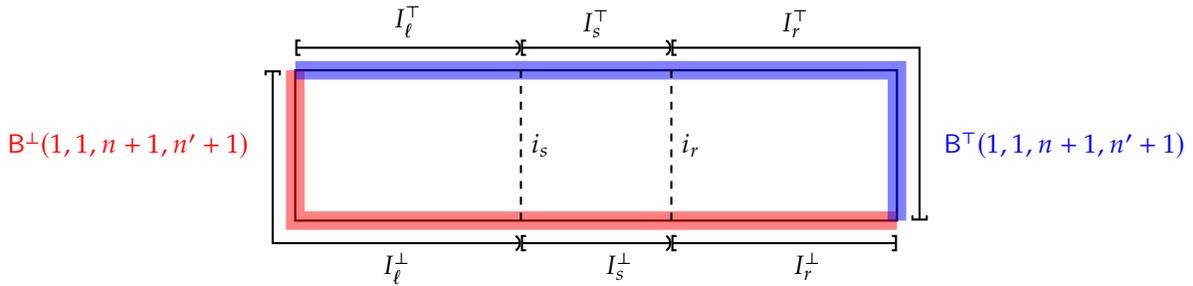

    We rewrite $\sB^{\bot}(1, 1, n+1, n'+1) = I^{\bot}_{\ell} \cup I^{\bot}_{s} \cup I^{\bot}_{r}$,
    where $I^{\bot}_{\ell} = \sB^{\bot}(1, 1, i_s-1, n'+1), I^{\bot}_{s} = \fragmentco{i_s}{i_r} \times 1, I^{\bot}_{r} = \fragment{i_r}{(n+1)} \times 1$,
    and $\sB^{\top}(1, 1, n+1, n'+1) = I^{\top}_{\ell} \cup I^{\top}_{s} \cup I^{\top}_{r}$,
    where $I^{\top}_{\ell} = \fragmentco{1}{i_s} \times (n'+1), I^{\top}_{s} = \fragmentco{i_s}{i_r} \times (n'+1), I^{\top}_{r} = \sB^{\top}(i_r, 1, n+1, n'+1)$.
    Refer to \cref{fig:patch_three} for a visualization of the newly defined index sets.
    Consider $\bar{u} \in I^{\bot}$ and $\bar{v} \in I^{\top}$ for some $(I^{\bot}, I^{\top}) \in \{I^{\bot}_{\ell}, I^{\bot}_{s}, I^{\bot}_{r}\} \times \{I^{\top}_{\ell}, I^{\top}_{s}, I^{\top}_{r}\}$.
    In order to calculate $\dist_{\bar{\bG}}(\bar{u}, \bar{v})$, we perform a case distinction on $I^{\bot}$ and $I^{\top}$.
    \begin{itemize}
        \item If $I^{\bot} = I^{\bot}_{s}$ and $I^{\top} = I^{\top}_{\ell}$,
            or $I^{\bot} = I^{\bot}_{r}$ and $I^{\top} \in \{I^{\top}_{\ell}, I^{\top}_{s}\}$,
            $\dist_{\bar{\bG}}(\bar{u}, \bar{v}) = -\infty$,
            as $\bar{v}$ is not reachable from $\bar{u}$.

        \item If $I^{\bot} = I^{\bot}_{\ell}, I^{\top} = I^{\top}_{\ell}$, then we can get $\dist_{\bar{\bG}}(\bar{u}, \bar{v})$
            from the output of the $(\w, {\bF_{\ell}}, {\bF'})$-\BBD instance.
        \item If $I^{\bot} = I^{\bot}_{s}, I^{\top} = I^{\top}_{s}$, then we can get $\dist_{\bar{\bG}}(\bar{u}, \bar{v})$
        from the output of the $(\w, {\bF_{s}}, {\bF'})$-\BBD instance.
        \item If $I^{\bot} = I^{\bot}_{r}, I^{\top} = I^{\top}_{r}$, then we can get $\dist_{\bar{\bG}}(\bar{u}, \bar{v})$
        from the output of the $(\w, {\bF_{r}}, {\bF'})$-\BBD instance.
        \item If $I^{\bot} = I^{\bot}_{\ell}, I^{\top} = I^{\top}_{s}$, observe that for any $i \in \fragmentco{1}{i_s}$ either $\pi_{\bF}(i) \leq i_s$ or $\pi_{\bF}(i) \geq i_r > i_s + n_s$ holds
        as $\bF_s$ is a synchronous subforest of $\bF$.
        Hence, a path in $\bar{\bG}$ that goes from $\bar{u}$ to $\bar{v}$ must pass through a node in $i_s \times \fragment{1}{(n'+1)}$,
        and we can write
        \begin{align*}
            \dist_{\bar{\bG}}(\bar{v}, \bar{u})
            &= \max\nolimits_{\bar{w} \in i_s \times\fragment{1}{(n'+1)}}
            \Big\{ \ \dist_{\bar{\bG}}(\bar{v}, \bar{w}) + \dist_{\bar{\bG}}(\bar{w}, \bar{u}) \ \Big\}.
        \end{align*}
        Similarly to \cref{lem:fed_patch_two}, we can rewrite the computation of $\dist_{\bar{\bG}}(\bar{v}, \bar{u})$ for all such $\bar{u}$ and $\bar{v}$
        using a max-plus product. This time, we use
        the output of the $(\w, {\bF_{\ell}}, {\bF'})$-\BBD instance and the
        output of the $(\w, {\bF_{s}}, {\bF'})$-\BBD instance to compute the entries of the matrices
        since $\bar{\bG}\fragment{1}{i_s}\fragment{1}{(n'+1)} = \bar{\bG}_{\ell}$ and
        $\bar{\bG}\fragment{i_s}{i_r}\fragment{1}{(n'+1)} = \bar{\bG}_{s}$.
        
        \unbalanced{If the diameter is at most $\D$, we note that following similar arguments as \Cref{lem:fed_patch_two}, the matrix multiplication requires time $\Oh(\sT_{\MonMUL}(m, m, m, \D))$.}

        \item $I^{\bot} = I^{\bot}_{s}, I^{\top} = I^{\top}_{r}$, observe that
        for $i \in \fragmentco{i_s}{i_r}$, we have $\pi_{\bF}(i) \leq i_r$ because $\bF_s$ is a synchronous subforest of $s$.
        Hence, a path in $\bar{\bG}$ that goes from $\bar{u}$ to $\bar{v}$ must pass through a node in $i_r \times \fragment{1}{n'+1}$,
        and we can write
        \begin{align*}
            \dist_{\bar{\bG}}(\bar{v}, \bar{u})
            &= \max\nolimits_{\bar{w} \in i_r \times\fragment{1}{n'+1}}
            \Big\{ \ \dist_{\bar{\bG}}(\bar{v}, \bar{w}) + \dist_{\bar{\bG}}(\bar{w}, \bar{u}) \ \Big\}.
        \end{align*}
        Similarly to the previous case, these computations can be performed using a max-plus product,
        where entries are take from the output of the $(\w, {\bF_{s}}, {\bF'})$-\BBD instance and the
        output of the $(\w, {\bF_{r}}, {\bF'})$-\BBD instance.

        \unbalanced{As above, if the diameter is at most $\D$, the matrix multiplication requires time $\Oh(\sT_{\MonMUL}(m, m, m, \D))$.}

        \item If $I^{\bot} = I^{\bot}_{\ell}, I^{\top} = I^{\top}_{r}$, we
        use again that either $\pi_{\bF}(i) \leq i_s$ or $\pi_{\bF}(i) \geq i_r > i_s + n_s$
        for $i \in \fragmentco{1}{i_s}$, and that ${\bF}(i) \leq i_r$ for $i \in \fragmentco{i_s}{i_r}$.
        Consequently, any path from $\bar{u}$ to $\bar{v}$ either stays entirely in $\bar{\bG}_{\setminus s}$, or
        passes through a node in $i_s \times \fragment{1}{(n'+1)}$ and a node in $i_r \times \fragment{1}{(n'+1)}$.
        Therefore, we obtain
        \begin{align*}
            \dist_{\bar{\bG}}(\bar{v}, \bar{u}) =
            \max \Biggl\{
                \ \dist_{\bar{\bG}_{\setminus s}}(\bar{v}, \bar{u}) \ ,
                \max_{\substack{\bar{w} \in i_s \times\fragment{1}{(n'+1)}\\\bar{z} \in i_r \times\fragment{1}{(n'+1)}}}
                \Big\{ \
                    \dist_{\bar{\bG}}(\bar{v}, \bar{w})
                    + \dist_{\bar{\bG}}(\bar{w}, \bar{z})
                    + \dist_{\bar{\bG}}(\bar{z}, \bar{u})
                \ \Big\} \
            \Biggl\}.
        \end{align*}
        These computation can be done by taking the maximum between outputs of the $(\w, {\bF_{\setminus s}}, {\bF'})$-\BBD instance,
        and a max-plus product between three matrices. The entries of these matrices can be
        retrieved from outputs of the other three \FED instance at our disposal.
        
        \unbalanced{As above, if the diameter is at most $\D$, the matrix multiplication requires time $\Oh(\sT_{\MonMUL}(m, m, m, \D))$.}
    \end{itemize}

    This concludes the proof of \cref{lem:fed_patch_three}.
\end{proof}

\subsection{Computing spine mapping similarities}

In this (sub)section we prove a lemma that will turn out to be useful in a special case
appearing in \cref{sec:sed}. We prove it here, since in order to solve it
we need to solve a \BBD instance, similar to the $(\w, {\bF_{\setminus s}}, {\bF'})$-\BBD instance
appearing in \cref{lem:fed_patch_three}.

\begin{lemma} \label{lem:sim_w_cut}
    Let $\bF, \bF'$ be forests, let $\bS \subseteq \bF$ be a spine, and let $q \in \bS$.

    Then, we can calculate the function $\similarity_q(\bF\fragmentco{x}{y}, \bF'\fragmentco{x'}{y'})$ for all
    $(x,x') \in \sB^{\bot}(1,1,\Left(q), 2\abs{\bF'}+1)$ and $(y,y') \in \sB^{\top}(\Right(q), 1, 2\abs{\bF}+1, 2\abs{\bF'}+1)$,
    where:
    \begin{itemize}
        \item $\similarity_q(\bF\fragmentco{x}{y}, \bF'\fragmentco{x'}{y'}) = \similarity(\bF\fragmentco{x}{y}, \bF'\fragmentco{x'}{y'})$
        if there exists $t \in \bS$ is such that $t \prec q$ and $\similarity(\bF\fragmentco{x}{y}, \bF'\fragmentco{x'}{y'})$ aligns $\sub(t)$ to any subtree of $\bF'$, and
        \item $\similarity_q(\bF\fragmentco{x}{y}, \bF'\fragmentco{x'}{y'}) \leq \similarity(\bF\fragmentco{x}{y}, \bF'\fragmentco{x'}{y'})$ otherwise.
    \end{itemize}
    This computation takes time $\Oh(\sT_{\MUL}(m_q) + m_q^{2 + o(1)})$, where $m_q = \max(\abs{\bF \setminus \sub(q)}, \abs{\bF'})$,
    and requires as input the values $\similarity(\sub(v), \sub(v'))$ for all $v \in \bF \setminus \sub(q), v' \in \bF'$.

    \unbalanced{In the unweighted setting, the computation takes time $\Oh(\sT_{\MonMUL}(m_q, m_q, m_q, \D) + m_{q}^{2 + o(1)} g(\D))$ where $\D = \max(\abs{\bF}, \abs{\bF'})$ and $\sT_{\MonMUL}(m, m, m, \D) = f(m) g(\D)$.}
\end{lemma}

\begin{proof}
    Let $n = \abs{\bF}, n' = \abs{\bF'}$, and $n_q = \abs{\sub(q)}$.
    Further, let $v_{1}, \ldots, v_{n}$ and $v_{1}', \ldots, v_{n'}'$ be the pre-orders of $\bF$ and $\bF'$, respectively, and let $i_q$ be such that
    $v_{i_q}, \ldots, v_{i_q+n_q}$ is the pre-order of $\sub(q)$.

    Consider the alignment graph $\bar{\bG}$ w.r.t.~$(\w_{\bF, \bF'}, {\bF}, {\bF'})$.
    Define the alignment graph $\bar{\bG}_{\setminus q}$, as the alignment graph obtained from $\bar{\bG}$
    by removing from all nodes (and their incident edges) of the set $\fragment{i_q}{(i_q + n_q)} \times \fragment{1}{(n'+1)}$,
    and by adding an edge from $(x, i_q-1)$ to $(x, i_q+n_q+1)$ of cost $0$ for all $x \in \fragment{1}{(n'+1)}$.
    Note, the construction of $\bar{\bG}_{\setminus q}$ is possible as it requires to know
    the values $\similarity(\sub(v), \sub(v'))$ for all $v \in \bF \setminus \sub(q), v' \in \bF$.

    We claim that solving \BBD on $\bar{\bG}_{\setminus q}$ in time $\Oh(\sT_{\MUL}(m_q) + m_{q}^{2 + o(1)})$, leads to the desired function.
    \unbalanced{In the unweighted setting, \Cref{cor:bbd_algo} shows that \BBD can instead by computed in time $\Oh(\sT_{\MonMUL}(m_{q}, m_{q}, m_{q}, \D) + m_{q}^{2 + o(1)} g(\D))$.}
    For $(x,x') \in \sB^{\bot}(1,1,\Left(q), 2\abs{\bF'}+1)$ and $(y,y') \in \sB^{\top}(\Right(q), 1, 2\abs{\bF}+1, 2\abs{\bF'}+1)$, we define $\similarity_q(\bF\fragmentco{x}{y}, \bF'\fragmentco{x'}{y'})$
    to be the longest path from $(x,y)$ to $(x',y')$ in $\bar{\bG}_{\setminus q}$ (in this proof we index nodes in $\bar{\bG}_{\setminus q}$ w.r.t. their original indices in $\bar{\bG}$).

    To see that the definition serves our purpose, consider arbitrary $(x,x') \in \sB^{\bot}(1,1,\Left(q), 2\abs{\bF'}+1)$ and $(y,y') \in \sB^{\top}(\Right(q), 1, 2\abs{\bF}+1, 2\abs{\bF'}+1)$.
    Whenever the longest path from $(x,x')$ to $(y, y')$
    takes an outgoing edge from a node $(i,i')$ of the form~\eqref{it:match},
    where $v_i = s$ for some $t \in S$ such that $t \prec q$,
    then $\similarity(\bF\fragmentco{x}{y}, \bF'\fragmentco{x'}{y'})$ aligns $\sub(t)$ to $\sub(v_{i'}')$.    
    Note, all such outgoing edges
    start from a node in the set $\fragment{1}{(i_q-1)} \times \fragment{1}{(n'+1)}$ and arrive to a node in
    the set $\fragment{(i_q+n_q+1)}{(n+1)} \times \fragment{1}{(n'+1)}$.
    Consequently, such longest path survives in $\bar{\bG}_{\setminus q}$,
    and we conclude $\similarity_q(\bF\fragmentco{x}{y}, \bF'\fragmentco{x'}{y'}) \geq \similarity(\bF\fragmentco{x}{y}, \bF'\fragmentco{x'}{y'})$.

    On the other hand, for any $(x,x') \in \sB^{\bot}(1,1,\Left(q), 2\abs{\bF'}+1)$ and $(y,y') \in \sB^{\top}(\Right(q), 1, 2\abs{\bF}+1, 2\abs{\bF'}+1)$, observe that the longest path between $(x,y)$ to $(x',y')$ in $\bar{\bG}_{\setminus q}$ is no longer than the one between $(x,y)$ to $(x',y')$ in $\bar{\bG}$. We conclude, $\similarity_q(\bF\fragmentco{x}{y}, \bF'\fragmentco{x'}{y'}) \leq \similarity(\bF\fragmentco{x}{y}, \bF'\fragmentco{x'}{y'})$.
\end{proof}

\subsection{Forest Edit Distance on Unbalanced Instances}
\label{sec:fed-unbalance}

In this section, we consider instances where one forest $\bF$ is significantly larger than the other $\bF'$. 
Without loss of generality, assume $n \geq n'$.
We show that we can solve a restricted versions of the \BBD problem more efficiently.

\defproblem{Left-Right Border-to-Border Distances (\LRBBD)}
{two forests $\bF$, $\bF'$ of size $\abs{\bF}=n, \abs{\bF'}=n'$, and a weight function $\w : D \rightarrow \mathbb{R}$ such that $D \subseteq \bF \times \bF'$.}
{$\dist_{\bar{\bG}}(\bar{u}, \bar{w})$ for all $\bar{u} \in 1 \times \fragment{1}{n' + 1}$ and $\bar{w} \in \sB^{\top}(1, 1, n+1, n'+1)$, 
where $\bar{\bG}$ is the alignment graph w.r.t.~$(\w, \bF, \bF')$.}

We say $(n, n')$ is the size of the \LRBBD instance.
We will solve the unbalanced \LRBBD instance by decomposing the larger forest into forests of size at most $n'$ using Mao's decomposition algorithm (\cref{alg:mao_decomp}), apply our \BBD algorithm on each instance, and combining the results using the following two lemmas.

\begin{restatable}{lemma}{fedpatchtwoLRBBD}
    \label{lem:fed_patch_two_lr}
    Suppose we are given a $(\w, \bF, {\bF'})$-\LRBBD
    instance of size $(n, n')$ such that $\bF = \bF_1 + \bF_2$.

    Then,
    given the outputs of the $(\w, {\bF_1}, {\bF'})$-\LRBBD instance
    and the $(\w, {\bF_2}, {\bF'})$-\LRBBD instance,
    we can solve the $(\w, \bF, {\bF'})$-\LRBBD instance
    in time $\Oh(\sT_{\MUL}(n', n', n))$.

    \unbalanced{Furthermore, if the \LRBBD instance has diameter $\Oh(n')$, we can solve the $(\w, \bF, {\bF'})$-\LRBBD instance
    in time $\Oh(\sT_{\MonMUL}(n', n', n, n'))$.}
\end{restatable}

\begin{restatable}{lemma}{fedpatchthreeLRBBD}
    \label{lem:fed_patch_three_lr}
    Suppose we are given a $(\w, \bF, {\bF'})$-\LRBBD
    instance of size $(n, n')$ and a synchronous subforest $\bF_s \subseteq \bF$
    of size $\abs{\bF_s} = n_s$.
    Further, let $v_1, \ldots, v_{\abs{\bF}}$ be the pre-order of $\bF$,
    and let $i_s$ be such that $v_{i_s}, \ldots, v_{i_s+n_s}$
    is the pre-order of $\bF_s$.
    Define $\bF_{\ell} = \bF\fragmentco{1}{\Left(v_{i_s})}$, $\bF_s = \bF\fragmentco{\Left(v_{i_s})}{\Left(v_{i_s+n_s+1})}$,
    and $\bF_{r} = \bF\fragmentco{\Left(v_{i_s+n_s+1})}{2\abs{\bF}+1}$.

    Then,
    given the outputs of the four \LRBBD instances $(\w, \bF_{\ell}, {\bF'})$-\LRBBD,
    $(\w, {\bF_s}, {\bF'})$-\LRBBD, $(\w, {\bF_{r}}, {\bF'})$-\LRBBD,
    and $(\w, {\bF \setminus \bF_{s}}, {\bF'})$-\LRBBD,
    we can solve the $(\w, \bF, {\bF'})$-\LRBBD instance
    in time $\Oh(\sT_{\MUL}(n', n', n))$.

    \unbalanced{Furthermore, if the \LRBBD instance has diameter $\Oh(n')$, we can solve the $(\w, \bF, {\bF'})$-\LRBBD instance
    in time $\Oh(\sT_{\MonMUL}(n', n', n, n'))$.}
\end{restatable}

The decomposition scheme of \Cref{alg:mao_decomp}, together with \cref{lem:fed_patch_two_lr} and \cref{lem:fed_patch_three_lr},
allows us to establish the following reduction from a
$(\w, \bF, {\bF'})$-\LRBBD instance of size $(n, n')$
to $\Oh(n/\Delta)$ \LRBBD instances, each of size at most $(\Delta, n')$ for any threshold $\Delta$.

\begin{lemma}\label{lem:fed_decomp_lr}
    Suppose, we are given a $(\w, \bF, {\bF'})$-\LRBBD instance of size $(n, n')$ and threshold $\Delta$.
    Assume without loss of generality $n \geq n'$.

    Then, there exist forests $\bF_{1}, \ldots, \bF_{k}$ all of size at most $\Delta$, such that $k = \Oh(n/\Delta)$ and such that,
    given the output of the $(\w, {\bF_i}, {\bF'})$-\LRBBD instance
    for all $i \in \fragment{1}{k}$,
    we can solve the $(\w, \bF, {\bF'})$-\LRBBD instance
    in time $\Oh(\frac{n^2}{\Delta n'} \cdot \sT_{\MUL}(n')$.

    \unbalanced{Furthermore, if the \LRBBD instance has diameter $\Oh(n')$, we can solve the $(\w, \bF, {\bF'})$-\LRBBD instance
    in time $\Oh(\frac{n^2}{\Delta n'} \cdot \sT_{\MonMUL}(n'))$.}
\end{lemma}

\begin{proof}
    As in \cref{lem:fed_decomp}, we utilize \cref{alg:mao_decomp} on the $(\w, {\bF}, {\bF'})$-\LRBBD instance.
    \cref{alg:mao_decomp} applies recursively Mao's transitions with the same parameter $\Delta$ on forests $\bH \subseteq \bF$,
    until it is left with forests of size at most $(\Delta, n')$.
    We use \Cref{lem:fed_patch_two_lr} and \Cref{lem:fed_patch_three_lr} to combine results as appropriate.
    The correctness of the algorithm follows directly from \cref{lem:fed_patch_two_lr} and \cref{lem:fed_patch_three_lr}.

    What remains to be demonstrated is a bound on $k$ and on the running time.
    For that purpose, notice that $k = \Oh(t_{\mathrm{I}} + t_{\mathrm{II}})$ where $t_{\mathrm{I}}, t_{\mathrm{II}}$ is the number of times we apply transition of type~\ref{it:decomp:1} and~\ref{it:decomp:2}, respectively.
    As argued in \cref{lem:fed_decomp}, we have $k = \Oh(t_{\mathrm{I}} + t_{\mathrm{II}}) = \Oh(n/\Delta)$.
    Then, the running time can be bounded by
    \begin{equation*}
        \Oh\left(\frac{n}{\Delta} \cdot \sT_{\MUL}(n',n', n)\right) = \Oh\left(\frac{n^2}{\Delta n'} \cdot \sT_{\MUL}(n')\right).
    \end{equation*}

    \unbalanced{If the diameter is at most $\Oh(n')$, we instead obtain
    \begin{equation*}
        \Oh\left(\frac{n}{\Delta} \cdot \sT_{\MonMUL}(n', n', n, n')\right) = \Oh\left(\frac{n^2}{\Delta n'} \cdot \sT_{\MonMUL}(n') \right)
    \end{equation*}}
\end{proof}

\begin{corollary} 
    \label{cor:lrbbd_algo}
    There is an algorithm that solves in time $(n/n')^{1+o(1)} \cdot \left( \sT_{\MUL}(n') + n'^{2 + o(1)} \right)$
    a $(\w, {\bF}, {\bF'})$-\LRBBD instance of size $(n, n')$ where $n \geq n'$.

    \unbalanced{Furthermore, if the \LRBBD instance has diameter $\Oh(n')$, then the algorithm runs in time $(n/n')^{1+o(1)} \cdot \left( \sT_{\MonMUL}(n') + n'^{2 + o(1)} g(n') \right)$ where $\sT_{\MonMUL}(n', n', n', \D) = \Oh(f(n') g(\D))$.}
\end{corollary}

\begin{proof}
    We apply \cref{lem:fed_decomp} recursively with threshold $\Delta = n/\alpha$ for some constant $\alpha \geq 1$ on forest $\bF$.
    As before, for small enough instances when $n = O(n')$ we can simply apply our algorithm for \BBD and compute the output in $\Oh(\sT_{\MUL}(n'))$ time.
    Thereby, we obtain an algorithm for the \LRBBD Problem where the running time is described by the formula
    \begin{align*}
        \sT(n, n') = \Oh\left( \frac{n}{\Delta} \right) \sT(\Delta, n') + \Oh \left( \frac{n^2}{\Delta n'} \sT_{\MUL}(n') \right).
    \end{align*}
    Note that the recurrence terminates at $\sT(n', n') = \Oh\left(\sT_{\MonMUL}(n') + n'^{2 + o(1)}\right)$.
    By our choice of $\Delta$, there exists a constant $C$ such that
    \begin{align*}
        \sT(n, n') &\le \left( C \alpha \right) \sT(n/\alpha, n') + \left( C \alpha \frac{n}{n'} \right) \sT_{\MUL}(n').
    \end{align*}
    We choose a constant alpha such that $\log_{\alpha} (C \alpha) < 1 + \epsilon$ for arbitrarily small $\epsilon > 0$ so that $\sT(n, n') = (n/n')^{1+o(1)} \cdot \left( \sT_{\MUL}(n') + n'^{2 + o(1)}\right)$. 

    \unbalanced{When the diameter is at most $\Oh(n')$, we instead have the recurrence for some constant $C$,
    \begin{align*}
        \sT(n, n') &\le \left( C \alpha \right) \sT(n/\alpha, n') + \left( C \alpha \frac{n}{n'} \right) \sT_{\MonMUL}(n').
    \end{align*}
    again with the recurrence terminating at $\sT(n', n') = \Oh(\sT_{\MonMUL}(n') + n'^{2 + o(1)} g(n'))$. 
    As before, we choose a constant $\alpha$ such that the time can be bounded by $\sT(n, n') = (n/n')^{1+o(1)} \cdot \left( \sT_{\MonMUL}(n') + n'^{2 + o(1)} g(n') \right)$.}

\end{proof}

Finally, to conclude, we give an algorithm computing certain outputs of the \FED problem efficiently, defined as the \UFED problem.

\defproblem
{Unbalanced Forest Edit Distance (\UFED)}
{Two forests $\bF$ and $\bF'$ and $\similarity(\sub(v), \sub(v'))$ for all $(v, v') \in \bF \times \bF'$.}
{The following values:
\begin{enumerate}
\item $\similarity(\bF, \bF'\fragmentco{x'}{y'})$ for all $x',y' \in \fragment{1}{(2|\bF'|+1)}$, and
\label{it:ufed:output1}
\item $\similarity(\bF\fragmentco{1}{y}, \bF'\fragmentco{x'}{(2|\bF'|+1)})$ for all $y \in \fragment{1}{(2|\bF|+1)}, x' \in \fragment{1}{(2|\bF'|+1)}$.
\label{it:ufed:output2}
\item $\similarity(\bF\fragmentco{x}{2|\bF| + 1}, \bF'\fragmentco{1}{y'})$ for all $x \in \fragment{1}{(2|\bF|+1)}, y' \in \fragment{1}{(2|\bF'|+1)}$.
\label{it:ufed:output3}
\end{enumerate}
}

Using an identical proof to \Cref{thm:fed}, we obtain the following result on \UFED, noting that the paths in the \LRBBD instance correspond to the desired outputs. 
In particular, in an \LRBBD instance we compute all distances in the \BBD instance except those originating from the bottom border. 
Then, 

\begin{restatable}{theorem}{FEDUnbalance}
    \label{thm:fed-unbalance}
    There is an algorithm for \UFED running in time $(n/n')^{1+o(1)} \cdot \left(\sT_{\MUL}(n') + n'^{2 + o(1)}\right)$ where $\abs{\bF} = n$, $\abs{\bF'} = n'$ and $n \geq n'$.
\end{restatable}

\UnweightFEDUnbalance

Since the proofs are essentially identical, we give both here.

\begin{proof}
    We apply \Cref{cor:lrbbd_algo} to the generalized alignment graph.
    We consider the three outputs separately.
    \begin{itemize}
        \item For Output~\eqref{it:ufed:output1}, we obtain the desired similarity value by the distance from $(1, x')$ to $(n + 1, y')$.
        \item For Output~\eqref{it:ufed:output2}, we obtain the desired similarity value by the distance from $(1, x')$ to $(y, n' + 1)$.
        \item For Output~\eqref{it:ufed:output3}, we obtain the desired similarity value via reverse symmetry. 
        In particular, we can begin by reversing $\bF$ and $\bF'$ and consider Output~\eqref{it:ufed:output2} in the reversed instance. \qedhere
    \end{itemize}
\end{proof}

\subsubsection{Handling the Recursive Case}
\label{sec:puttogether-unbalance}

\fedpatchtwoLRBBD*

\begin{proof}
    We proceed with a similar argument as in \Cref{lem:fed_patch_two}, using the same notation as before.
    Let $n_1 = \abs{\bF_1}$, $n_2 = \abs{\bF_2}$, $n = \abs{\bF} = n_1 + n_2$, and set $i_2 = n_1+1$.
    Consider the pre-order traveral $v_1, \ldots, v_n$ of $\bF$,
    and observe that $v_1, \ldots, v_{n_1}$ corresponds to the pre-order of $\bF_1$
    and that $v_{i_2}, \ldots, v_{n}$ corresponds to the pre-order traversal of $\bF_2$.
    Further, consider the alignment graph $\bar{\bG}$ w.r.t.~$(\w, {\bF}, {\bF'})$.
    The alignment graph \(\bar{\bG}_1\) w.r.t.~\((\w, {\bF_1}, {\bF'})\)
    corresponds to $\bar{\bG}\fragment{1}{i_2}\fragment{1}{(n'+1)}$,
    and the alignment graph \(\bar{\bG}_2\) w.r.t.~\((\w, {\bF_2}, {\bF'})\)
    corresponds to $\bar{\bG}\fragment{i_2}{(n+1)}\fragment{1}{(n'+1)}$.

    We now show how to compute the outputs of the $(\w, \bF, \bF')$-\LRBBD instance.
    As before, we decompose $\sB^{\top}(1, 1, n+1, n'+1)$ into $I_{1}^{\top}$ and $I_{2}^{\top}$.
    See \cref{fig:patch_two_lr} for a visualization of the alignment graph and the computed distances.

    \begin{figure}[htbp]
        \centering
        \usetikzlibrary{arrows.meta, bending}

\begin{tikzpicture}
    \draw[thick] (0,0) rectangle (6,2);
    \draw[red, line width=7pt, opacity=0.5] (0,2) --node[left=20pt, thick, opacity=1] {$1 \times \fragment{1}{n' + 1}$} (0,0);
    \draw[blue, line width=7pt, opacity=0.5] (6, 0) --node[right=14pt, thick, opacity=1] {$\sB^{\top}(1, 1, n+1, n'+1)$} (6,2) -- (0,2);
    \draw[dashed, thick] (3,0) --node[right] {$i_2$} (3,2);

    \draw[Bracket-Parenthesis, thick] (0,2.3) --node[above] {$I_1^{\top}$} (3,2.3);
    \draw[Bracket-Bracket, thick] (3,2.3) --node[above] {$I_2^{\top}$} (6.3,2.3) -- (6.3,0);
\end{tikzpicture}
        \caption{Computation of left to top and right borders in the \LRBBD instance.
        In \cref{lem:fed_patch_two_lr} we `cut' the rectangle between $\bF_1$ and $\bF_2$,
        and given the answer for the instances corresponding to the two resulting rectangle halves,
        we show how to patch them together for the full rectangle.}
        \label{fig:patch_two_lr}
    \end{figure}
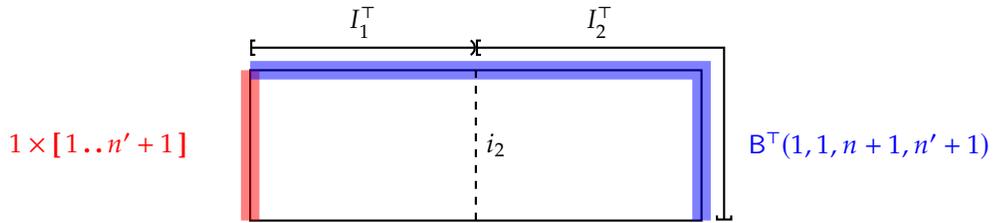
    
    We proceed by case analysis on $\bar{v}$.
    \begin{itemize}
        \item For $\bar{v} \in I_{1}^{\top}$, we obtain the outputs $\dist_{\bar{\bG}}(\bar{u}, \bar{v})$ using the output of the $(\w, {\bF_1}, {\bF'})$-\LRBBD instance.
        \item For $\bar{v} \in I_{2}^{\top}$, notice that for every $i \in \fragmentco{1}{i_2}$ we have $\pi_{\bF}(i) \leq i_2$ as $\bF = \bF_1 + \bF_2$.
        As a consequence, a path from $\bar{u}$ to $\bar{v}$ must pass through a node contained in the set $i_2 \times \fragment{1}{(n'+1)}$, and we can write
        \begin{align*}
            \dist_{\bar{\bG}}(\bar{u}, \bar{v})
            &= \max\nolimits_{\bar{w} \in i_2 \times\fragment{1}{(n'+1)}}
            \Big\{ \ \dist_{\bar{\bG}}(\bar{u}, \bar{w}) + \dist_{\bar{\bG}}(\bar{w}, \bar{v}) \ \Big\}.
        \end{align*}
        As before, the two summands that appear in the latter maximization are a subset of the output of the $(\w, {\bF_1}, {\bF'})$-\LRBBD instance
        and of the $(\w, {\bF_2}, {\bF'})$-\LRBBD instance.
        Moreover, note that we may rewrite the computation of $\dist_{\bar{\bG}}(\bar{u}, \bar{v})$ for all such $\bar{u}$ and $\bar{v}$
        as a max-plus product $A = B \star C$, where $B$ is a matrix of size
        $(n'+1) \times (n'+1)$
        and $C$ is a matrix of size
        $(n'+1) \times (n_2 + n' + 1)$
        Thus, we spend at most $\Oh(\sT_{\MUL}(n', n', n))$ time in the computation of these distances.

        \unbalanced{As argued in \Cref{lem:fed_patch_two}, the matrix $C$ is a column-monotone matrix with entries in both matrices bounded by $\Oh(n')$ if the diameter is bounded by $O(n')$, so that we can compute the required distances in time $\Oh(\sT_{\MonMUL}(n', n', n, n'))$.}
        \qedhere
    \end{itemize}
\end{proof}

\fedpatchthreeLRBBD*

\begin{proof}
    Let $\abs{\bF} = n$, and define $i_r = i_s + n_s + 1$.
    Consider the alignment graph $\bar{\bG}$ w.r.t.~$(\w, \bF, {\bF'})$,
    and observe that $\bar{\bG}\fragment{1}{i_s}\fragment{1}{n'}$ corresponds to the alignment graph
    $\bar{\bG}_{\ell}$ w.r.t.~$(\w, {\bF_{\ell}}, {\bF'})$,
    $\bar{\bG}\fragment{i_s}{i_r}\fragment{1}{n'}$ corresponds to the alignment graph
    $\bar{\bG}_{s}$ w.r.t.~$(\w, {\bF_{s}}, {\bF'})$, and
    $\bar{\bG}\fragment{i_r}{n+1}\fragment{1}{n'}$ corresponds to the alignment graph
    $\bar{\bG}_{r}$ w.r.t.~$(\w, {\bF_{r}}, {\bF'})$.
    The alignment graph $\bar{\bG}_{\setminus s}$ w.r.t.~$(\w, {\bF \setminus \bF_s}, {\bF'})$,
    corresponds to the graph obtained by removing from $\bar{\bG}$ all nodes (and their incident edges)
    of the set $\fragmentco{i_s}{i_r} \times \fragment{1}{(n'+1)}$, and by adding an edge from $(x, i_s-1)$ to $(x, i_r)$
    of cost $0$ for all $x \in \fragment{1}{(n'+1)}$.
    We abuse notation, and we index nodes in $\bar{\bG}_{\setminus s}$ w.r.t.~the corresponding positions in $\bar{\bG}$, taking care of never
    using any index from the set $\fragmentco{i_s}{i_r} \times \fragment{1}{(n'+1)}$.

    As before, we decompose $\sB^{\top}(1, 1, n + 1, n' + 1)$ into $I_{\ell}^{\top}, I_{s}^{\top}, I_{r}^{\top}$.
    See \Cref{fig:patch_three_lr} for a visualization of the alignment graph and computed instances.
    \begin{figure}[htbp]
        \centering
        \usetikzlibrary{arrows.meta, bending}

\begin{tikzpicture}
    \draw[thick] (0,0) rectangle (8,2);
    \draw[red, line width=7pt, opacity=0.5] (0,2) --node[left=14pt, thick, opacity=1] {$1 \times \fragment{1}{n' + 1}$} (0,0);
    \draw[blue, line width=7pt, opacity=0.5] (8, 0) --node[right=14pt, thick, opacity=1] {$\sB^{\top}(1, 1, n+1, n'+1)$} (8,2) -- (0,2);
    \draw[dashed, thick] (3,0) --node[right] {$i_{s}$} (3,2);
    \draw[dashed, thick] (5,0) --node[right] {$i_{r}$} (5,2);

    \draw[Bracket-Parenthesis, thick] (0,2.3) --node[above] {$I_{\ell}^{\top}$} (3,2.3);
    \draw[Bracket-Parenthesis, thick] (3,2.3) --node[above] {$I_{s}^{\top}$} (5,2.3);
    \draw[Bracket-Bracket, thick] (5,2.3) --node[above] {$I_{r}^{\top}$} (8.3,2.3) -- (8.3,0);

\end{tikzpicture}
        \caption{Computation of distances from the left border to top and right borders.
        }
        \label{fig:patch_three_lr}
    \end{figure}
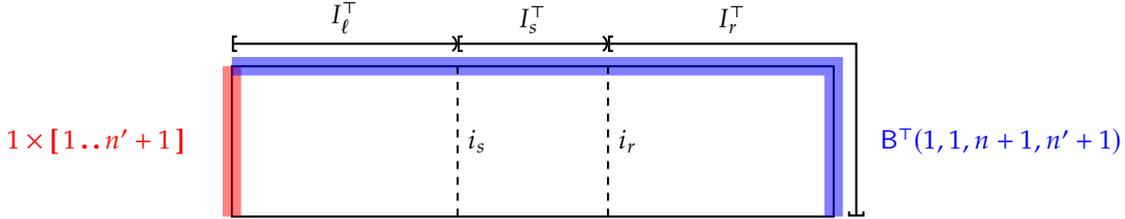

    Our goal is to compute $\dist_{\bar{\bG}}(\bar{u}, \bar{v})$ for all $\bar{u}, \bar{v}$ where $\bar{u} \in \{1\} \times \fragment{1}{n' + 1}$ and $\bar{v} \in \sB^{\top}(1, 1, n + 1, n' + 1)$.
    We proceed by case analysis on $\bar{v}$. 
    \begin{itemize}
        \item For $\bar{v} \in I_{\ell}^{\top}$, we obtain the outputs $\dist_{\bar{\bG}}(\bar{u}, \bar{v})$ using the output of the $(w, \bF_{\ell}, \bF')$-\LRBBD instance.
        \item For $\bar{v} \in I^{\top}_{s}$, observe that for any $i \in \fragmentco{1}{i_s}$ either $\pi_{\bF}(i) \leq i_s$ or $\pi_{\bF}(i) \geq i_r > i_s + n_s$ holds as $\bF_s$ is a synchronous subforest of $\bF$.
        Hence, a path in $\bar{\bG}$ that goes from $\bar{u}$ to $\bar{v}$ must pass through a node in $i_s \times \fragment{1}{(n'+1)}$,
        and we can write
        \begin{align*}
            \dist_{\bar{\bG}}(\bar{u}, \bar{v})
            &= \max\nolimits_{\bar{w} \in i_s \times\fragment{1}{(n'+1)}}
            \Big\{ \ \dist_{\bar{\bG}}(\bar{u}, \bar{w}) + \dist_{\bar{\bG}}(\bar{w}, \bar{v}) \ \Big\}.
        \end{align*}
        Similarly to \cref{lem:fed_patch_two}, we can rewrite the computation of $\dist_{\bar{\bG}}(\bar{u}, \bar{v})$ for all such $\bar{u}$ and $\bar{v}$
        using a max-plus product. This time, we use
        the output of the $(\w, {\bF_{\ell}}, {\bF'})$-\LRBBD instance and the
        output of the $(\w, {\bF_{s}}, {\bF'})$-\LRBBD instance to compute the entries of the matrices
        since $\bar{\bG}\fragment{1}{i_s}\fragment{1}{(n'+1)} = \bar{\bG}_{\ell}$ and
        $\bar{\bG}\fragment{i_s}{i_r}\fragment{1}{(n'+1)} = \bar{\bG}_{s}$.
        Note that the matrices have shape $(n' + 1) \times (n' + 1)$ and $(n' + 1) \times (n_s + 1)$ so that the multiplication requires $\Oh(\sT_{\MUL}(n', n', n))$ time.

        \item For $\bar{v} \in I^{\top}_{r}$, we
        use again that either $\pi_{\bF}(i) \leq i_s$ or $\pi_{\bF}(i) \geq i_r > i_s + n_s$
        for $i \in \fragmentco{1}{i_s}$, and that $\pi_{\bF}(i) \leq i_r$ for $i \in \fragmentco{i_s}{i_r}$.
        Consequently, any path from $\bar{u}$ to $\bar{v}$ either stays entirely in $\bar{\bG}_{\setminus s}$, or
        passes through a node in $i_s \times \fragment{1}{(n'+1)}$ and a node in $i_r \times \fragment{1}{(n'+1)}$.
        Therefore, we obtain
        \begin{align*}
            \dist_{\bar{\bG}}(\bar{u}, \bar{v}) =
            \max \Biggl\{
                \ \dist_{\bar{\bG}_{\setminus s}}(\bar{u}, \bar{v}) \ ,
                \max_{\substack{\bar{w} \in i_s \times\fragment{1}{(n'+1)}\\\bar{z} \in i_r \times\fragment{1}{(n'+1)}}}
                \Big\{ \
                    \dist_{\bar{\bG}}(\bar{u}, \bar{w})
                    + \dist_{\bar{\bG}}(\bar{w}, \bar{z})
                    + \dist_{\bar{\bG}}(\bar{z}, \bar{v})
                \ \Big\} \
            \Biggl\}.
        \end{align*}
        These computation can be done by taking the maximum between outputs of the $(\w, {\bF_{\setminus s}}, {\bF'})$-\LRBBD instance,
        and a max-plus product between three matrices. The entries of these matrices can be
        retrieve from outputs of the other three \LRBBD instances at our disposal.
        Note that the matrices have shape $(n' + 1) \times (n' + 1), (n' + 1) \times (n' + 1)$ and $(n' + 1) \times (n_r + 1)$ so that the multiplications require $\Oh(\sT_{\MUL}(n', n', n))$ time.
    \end{itemize}

    \unbalanced{As before, if the diameter is at most $\Oh(n')$, we have that the matrices are monotone and have entries bounded by $\Oh(n')$.
    Thus, we may compute the necessary matrix products in time $\Oh(\sT_{\MonMUL}(n', n', n, n'))$.}
\end{proof}

\section{Reduction from Spine Edit Distance to \APSP}
\label{sec:sed}

In this section, we focus on proving \cref{thm:sed}.

\sed*

Similar to the previous one,
we assume that there exists an algorithm computing
the min/max-plus product of two $m \times m$ matrices in time $\sT_{\MUL}(m)$.
\unbalanced{We use $\sT_{\MUL}(a, b, c)$ to the denote the time required to multiple an $a \times b$ matrix and a $b \times c$ matrix. 
As in \Cref{sec:fed}, we describe our algorithm for general weighted tree edit distance instances, describing modifications for the unweighted case where necessary.
To this end, recall that we use $\sT_{\MonMUL}(a, b, c, d)$ to denote the time required to multiply an $a \times b$ matrix with a $b \times c$ matrix where both matrices have entries bounded by $d$ and the latter matrix is row or column-monotone.}

To formulate an algorithm for \SED, we assume that $\bF$ and $\bF'$ are trees. If they are not, we add virtual roots $t$ and $t'$, each labeled with a unique symbol, and include these in $\bS$ and $\bS'$ respectively. This allows us to assume that the first nodes in $\bS$ and $\bS'$ correspond to the first nodes in the pre-order traversals of $\bF$ and $\bF'$, respectively. For such new roots, we set $\similarity(\sub(t),\sub(v')) = -\infty$ for all $v' \in \bF' \setminus \{t'\}$, even in the unweighted case, and similarly for $\similarity(\sub(v),\sub(t'))$.
This adjustment will not affect the algorithm we discuss, as these similarities only appear in the first column and row of some matrices we will work with. It is not difficult to see that we can compute the min/max-plus product of such matrices, in the same time as there would not be such column / row. Moreover, note that the solutions for \SED on the original forests form a subset of the solutions on the new trees.

\begin{definition}
    Given two forests $\bF, \bF'$ and four nodes $s,q \in \bS, s',q' \in \bS'$ such that $s \prec q$ and $s' \prec q'$, set
    \begin{align*}
        &\sB^{\bot}_{\Left}(s,s',q,q') \coloneqq \sB^{\bot}(\Left(s), \Left(s'), \Left(q), \Left(q')),
        &&\sB^{\top}_{\Left}(s,s',q,q') \coloneqq \sB^{\top}(\Left(s), \Left(s'), \Left(q), \Left(q')), \\
        &\sB^{\bot}_{\Right}(s,s',q,q') \coloneqq \sB^{\top}(\Right(q), \Right(q'), \Right(s), \Right(s')),
        &&\sB^{\top}_{\Right}(s,s',q,q') \coloneqq \sB^{\bot}(\Right(q), \Right(q'), \Right(s), \Right(s')). \qedhere
    \end{align*}
\end{definition}

We can interpret these sets as follows. Consider the set  
$\fragment{1}{(2\abs{\bF}+1)} \times \fragment{1}{(2\abs{\bF'}+1)}$. The sets $\sB^{\bot}_{\Left}(s,s',q,q')$ and $\sB^{\top}_{\Left}(s,s',q,q')$ are the lower left and upper right border of the subrectangle $\fragment{\Left(s)}{\Left(q)} \times \fragment{\Left(s')}{\Left(q')}$.
On the other hand, $\sB^{\bot}_{\Right}(s,s',q,q')$ and $\sB^{\top}_{\Right}(s,s',q,q')$
are the upper right and lower left border of the subrectangle $\fragment{\Right(q)}{\Right(s)} \times \fragment{\Right(q')}{\Right(s')}$.

This allows us to define the following divide-et-impera scheme that we will use to solve \SED.

\defproblemwlist
{Divide-et-Impera Spine Edit Distance (\DISED)}
{$s \prec q \in \bS$, $s' \prec q' \in \bS'$ and
\begin{enumerate}[(i)]
    \item $\similarity(\bF\fragmentco{x}{y}, \bF'\fragmentco{x'}{y'})$ for all $(x, x') \in \sB^{\top}_{\Left}(s,s',q,q')$ and $(y, y') \in \sB^{\top}_{\Right}(s,s',q,q')$.
    \label{it:sed:input1}
    \item $\similarity(\bF\fragmentco{x}{y}, \sub(q'))$ for all $x,y \in \fragment{\Left(s)}{\Left(q)}$.
    \label{it:sed:input2}
    \item $\similarity(\bF\fragmentco{x}{y}, \sub(q'))$ for all $x,y \in \fragment{\Right(q)}{\Right(s)}$.
    \label{it:sed:input3}
    \item $\similarity(\sub(q), \bF'\fragmentco{x'}{y'})$ for all $x',y' \in \fragment{\Left(s')}{\Left(q')}$.
    \label{it:sed:input4}
    \item $\similarity(\sub(q), \bF'\fragmentco{x'}{y'})$ for all $x',y' \in \fragment{\Right(q')}{\Right(s')}$.
    \label{it:sed:input5}
\end{enumerate}}{
The values
\begin{enumerate}[(i)]
    \item $\similarity(\bF\fragmentco{x}{y}, \bF'\fragmentco{x'}{y'})$ for all $(x, x') \in \sB^{\bot}_{\Left}(s,s',q,q')$ and $(y, y') \in \sB^{\bot}_{\Right}(s,s',q,q')$.
    \label{it:sed:output}
\end{enumerate}
}

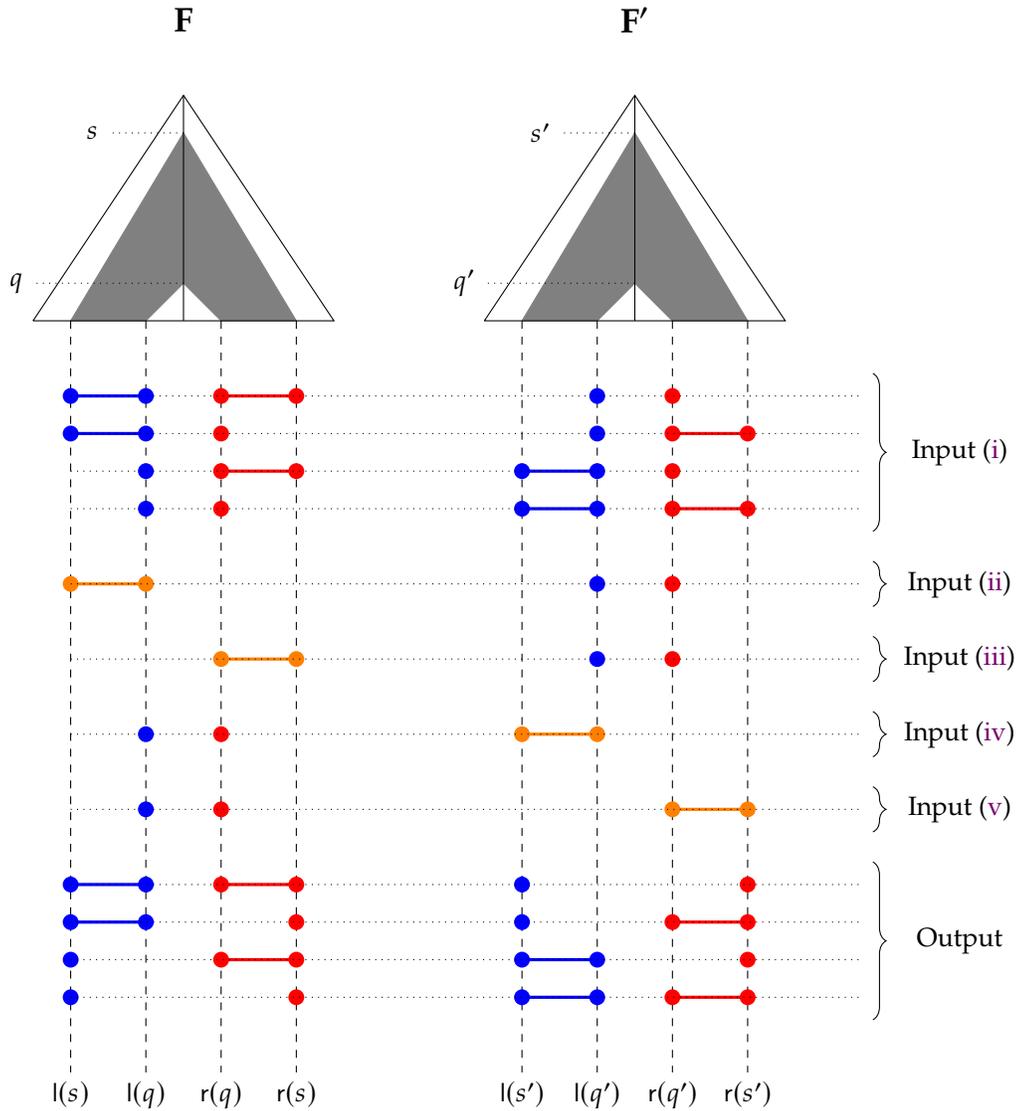
\begin{figure}[htbp]
    \centering
    \usetikzlibrary{arrows.meta, bending}

\begin{tikzpicture}
    \draw (0,0) -- (4,0) -- (2, 3) -- (0, 0);
    \draw (2, 3) -- (2, 0);
    \draw[fill, opacity=0.5] (0.5, 0) -- (2,2.5) -- (3.5,0) -- (2.5, 0) -- (2,0.5) -- (1.5,0) -- (0.5, 0);
    \draw[dotted] (2,2.5) -- (1, 2.5) node[left] {$s$};
    \draw[dotted] (2,0.5) -- (0, 0.5) node[left] {$q$};

    \draw[dashed] (0.5, 0) -- (0.5, -10) node[below] {$\Left(s)$};
    \draw[dashed] (1.5, 0) -- (1.5, -10) node[below] {$\Left(q)$};
    \draw[dashed] (2.5, 0) -- (2.5, -10) node[below] {$\Right(q)$};
    \draw[dashed] (3.5, 0) -- (3.5, -10) node[below] {$\Right(s)$};

    \node[scale=1.3] at (2, 4) {$\bF$};

    \fill[color=blue] (0.5,-1) circle (3pt);
    \draw[color=blue, very thick] (0.5,-1) -- (1.5,-1);
    \fill[color=blue] (1.5,-1) circle (3pt);
    \fill[color=red] (2.5,-1) circle (3pt);
    \draw[color=red, very thick] (2.5,-1) -- (3.5,-1);
    \fill[color=red] (3.5,-1) circle (3pt);

    \fill[color=blue] (0.5,-1.5) circle (3pt);
    \draw[color=blue, very thick] (0.5,-1.5) -- (1.5,-1.5);
    \fill[color=blue] (1.5,-1.5) circle (3pt);
    \fill[color=red] (2.5,-1.5) circle (3pt);

    \fill[color=blue] (1.5,-2) circle (3pt);
    \fill[color=red] (2.5,-2) circle (3pt);
    \draw[color=red, very thick] (2.5,-2) -- (3.5,-2);
    \fill[color=red] (3.5,-2) circle (3pt);

    \fill[color=blue] (1.5,-2.5) circle (3pt);
    \fill[color=red] (2.5,-2.5) circle (3pt);

    \fill[color=orange] (0.5,-3.5) circle (3pt);
    \draw[color=orange, very thick] (0.5,-3.5) -- (1.5,-3.5);
    \fill[color=orange] (1.5,-3.5) circle (3pt);

    \fill[color=orange] (2.5,-4.5) circle (3pt);
    \draw[color=orange, very thick] (2.5,-4.5) -- (3.5,-4.5);
    \fill[color=orange] (3.5,-4.5) circle (3pt);

    \fill[color=blue] (1.5,-5.5) circle (3pt);
    \fill[color=red] (2.5,-5.5) circle (3pt);

     \fill[color=blue] (1.5,-6.5) circle (3pt);
     \fill[color=red] (2.5,-6.5) circle (3pt);

    \fill[color=blue] (0.5,-7.5) circle (3pt);
    \draw[color=blue, very thick] (0.5,-7.5) -- (1.5,-7.5);
    \fill[color=blue] (1.5,-7.5) circle (3pt);
    \fill[color=red] (2.5,-7.5) circle (3pt);
    \draw[color=red, very thick] (2.5,-7.5) -- (3.5,-7.5);
    \fill[color=red] (3.5,-7.5) circle (3pt);

    \fill[color=blue] (0.5,-8) circle (3pt);
    \draw[color=blue, very thick] (0.5,-8) -- (1.5,-8);
    \fill[color=blue] (1.5,-8) circle (3pt);
    \fill[color=red] (3.5,-8) circle (3pt);

    \fill[color=blue] (0.5,-8.5) circle (3pt);
    \fill[color=red] (2.5,-8.5) circle (3pt);
    \draw[color=red, very thick] (2.5,-8.5) -- (3.5,-8.5);
    \fill[color=red] (3.5,-8.5) circle (3pt);

    \fill[color=blue] (0.5,-9) circle (3pt);
    \fill[color=red] (3.5,-9) circle (3pt);

    \begin{scope}[shift={(6, 0)}]
        \draw (0,0) -- (4,0) -- (2, 3) -- (0, 0);
        \draw (2, 3) -- (2, 0);
        \draw[fill, opacity=0.5] (0.5, 0) -- (2,2.5) -- (3.5,0) -- (2.5, 0) -- (2,0.5) -- (1.5,0) -- (0.5, 0);
        \draw[dotted] (2,2.5) -- (1, 2.5) node[left] {$s'$};
        \draw[dotted] (2,0.5) -- (0, 0.5) node[left] {$q'$};

        \draw[dashed] (0.5, 0) -- (0.5, -10) node[below] {$\Left(s')$};
        \draw[dashed] (1.5, 0) -- (1.5, -10) node[below] {$\Left(q')$};
        \draw[dashed] (2.5, 0) -- (2.5, -10) node[below] {$\Right(q')$};
        \draw[dashed] (3.5, 0) -- (3.5, -10) node[below] {$\Right(s')$};

        \node[scale=1.3] at (2, 4) {$\bF'$};

        \fill[color=blue] (1.5,-1) circle (3pt);
        \fill[color=red] (2.5,-1) circle (3pt);

        \fill[color=blue] (1.5,-1.5) circle (3pt);
        \fill[color=red] (2.5,-1.5) circle (3pt);
        \draw[color=red, very thick] (2.5,-1.5) -- (3.5,-1.5);
        \fill[color=red] (3.5,-1.5) circle (3pt);

        \fill[color=blue] (0.5,-2) circle (3pt);
        \draw[color=blue, very thick] (0.5,-2) -- (1.5,-2);
        \fill[color=blue] (1.5,-2) circle (3pt);
        \fill[color=red] (2.5,-2) circle (3pt);

        \fill[color=blue] (0.5,-2.5) circle (3pt);
        \draw[color=blue, very thick] (0.5,-2.5) -- (1.5,-2.5);
        \fill[color=blue] (1.5,-2.5) circle (3pt);
        \fill[color=red] (2.5,-2.5) circle (3pt);
        \draw[color=red, very thick] (2.5,-2.5) -- (3.5,-2.5);
        \fill[color=red] (3.5,-2.5) circle (3pt);

        \fill[color=blue] (1.5,-3.5) circle (3pt);
        \fill[color=red] (2.5,-3.5) circle (3pt);

        \fill[color=blue] (1.5,-4.5) circle (3pt);
        \fill[color=red] (2.5,-4.5) circle (3pt);

        \fill[color=orange] (0.5,-5.5) circle (3pt);
        \draw[color=orange, very thick] (0.5,-5.5) -- (1.5,-5.5);
        \fill[color=orange] (1.5,-5.5) circle (3pt);

        \fill[color=orange] (2.5,-6.5) circle (3pt);
        \draw[color=orange, very thick] (2.5,-6.5) -- (3.5,-6.5);
        \fill[color=orange] (3.5,-6.5) circle (3pt);

        \fill[color=blue] (0.5,-9) circle (3pt);
        \draw[color=blue, very thick] (0.5,-9) -- (1.5,-9);
        \fill[color=blue] (1.5,-9) circle (3pt);
        \fill[color=red] (2.5,-9) circle (3pt);
        \draw[color=red, very thick] (2.5,-9) -- (3.5,-9);
        \fill[color=red] (3.5,-9) circle (3pt);

        \fill[color=blue] (0.5,-8.5) circle (3pt);
        \draw[color=blue, very thick] (0.5,-8.5) -- (1.5,-8.5);
        \fill[color=blue] (1.5,-8.5) circle (3pt);
        \fill[color=red] (3.5,-8.5) circle (3pt);

        \fill[color=blue] (0.5,-8) circle (3pt);
        \fill[color=red] (2.5,-8) circle (3pt);
        \draw[color=red, very thick] (2.5,-8) -- (3.5,-8);
        \fill[color=red] (3.5,-8) circle (3pt);

        \fill[color=blue] (0.5,-7.5) circle (3pt);
        \fill[color=red] (3.5,-7.5) circle (3pt);
    \end{scope}

    \draw[dotted] (0.5, -1) -- (11, -1);
    \draw[dotted] (0.5, -1.5) -- (11, -1.5);
    \draw[dotted] (0.5, -2) -- (11, -2);
    \draw[dotted] (0.5, -2.5) -- (11, -2.5);
    \draw[decorate, decoration={brace, amplitude=5pt, raise=1ex}]
    (11, -0.7) -- (11,-2.8) node[midway, xshift=8ex] {Input~\eqref{it:sed:input1}};

    \draw[dotted] (0.5, -3.5) -- (11, -3.5);
    \draw[decorate, decoration={brace, amplitude=5pt, raise=1ex}]
    (11, -3.2) -- (11,-3.8) node[midway, xshift=8ex] {Input~\eqref{it:sed:input2}};

    \draw[dotted] (0.5, -4.5) -- (11, -4.5);
    \draw[decorate, decoration={brace, amplitude=5pt, raise=1ex}]
    (11, -4.2) -- (11,-4.8) node[midway, xshift=8ex] {Input~\eqref{it:sed:input3}};

    \draw[dotted] (0.5, -5.5) -- (11, -5.5);
    \draw[decorate, decoration={brace, amplitude=5pt, raise=1ex}]
    (11, -5.2) -- (11,-5.8) node[midway, xshift=8ex] {Input~\eqref{it:sed:input4}};

    \draw[dotted] (0.5, -6.5) -- (11, -6.5);
    \draw[decorate, decoration={brace, amplitude=5pt, raise=1ex}]
    (11, -6.2) -- (11,-6.8) node[midway, xshift=8ex] {Input~\eqref{it:sed:input5}};

    \draw[dotted] (0.5, -7.5) -- (11, -7.5);
    \draw[dotted] (0.5, -8) -- (11, -8);
    \draw[dotted] (0.5, -8.5) -- (11, -8.5);
    \draw[dotted] (0.5, -9) -- (11, -9);
    \draw[decorate, decoration={brace, amplitude=5pt, raise=1ex}]
    (11, -7.2) -- (11,-9.3) node[midway, xshift=8ex] {Output};
\end{tikzpicture}
    \caption{The figure visualizes inputs and outputs of a \DISED instance.
    Such inputs and outputs consist of the values $\similarity(\bF\fragmentco{x}{y}, \bF'\fragmentco{x'}{y'})$
    for certain ranges of $x,y,x',y'$. Every dotted horizontal line shows a combination of such ranges.
    Blue color indicates that the range concerns $x,x'$, red color that it concerns $y,y'$,
    and orange color that on it lie both $x,y$ and $x',y'$.
    The shaded gray areas in $\bF$ and $\bF'$ span the nodes contained in $\sub(s) \setminus \sub(q)$
    and $\sub(s') \setminus \sub(q')$, respectively.
    The maximum number of nodes contained in any such set determines the size of the \DISED instance.
    }
    \label{fig:dised}
\end{figure}

We write $(s,s',q,q')$\text{-}\DISED to denote an instance of the \DISED Problem with input $s \prec q \in \bS$ and $s' \prec q' \in \bS'$.

Given an $(s,s',q,q')$\text{-}\DISED instance, we define its \emph{size} as
\[
        \max \Big\{ \
        \abs{\ \sub(s) \setminus \sub(q) \ } \ ,\
        \abs{\ \sub(s') \setminus \sub(q')\ }
        \ \Big\},
\]
and \emph{global tree size} as
\[
        \min \Big\{ \
        \abs{\ \sub(s) \ } \ ,\
        \abs{\ \sub(s') \ }
        \ \Big\}.
\]
We make two further remarks about an  $(s,s',q,q')$\text{-}\DISED instance.
First, input and output of an $(s,s',q,q')$\text{-}\DISED of size $m$ are of order $\Oh(m^2)$.
Second, there exist two types of symmetry in the \DISED Problem:
\begin{enumerate}[(a)]
    \item swapping the role of $\bF$ and $\bF'$ does not change the input/output of the problem, we refer to this symmetry as \emph{swap symmetry}; and
    \label{symmetry:a}
    \item reversing $\bF$ and $\bF'$ does not change the input/output of the problem, we refer to this symmetry as \emph{reverse symmetry}.
    \label{symmetry:b}
\end{enumerate}
These symmetries help simplify computations,
reducing the need for different computations in symmetric cases.

\subsection{Algorithm for \DISED}

In the following (sub)section,
we present a divide-et-impera algorithm $\mA_{\DISED}$ solving \DISED.
The following two lemmas, the proof of which we defer to  \cref{sec:sed_patch} and \cref{sec:sed_basecase},
 motivate such an approach.

\begin{restatable}{lemma}{sedbasecase}\label{lem:sed_basecase}
    Consider an $(s,s',q,q')$\text{-}\DISED instance of size $m$
    such that $s$ immediately precedes $q$ or $s'$ immediately precedes $q'$.
    Then, we can solve the $(s,s',q,q')$\text{-}\DISED instance in time $\Oh(\sT_{\MUL}(m) + m^{2+o(1)})$.

    \unbalanced{Furthermore, if the \DISED instance is unweighted and has global tree size $\D$, then the algorithm takes time $\Oh(\sT_{\MonMUL}(m, m, m, \D) +  m^{2+o(1)} g(\D))$ where $\sT_{\MonMUL}(m, m, m, \D) = \Oh(f(m) g(\D))$ for some functions $f, g$.}
\end{restatable}

\begin{restatable}{lemma}{sedpatch}\label{lem:sed_patch}
    Suppose we are given an $(s,s',q,q')$\text{-}\DISED instance of size $m$, and let
    $r \in \bS$ be such that $s \prec r \prec q$.
    Then, we can reduce the $(s,s',q,q')$\text{-}\DISED instance to the
    $(s,s',r,q')$\text{-}\DISED and $(r,s',q,q')$\text{-}\DISED instances
    in time $\Oh(\sT_{\MUL}(m) + m^{2+o(1)})$.

    \unbalanced{Furthermore, if the \DISED instance is unweighted and has global tree size $\D$, then the reduction takes time $\Oh(\sT_{\MonMUL}(m, m, m, D) +  m^{2+o(1)} g(\D))$ where $\sT_{\MonMUL}(m, m, m, \D) = \Oh(f(m) g(\D))$  for some functions $f, g$.}
\end{restatable}

We can think of \cref{lem:sed_basecase} and \cref{lem:sed_patch}
as tools that allow us to cover the base case and
to divide a larger problem into two smaller ones, respectively.
Note, when we say 'reducing',
we are not merely referring to utilizing the output of the smaller instance to compute the larger instance.
We are also accounting for the time and computation required to compute the input of the smaller instances.
This input may not necessarily be an input of the larger instance but is essential for invoking any subroutine on the smaller instances.
In these reductions, we assume that the output of a smaller instance is available to us as soon as we compute its inputs.

Intuitively, in \cref{lem:sed_patch} we would like to break down an instance $(s,s',q,q')$\text{-}\DISED
into two instances $(s,s',r,q')$\text{-}\DISED and $(r,s',q,q')$\text{-}\DISED
such that $\abs{\sub(s) \setminus \sub(r)} \approx \abs{\sub(r) \setminus \sub(q)} \approx \abs{\sub(s) \setminus \sub(q)} / 2$.
However, since this is not always possible, we need another division scheme that uses \cref{lem:sed_patch} as a subroutine.

\begin{lemma}
    \label{lem:dised_divide}
    Let be given an $(s,s',q,q')$\text{-}\DISED instance of size $m$ and a positive threshold $\Delta$.
    Then, we can find $9 m^2/\Delta^2$ \DISED instances in time $\Oh(m^2)$,
    each of size at most $\Delta$,
    such that given the input of all such instances, we can solve the $(s,s',q,q')$\text{-}\DISED instance in time $\Oh(m^2/\Delta^2 \cdot (\sT_{\MUL}(m) + m^{2 + o(1)}))$.

    \unbalanced{Furthermore, if the \SED instance is unweighted and has global tree size $\D$, the algorithm requires time $\Oh(m^2/\Delta^2 \cdot (\sT_{\MonMUL}(m, m, m, \D) + m^{2 + o(1)} g(\D)))$ where $\sT_{\MonMUL}(m, m, m, \D) = \Oh(f(m) g(\D))$.}
\end{lemma}

\begin{proof}
    We can simplify the proof to demonstrate the following:
    Given an $(s,s',q,q')$\text{-}\DISED instance of size $n$ and a threshold $\Delta$,
    we can find in time $\Oh(m)$ a set $I \subseteq \bS \times \bS$ satisfying all of the following:
    \begin{itemize}
        \item $|I| \leq 3m/\Delta$;
        \item $s \preceq a \prec b \preceq q$ for $(a,b) \in I$;
        \item $\abs{\sub(a) \setminus \sub(b)} \leq \Delta$ for $(a,b) \in I$; and
        \item the $(s,s',q,q')$\text{-}\DISED instance
        can be reduced to the $(a,s',b,q')$\text{-}\DISED instances for $(a,b) \in I$
        in time $\Oh(m/\Delta \cdot (\sT_{\MUL}(m) + m^{2 + o(1)}) + \Oh(m))$.
        \unbalanced{In the unweighted setting, we require time $\Oh( m/\Delta \cdot (\sT_{\MonMUL}(m, m, m, \D) + m^{2 + o(1)} g(\D)) + \Oh(m))$.}
    \end{itemize}
    Using swap symmetry, we can use the simplified version of the lemma first on the $(s,s',q,q')$\text{-}\DISED instance,
    and then apply it on the $(a,s',b,q')$\text{-}\DISED instances for all $(a,b) \in I$,
    thereby proving the original statement of the lemma.

    \begin{algorithm}[t]
        \KwInput{two nodes $s,q \in \bS$ such that $s \prec q$, and a threshold $\Delta$.}
        \KwOutput{a subset $s=r_1 \prec r_2 \prec \cdots \prec r_d = q \in \bS$.}
        $r_1 \leftarrow s$\;
        $i \leftarrow 1$\;
        \While{$r_i \neq q$}{ \label{line:while}
            Define \( r_{i+1} \in \bS \)
            to be the farthest node from \( r_{i} \) on $\bS$
            such that $r_{i} \prec r_{i+1} \preceq q$ and \(  \abs{\sub(r_{i}) \setminus \sub(r_{i+1})} \leq \Delta \).
            If no such $r_{i+1}$ exists, pick $r_{i+1}$ such that $r_{i}$ immediately precedes $r_{i+1}$\;
            $i \leftarrow i + 1$\; \label{line:take}
        }
        \Return $\{r_i\}_i$\;
        \caption{Partition Algorithm} \label{alg:sed_count}
    \end{algorithm}

    To construct set $I$, we employ \cref{alg:sed_count} on $s$ and $q$ with threshold $\Delta$,
    yielding a sequence of spine nodes $s=r_1 \prec r_2 \prec \cdots \prec r_d = q$.
    The algorithm guarantees that for each $i \in \fragmentco{1}{d}$,
    either the $(r_i,s',r_{i+1},q')$\text{-}\DISED instance has size at most $\Delta$, or $r_{i}$ immediately precedes $r_{i+1}$.
    To complete the reduction, we recursively solve $(r_i,s',r_j,q')$\text{-}\DISED instances for $i,j \in \fragment{i}{d}$ with $i < j$.
    This is achieved by selecting an arbitrary $i < k < j$, and applying \cref{lem:sed_patch} to the $(r_i,s',r_k,q')$\text{-}\DISED and $(r_k,s',r_j,q')$\text{-}\DISED instances.
    Recursion stops when $j = i+1$.
    At this point, if $r_i$ immediately precedes $r_{i+1}$, we use \cref{lem:sed_basecase}; otherwise, we add $(r_i, r_{i+1})$ to $I$.

    Therefore, starting with $i=1$ and $j=d$, we achieve the desired reduction of the $(s,s',q,q')$\text{-}\DISED instance
    to $|I| \leq d-1$ \DISED instances in time $\Oh((d-1) \cdot (\sT_{\MUL}(m) + m^{2 + o(1)}))$.
    \unbalanced{In the unweighted setting, the time is $\Oh((d - 1) \cdot (\sT_{\MonMUL}(m, m, m, \D) + m^{2 + o(1)} g(\D))$.}

    To conclude the proof, it suffices to prove $d-1 \leq 2m/\Delta+1 \leq 3m/\Delta$.
    To this end, we analyze \cref{line:take} in \cref{alg:sed_count}.
    Here is an important observation: if $\abs{\sub(r_{i}) \setminus \sub(r_{i+1})} < \Delta$,
    then either $d = i + 1$ or $\abs{\sub(r_{i}) \setminus \sub(r_{i+1})} + \abs{\sub(r_{i+1}) \setminus \sub(r_{i+2})}> \Delta$.
    Consequently, we iterate at most
    \[
        2\abs{\sub(s) \setminus \sub(q)} / \Delta + 1 \leq 2m / \Delta + 1
    \]
    times in \cref{line:while}, obtaining $d - 1 \leq 2m/\Delta+1$.
\end{proof}

\begin{corollary}
    \label{cor:dised}
    There exists an algorithm $\mA_{\DISED}$ solving \DISED
    running in time $\Oh(\sT_{\MUL}(m)+m^{2+o(1)})$ on instances of size $m$.

    \unbalanced{Furthermore, if the \DISED instance is unweighted and has global tree size $\D$, then the algorithm requires time $\Oh(\sT_{\MonMUL}(m, m, m, \D) + m^{2 + o(1)} g(\D))$ where $\sT_{\MonMUL}(m, m, m, \D) = \Oh(f(m) g(\D))$.}
\end{corollary}

\begin{proof}
    To devise such $\mA_{\DISED}$, it suffices to apply recursively \cref{lem:dised_divide}
    with threshold $\Delta(m) = m/\alpha$
    for a constant $\alpha \geq 1$ to be determined later.
    The base case is handled using \cref{lem:sed_basecase}.
    Thus, we obtain the following recurrence for the running time
    \begin{align*}%
        \sT(m) 
        &\leq c_1 \cdot m^2/\Delta^2 \cdot \sT(\Delta) + c_2 \cdot m^2/\Delta^2 \cdot (\sT_{\MUL}(m) + m^{2 + o(1)}) + c_3 \cdot m^{2+o(1)} \\
        &= c_1 \alpha^{2} \cdot \sT(m/\alpha) + c_2 \alpha^{2} \cdot (\sT_{\MUL}(m) + m^{2 + o(1)}) + c_3 \cdot m^{2+o(1)}
    \end{align*}
    where $c_1 = 9$ and we write out the constants $c_2, c_3$ hidden
    in the time needed to solve the original $(s,s',q,q')$\text{-}\DISED instance,
    and to find the smaller \DISED instances, respectively.
    Performing similar analysis to \cref{cor:bbd_algo},
    we conclude that $\sT(m) = \Oh(\sT_{\MUL}(m)+m^{2+o(1)})$.
    \unbalanced{In the unweighted case, we note that the recurrence is instead
    \begin{equation*}
        \sT(m) = c_1 \alpha^{2} \cdot \sT(\alpha \cdot m) + c_2 \alpha^{2} \cdot (\sT_{\MonMUL}(m, m, m, \D) + m^{2 + o(1)} g(\D)) + c_3 \cdot m^{2+o(1)}.
    \end{equation*}
    and we bound the running time as in \Cref{cor:bbd_algo}.}
\end{proof}

\subsection{Proof of \texorpdfstring{\cref{thm:sed}}{Main Theorem~\ref{thm:sed}}}

Now, we show how to prove \cref{thm:sed} using the algorithm for \DISED from \cref{cor:dised}.

\sed*

We also prove the following theorem for unweighted \SED.

\begin{theorem}
    \label{thm:unweighted-sed}
    
    \unbalanced{There is an algorithm for unweighted \SED running in time $\Oh(\sT_{\MonMUL}(n) + n^{2 + o(1)} g(n))$, where $n = \max(|\bF|,|\bF'|)$ and $\sT_{\MonMUL}(m, m, m, \D) = \Oh(f(m) g(\D))$.} \lipicsEnd
\end{theorem}

\begin{proof}[\lipicsStart Proof of \cref{thm:sed} and \cref{thm:unweighted-sed}]
    Suppose $s$ is the first node (i.e., the root) and $q$ is the last node appearing in $\bS$.
    Similarly, suppose $s'$ is the first node and $q'$ is the last node appearing in $\bS'$.
    The algorithm for \SED is simple, we use $\mA_{\DISED}$ to solve the $(s,s',q,q')$\text{-}\DISED instance.

    This approach works as in the recursive calls performed by $\mA_{\DISED}$ we compute $\similarity(\sub(v), \sub(v'))$ for all $(v,v') \in \bS \times \bS'$.
    Indeed, fix $(v,v') \in \bS \times \bS'$. We have the invariant that we always recurse
    on an instance where both $v$ and $v'$ are contained, until we reach an $(u,u',w,w')$\text{-}\DISED instance
    such that either $u \preceq v \prec w$ and $u'$ immediately precedes $w'$, or $u$ immediately precedes $w$ and $u' \preceq v' \prec w'$.
    In either case $\similarity(\sub(v), \sub(v'))$ is among the output of the $(u,u',w,w')$\text{-}\DISED instance.

    To conclude the proof, we must also explain where we get the inputs from.
    We rewrite the indices for Input~\eqref{it:sed:input1} as
    \begin{align}
        \MoveEqLeft \{\ (x,x',y,y') \mid x \in \sB^{\top}_{\Left}(s,s',q,q'), \ (y,y') \in \sB^{\top}_{\Right}(s,s',q,q')\ \} = \nonumber \\[5pt]
        & \{\ (x,x',y,y') \mid (x,x') \in \fragment{\Left(s)}{\Left(q)} \times \Left(q'), \
            (y,y') \in \fragment{\Right(q)}{\Right(s)} \times \Right(q')\ \} \label{eq:sedinput:1}\\
        \cup \quad
        & \{\ (x,x',y,y') \mid (x,x') \in \fragment{\Left(s)}{\Left(q)} \times \Left(q'),
        (y,y') \in \Right(q) \times \fragment{\Right(q')}{\Right(s')}\ \} \label{eq:sedinput:2}\\
        \cup \quad
        & \{\ (x,x',y,y') \mid (x,x') \in \Left(q) \times \fragment{\Left(s')}{\Left(q')}, \
        (y,y') \in \fragment{\Right(q)}{\Right(s)} \times \Right(q')\ \} \label{eq:sedinput:3}\\
        \cup \quad
        & \{\ (x,x',y,y') \mid (x,x') \in \Left(q) \times \fragment{\Left(s')}{\Left(q')}, \
        (y,y') \in \Right(q) \times \fragment{\Right(q')}{\Right(s')}\ \}. \label{eq:sedinput:4}
    \end{align}
    We can compute the similarity for the various subsets of indices~\eqref{eq:sedinput:1}, \eqref{eq:sedinput:2}, \eqref{eq:sedinput:3} and \eqref{eq:sedinput:4} as follows.
    \begin{itemize}
        \item For~\eqref{eq:sedinput:1}, note that $\bF'\fragmentco{x'}{y'} = \bF'\fragmentco{\Left(q')}{\Right(q')} = \sub(q')$ contains a single node, namely $q'$, because $q'$ is the last node appearing in $\bS'$.
        Consequently, the similarity to be computed corresponds to
        \[
            \max\Big\{ \ 0,\  \max\nolimits_{v \in \bF\fragmentco{x}{y}} \{\ \eta(v, q')\ \} \ \Big\}.
        \]
        By using clever bottom-up dynamic programming we can compute these similarities for all $x \in \fragment{\Left(s)}{\Left(q)}$ and $y \in \fragment{\Right(q)}{\Right(s)}$
        in time $\Oh(m^2)$.

        \item For~\eqref{eq:sedinput:2}, note that $\bS \cap \bF\fragmentco{\Left(s)}{\Right(q)} = \{q\}$ and $\bS' \cap \bF'\fragmentco{\Left(q')}{\Right(s')} = \{q'\}$. Furthermore, since $\sub(q) = q$ and $\sub(q') = q'$, it follows that $\similarity(\sub(q), \sub(q')) = \max\{0, \eta(q, q')\}$. 
        This last similarity, combined with the input provided to \SED, gives us $\similarity(\sub(v), \sub(v'))$ for all pairs $(v, v') \in \bF\fragmentco{\Left(s)}{\Right(q)} \times \bF'\fragmentco{\Left(q')}{\Right(s')}$. Thus, we can apply \cref{thm:fed} on $\bF\fragmentco{\Left(s)}{\Right(q)}$ and $\bF'\fragmentco{\Left(q')}{\Right(s')}$ to obtain the suffix-prefix similarities between these intervals, which is exactly what we need to compute. 

         \item For~\eqref{eq:sedinput:3}, as~\eqref{eq:sedinput:3} equals to~\eqref{eq:sedinput:2} under reverse symmetry, it suffices to apply symmetric computations.

        \item For~\eqref{eq:sedinput:4}, as~\eqref{eq:sedinput:1} equals to~\eqref{eq:sedinput:2} under swap symmetry, it suffices to apply symmetric computations.
    \end{itemize}
    Finally, Inputs~\eqref{it:sed:input2},~\eqref{it:sed:input3},~\eqref{it:sed:input4}, and~\eqref{it:sed:input5}
    can all be computed using similar computations as in~\eqref{eq:sedinput:1}.

    The overall complexity can be bounded by $\Oh(\sT_{\MUL}(n) + n^{2 + o(1)})$ since the bottlenecks are applying the \DISED algorithm and \Cref{thm:fed}.
    \unbalanced{In the unweighted setting, we note that global tree size is at most $n$, thus obtaining a running time of $\Oh(\sT_{\MonMUL}(n) + n^{2 + o(1)} g(n))$ where $\sT_{\MonMUL}(n, n, n, \D) = \Oh(f(n) g(\D))$. Furthermore, we can use \Cref{thm:unweighted-fed-balance} instead of \Cref{thm:fed}.}
\end{proof}

\subsection{Patching together subproblems}
\label{sec:sed_patch}

This section is dedicated to proving \cref{lem:sed_patch}. Given the numerous cases involved, we first prove \cref{claim:sed_specialcase}, \cref{lem:sed_specialcase}, and \cref{lem:sed_specialcase2}, which allow us to break down the proof of \cref{lem:sed_patch} into three more ``digestible'' steps.

\begin{figure}[htbp]
    \centering
    \input{tikz/dised_patch.tex}
    \caption{The figure visualizes ranges of $x,y,x',y'$ for \cref{lem:sed_specialcase}, \cref{claim:sed_specialcase}, \cref{cl:retrieveinputs1}, \cref{lem:sed_specialcase2}
    and \cref{lem:sed_patch}. Colors have the same meaning as in \cref{fig:dised}.}
    \label{fig:dised_patch}
\end{figure}

\begin{lemma} \label{claim:sed_specialcase}
    Suppose we are provided with the input of an $(s,s',q,q')$\text{-}\DISED instance of size $m$.
    Further, let $r \in \bS$ such that $s \prec r \prec q$.
    Then, we can compute $\similarity(\bF\fragmentco{x}{y}, \bF'\fragmentco{x'}{y'})$
    for all
    $(x,x') \in (\Left(r) \ \times \ \fragment{\Right(q')}{\Right(s')})
    \ \cup \ (\fragment{\Left(r)}{\Left(q)} \ \times \ \Right(q'))$
    and
    $(y,y') \in  (\Right(r) \times \fragment{\Right(q')}{\Right(s')})
    \ \cup \ (\fragment{\Right(q)}{\Right(r)} \ \times \ \Right(s'))$
    in time $\Oh(\sT_{\MUL}(m) + m^{2+o(1)})$.
    
    \unbalanced{Furthermore, if the \DISED instance is unweighted and has global tree size $\D$, then the algorithm requires time $\Oh(\sT_{\MonMUL}(m, m, m, D) + m^{2+o(1)} g(\D))$ where $\sT_{\MonMUL}(m, m, m, \D) = \Oh(f(m) g(\D))$.}
\end{lemma}

    \begin{proof}
    Refer to \cref{fig:dised_patch} for a visualization of the ranges of $x,y,x',y'$.
    
    Consider arbitrary $(x,x') \in (\Left(r) \ \times \ \fragment{\Right(q')}{\Right(s')})
    \ \cup \ (\fragment{\Left(r)}{\Left(q)} \ \times \ \Right(q'))$ and
    $(y,y') \in  (\Right(r) \times \fragment{\Right(q')}{\Right(s')})
    \ \cup \ (\fragment{\Right(q)}{\Right(r)} \ \times \ \Right(s'))$.
    We perform a case distinction.
    \begin{enumerate}[(1)]
        \item If there exists a node $t \in \bS$ such that $r \preceq t \prec q$
        and $\similarity(\bF\fragmentco{x}{y}, \bF'\fragmentco{x'}{y'})$ aligns $\sub(t)$ to a subtree $\bF'\fragmentco{x'}{y'}$,
        then we note $\similarity(\bF\fragmentco{x}{y}, \bF'\fragmentco{x'}{y'}) = \similarity_q(\bF\fragmentco{x}{y}, \bF'\fragmentco{x'}{y'})$,
        where $\similarity_q$ comes from \cref{lem:sim_w_cut} applied on $\sub(r)$, $\bF'\fragmentco{\Right(q')}{\Right(s')}$, $\bS$ and $q$.
     
        \item Otherwise, by \cref{prp:align_or_remove},
        $\similarity(\bF\fragmentco{x}{y}, \bF'\fragmentco{x'}{y'})$ = $\similarity(\bF\fragmentco{x}{\Left(q)} + \bF\fragmentco{\Left(q)}{y}, \bF'\fragmentco{x'}{y'})$.
        From \cref{rmk:paired_anchors}\eqref{it:paired_anchors:a}, follows that
        $\Left(q) \times \fragment{\Right(q')}{\Right(s')}$ and $\Right(q) \times \fragment{\Right(q')}{\Right(s')}$ are paired anchor sets of $\similarity(\bF\fragmentco{x}{y}, \bF'\fragmentco{x'}{y'})$. Thus, we can write
        \begin{align*}
            \similarity(\bF\fragmentco{x}{y}, \bF'\fragmentco{x'}{y'}) =
            \max\nolimits_{z',w' \in \fragment{\Right(q')}{\Right(s')} \ : x'\leq z' \leq w' \leq y'} \Big\{ \
                &\similarity(\bF\fragmentco{x}{\Left(q)}, \bF'\fragmentco{x'}{z'}) \\
                &+ \similarity(\sub(q), \bF'\fragmentco{z'}{w'})\\
                &+ \similarity(\bF\fragmentco{\Right(q)}{y}, \bF'\fragmentco{w'}{y'})
            \ \Big\}.
        \end{align*}
    \end{enumerate}
    
    This concludes our case distinction.
    To compute $\similarity(\bF\fragmentco{x}{y}, \bF'\fragmentco{x'}{y'})$ for the desired ranges of $x,y,x',y'$, we break down the index set for $x,y,x',y'$ as follows (see \cref{fig:dised_patch}):
    \begin{align}
        \MoveEqLeft \{\ (x,x',y,y') \mid (x,x') \in (\Left(r) \times \fragment{\Right(q')}{\Right(s')})
    \cup (\fragment{\Left(r)}{\Left(q)} \times \Right(q')), \ \nonumber \\
         \MoveEqLeft (y,y') \in  (\Right(r) \times \fragment{\Right(q')}{\Right(s')})
     \cup (\fragment{\Right(q)}{\Right(r)} \times \Right(s')) \ \} = \nonumber \\[5pt]
        \vspace{10pt}
        & \{\ (x,x',y,y') \mid (x,x') \in \Left(r) \times \fragment{\Right(q')}{\Right(s')}, \
            (y,y') \in \Right(r) \times \fragment{\Right(q')}{\Right(s')} \ \} \label{eq:sed_specialcase:1}\\
        \cup \quad
        & \{\ (x,x',y,y') \mid (x,x') \in \Left(r) \times \fragment{\Right(q')}{\Right(s')}, \
        (y,y') \in \fragment{\Right(q)}{\Right(r)} \times \Right(s')\ \} \label{eq:sed_specialcase:2}\\
        \cup \quad
        & \{\ (x,x',y,y') \mid (x,x') \in \fragment{\Left(r)}{\Left(q)} \times \Right(q'), \
        (y,y') \in \Right(r) \times \fragment{\Right(q')}{\Right(s')} \ \} \label{eq:sed_specialcase:3}\\
        \cup \quad
        & \{\ (x,x',y,y') \mid (x,x') \in \fragment{\Left(r)}{\Left(q)} \times \Right(q'), \
        (y,y') \in \fragment{\Right(q)}{\Right(r)} \times \Right(s') \ \}. \label{eq:sed_specialcase:4}
    \end{align}
     
     We construct the matrices $A^{(1)} = (a_{i,j}^{(1)})$, $A^{(2)} = (a_{i,j}^{2})$, $B = (b_{i,j})$, $C^{(1)} = (c_{i,j}^{(1)})$, $C^{(2)} = (c_{i,j}^{(2)})$ of sizes $\mathsf{rsz'} \times \mathsf{rsz'}$, $\mathsf{lsz} \times \mathsf{rsz'}$, 
      $\mathsf{rsz'} \times \mathsf{rsz'}$,  $\mathsf{rsz'} \times \mathsf{rsz'}$, $\mathsf{rsz} \times \mathsf{rsz'}$,
      where $\mathsf{rsz'} = \Right(s')-\Right(q')+1$, $\mathsf{lsz} = \Left(q)-\Left(r)+1$, $\mathsf{rsz} = \Right(r)-\Right(q)+1$, and defined as (in the following definitions, we index matrix entries from zero, instead of one):
      \begin{align*}
        a_{i,j}^{(1)} &= \begin{cases}
            \similarity(\bF\fragmentco{\Left(r)}{\Left(q)}, \bF'\fragmentco{(\Right(q')+i)}{(\Right(q')+j)}) & i \leq j\\
            -\infty & \text{otherwise}
        \end{cases}, \\
        a_{i,j}^{(2)} &= \similarity(\bF\fragmentco{(\Left(r)+i)}{\Left(q)}, \bF'\fragmentco{\Right(q')}{(\Right(q')+j)}) \\
        b_{i,j} &=
        \begin{cases}
            \similarity(\sub(q), \bF'\fragmentco{\Right(q')+i}{\Right(q')+j}) & i \leq j\\
            -\infty & \text{otherwise}
        \end{cases},\\
        c_{i,j}^{(1)} &= 
         \begin{cases} \similarity(\bF\fragmentco{\Right(q)}{\Right(r)}, \bF'\fragmentco{(\Right(q')+i)}{(\Right(q')+j)} & i\leq j\\
         -\infty & \text{otherwise}
         \end{cases} \\
         c_{i,j}^{(2)} &= \similarity(\bF\fragmentco{(\Right(q)+j)}{\Right(r)}, \bF'\fragmentco{(\Right(q')+i)}{\Right(q')}
    \end{align*}

    Observe that all matrices \( A^{(1)}, A^{(2)}, B, C^{(1)}, \) and \( C^{(2)} \) have dimensions bounded by \( m \). We now outline the process for retrieving the required information to construct these matrices. The entries of \( A^{(1)} \) and \( A^{(2)} \) are derived from whole-infix and suffix-prefix similarities as specified in \cref{thm:fed} for \(\bF\fragmentco{\Left(r)}{\Left(q)}\) and \(\bF'\fragmentco{\Right(q')}{\Right(s')}\), respectively. The entries for matrix \( B \) are directly provided by Input~\eqref{it:sed:input5} of the $(s,s',q,q')$\text{-}\DISED instance. Similarly, \( C^{(1)} \) and \( C^{(2)} \) are constructed from whole-infix and suffix-prefix similarities via \cref{thm:fed} on \(\bF\fragmentco{\Right(q)}{\Right(s)}\) and \(\bF'\fragmentco{\Right(q')}{\Right(s')}\). Given the dimensions, this process requires no more than \(\Oh(\sT_{\MUL}(m) + m^{2+o(1)})\) time. 
    \unbalanced{In the unweighted setting, we note that all matrices are monotone and have entries bounded by global tree size $\Oh(\D)$. Furthermore, we apply \Cref{thm:unweighted-fed-balance} to obtain the necessary inputs.
    The total time requires no more than $\Oh(\sT_{\MonMUL}(m, m, m, \D) + m^{2 + o(1)} g(\D))$ time.}
    It is important to note that in both cases where \cref{thm:fed} is applied, the required similarities \(\similarity(\sub(v), \sub(v'))\) are available from the \SED instance, because at most one of \( v \) and \( v' \) belongs to a spine.
    
    Note that in \(\Oh(\sT_{\MUL}(m))\) time we can also compute via max-plus products the four matrices 
    \[
        D^{(11)} = (d^{(11)}_{i,j}), \ D^{(21)} = (d^{(21)}_{i,j}), \ D^{(12)} = (d^{(12)}_{i,j}) \quad \text{and} \quad D^{(22)} = (d^{(22)}_{i,j}),
    \]
    defined as $D^{(\mathtt{xy})} = A^{(\mathtt{x})} \star B \star C^{(\mathtt{y})}$ for $\mathtt{x},\mathtt{y} \in \{1,2\}$.
    
    By the discussion in the case distinction, we obtain that we can compute $\similarity(\bF\fragmentco{x}{y}, \bF'\fragmentco{x'}{y'})$ for the sets \eqref{eq:sed_specialcase:1}, \eqref{eq:sed_specialcase:2}, \eqref{eq:sed_specialcase:3}, and \eqref{eq:sed_specialcase:4} as
    \[
        \similarity(\bF\fragmentco{x}{y}, \bF'\fragmentco{x'}{y'}) = 
        \max \left\{
        \ \similarity_q(\bF\fragmentco{x}{y}, \bF'\fragmentco{x'}{y'}) \ , \
        \left\{
        \begin{aligned}
            & d^{(11)}_{x' - \Right(q'), y' - \Right(q')} &&(x,y,x',y') \in \eqref{eq:sed_specialcase:1} \\
            & d^{(12)}_{x' - \Right(q'), y - \Right(q)} &&(x,y,x',y') \in \eqref{eq:sed_specialcase:2} \\
            & d^{(21)}_{y - \Left(r), y' - \Right(q')} &&(x,y,x',y') \in \eqref{eq:sed_specialcase:3} \\
            & d^{(22)}_{y - \Left(r), y - \Right(q)} &&(x,y,x',y') \in \eqref{eq:sed_specialcase:4}
        \end{aligned}
         \right\}
        \right\},
    \]        
    where the values coming from $\similarity_q$ can be obtained in time $\Oh(\sT_{\MUL}(m) + m^{2+o(1)})$ \unbalanced{($\Oh(\sT_{\MonMUL}(m, m, m, \D) + m^{2 + o(1)} g(\D))$ in the unweighted setting)}
    using \cref{lem:sim_w_cut} on $\sub(r)$, $\bF'\fragmentco{\Right(q')}{\Right(s')}$, $\bS$ and $q$.
    As for the two applications of \cref{thm:sed}, the necessary input comes from the \SED instance.
\end{proof}

The proofs in the remaining part of this (sub)sections closely resemble the approach taken in \cref{claim:sed_specialcase}.
They tackle computing similarities of the form $\similarity(\bF\fragmentco{x}{y}, \bF'\fragmentco{x'}{y'})$ by distinguishing various possible cases for the ranges of $x,y,x',y'$.
Each case will be addressed with expressions that may be rewritten as max-plus products, where the necessary matrices can be 
constructed using the inputs of the $(s,s',q,q')$\text{-}\DISED instance, or by applying one of \cref{lem:sim_w_cut} or \cref{thm:fed}. Ultimately, we select the maximum among these cases.

For brevity, we omit details beyond the initial case distinction in these proofs—specifically,
how all distinct cases are put together into a single formula and how the distinct cases
can be written as max-plus product between matrices (we still specify where the entries of these matrices can be retrieved from).
It is important to note that in each of these cases,
we employ max-plus products on matrices no larger than the instance size and apply \cref{thm:fed} and \cref{lem:sim_w_cut} exclusively to forests of instance size,
ensuring that either no nodes from $\bS$ or no nodes from $\bS'$ are included.
In the unweighted tree edit distance problem, we further note that all matrices are monotone with entries bounded by the size of the trees.

\begin{lemma} \label{lem:sed_specialcase}
    Suppose we are provided with the input of an $(s,s',q,q')$\text{-}\DISED instance of size $m$.
    Further, let $r \in \bS$ such that $s \prec r \prec q$.
    Then, we can compute $\similarity(\bF\fragmentco{x}{y}, \bF'\fragmentco{x'}{y'})$
    for all
    $(x,x') \in (\Left(r) \ \times \ \fragment{\Right(q')}{\Right(s')})
    \ \cup \ (\fragment{\Left(s)}{\Left(q)} \ \times \ \Right(q'))$
    and
    $(y,y') \in  (\Right(r) \times \fragment{\Right(q')}{\Right(s')})
    \ \cup \ (\fragment{\Right(q)}{\Right(r)} \ \times \ \Right(s'))$
    in time $\Oh(\sT_{\MUL}(m) + m^{2+o(1)})$.

    \unbalanced{Furthermore, if the \DISED instance is unweighted and has global tree size $\D$, then the algorithm requires time $\Oh(\sT_{\MonMUL}(m, m, m, \D) + m^{2+o(1)} g(\D))$ where $\sT_{\MonMUL}(m, m, m, \D) = \Oh(f(m) g(\D))$.}
\end{lemma}

\begin{proof}
    Refer to \cref{fig:dised_patch} for a visualization of the ranges of $x,y,x',y'$
    for which $\similarity(\bF\fragmentco{x}{y}, \bF'\fragmentco{x'}{y'})$ needs to be computed.
    Observe in \cref{fig:dised_patch} that the ranges of $x,y,x',y'$ for \cref{lem:sed_specialcase},
    are almost identical to the ones for \cref{claim:sed_specialcase}.
    Thus, to prove \cref{lem:sed_specialcase},
    it suffices to compute the similarities for the ranges of $x$ that are left out in \cref{claim:sed_specialcase}.
    That is, we need to compute $\similarity(\bF\fragmentco{x}{x'}, \bF'\fragmentco{y}{y'})$
    for 
    \[
        (x,x') \in (\fragment{\Left(s)}{\Left(r)} \times \Right(q')) 
        \quad
        \text{and}
        \quad
     (y,y') \in  (\Right(r) \times \fragment{\Right(q')}{\Right(s')})
        \ \cup \ (\fragment{\Right(q)}{\Right(r)} \ \times \ \Right(s')).
    \]

    To this end, observe that for such values of $x$ and $y$,
    we have $\bF\fragmentco{x}{y} = \bF\fragmentco{x}{\Left(r)} + \bF\fragmentco{\Left(r)}{y}$.
    By \cref{rmk:mao}, we have that
    $\Left(r) \times \fragment{\Right(q')}{\Right(s')}$ serves as an anchor set for
    $\similarity(\bF\fragmentco{x}{y}, \bF'\fragmentco{x'}{y'})$.
    Consequently, for such ranges of $x,y,x',y'$ we can write:

    \begin{align*}
        \similarity(\bF\fragmentco{x}{y}, \bF'\fragmentco{x'}{y'}) =
        \max_{z' \in \fragment{\Right(q')}{\Right(s')} \mid z' \leq y'} \Big\{ \
        &\similarity(\bF\fragmentco{x}{\Left(r)}, \bF'\fragmentco{x'}{z'})\\
        &+ \similarity(\bF\fragmentco{\Left(r)}{y}, \bF'\fragmentco{z'}{y'})
        \ \Big\}.
    \end{align*}
    
    Recall that we are considering the case \( x' = \Right(q') \). Additionally, note that one of \( y \) and \( y' \) is always fixed, specifically either \( y = \Right(r) \) or \( y = \Right(s') \). This means that, in the final expression, we are always working with at most two of \( x, y, x', y' \) which are not fixed. As before, this expression can be computed as a max-plus product, with \( z' \) traversing the hidden dimension of the product. In this computation, the second summand is obtained from \cref{claim:sed_specialcase}, while the first summand is calculated by applying \cref{thm:fed} on \(\bF\fragmentco{\Left(s)}{\Left(r)}\) and \(\bF'\fragmentco{\Right(q')}{\Right(s')}\).
    \unbalanced{In the unweighted setting, we again observe that the matrices are monotone and have entries bounded by the global tree size $\Oh(\D)$, and apply \Cref{thm:unweighted-fed-balance} to compute \FED.}
\end{proof}

\begin{lemma} \label{lem:sed_specialcase2}
    Suppose we are provided with the input of an $(s,s',q,q')$\text{-}\DISED instance of size $m$.
    Further, let $r \in \bS$ such that $s \prec r \prec q$.
    Then, we can compute $\similarity(\bF\fragmentco{x}{y}, \bF'\fragmentco{x'}{y'})$
    for all
    $(x,x') \in (\fragment{\Left(s)}{\Left(r)} \ \times \ \Left(q'))$
    and
    $(y,y') \in  (\Right(r) \times \fragment{\Right(q')}{\Right(s')})
    \ \cup \ (\fragment{\Right(q)}{\Right(r)} \ \times \ \Right(s'))$
    in time $\Oh(\sT_{\MUL}(m) + m^{2+o(1)})$.

    \unbalanced{Furthermore, if the \DISED instance is unweighted and has global tree size $\D$, then the algorithm requires time $\Oh(\sT_{\MonMUL}(m, m, m, \D) +  m^{2+o(1)} g(\D))$ where $\sT_{\MonMUL}(m, m, m, \D) = \Oh(f(m) g(\D))$.}
\end{lemma}

\begin{proof} 
    Observe that for such ranges of $x',y'$ (refer to \cref{fig:dised_patch} for a visualization), we may write
    \[
        \bF'\fragmentco{x'}{y'} = \bF'\fragmentco{\Left(q')}{y'} = \sub(q') + \bF'\fragmentco{\Right(q')}{y'}.
    \]
    By \cref{rmk:mao}, we have that $\sB \coloneqq \fragment{\Left(s)}{\Right(r)} \times \Right(q')$
    is an anchor set of $\similarity(\bF\fragmentco{x}{y}, \bF'\fragmentco{x'}{y'})$.
    Now, consider the subset $\sB' \subseteq \sB$ defined as $\sB' \coloneqq \fragment{\Left(q)}{\Right(r)} \times \Right(q')$.
    We distinguish two cases.
    \begin{enumerate}[(1)]
        \item There is an anchor $(z,z') \in (\sB \setminus \sB')$,
        then we compute the similarity $\similarity(\bF\fragmentco{x}{y}, \bF'\fragmentco{x'}{y'})$ as
            \begin{align*}
            \similarity(\bF\fragmentco{x}{y}, \bF'\fragmentco{x'}{y'})
            &= \max_{\substack{z \in \fragment{\Left(s)}{\Left(q)} \mid x \leq z \leq y}} \Big\{ \
            \similarity(\bF\fragmentco{x}{z}, \sub(q'))
            + \similarity(\bF\fragmentco{z}{y}, \bF'\fragmentco{\Right(q')}{y'}) \ \Big\},
        \end{align*}

        using as first summand  the input of the $(s,s',q,q')$\text{-}\DISED instance,
        and as second summands the values obtained via \cref{lem:sed_specialcase}.
        \unbalanced{In the unweighted setting, observe that entries are bounded by global tree size $\Oh(\D)$.}

        \item There is an anchor $(z,z') \in \sB'$.
        Notice that for every $(z,z') \in \sB'$, we have
        \[
            \bF\fragmentco{z}{y} = \bF\fragmentco{z}{\Right(q)} + \bF\fragmentco{\Right(q)}{y}.
        \]
        By \cref{rmk:mao}, we have that $\sB'' \coloneqq \Right(q) \times \fragment{\Right(q')}{y'}$
        is an anchor set of $\similarity(\bF\fragmentco{x}{z}, \bF\fragmentco{x'}{z'})$.
        But since $(z,z')$ is an anchor of $\similarity(\bF\fragmentco{x}{y}, \bF'\fragmentco{x'}{y'})$, by \cref{rmk:anchor_transform}\eqref{it:anchor_transform:ii}, we obtain that $\sB''$ is anchor set of $\similarity(\bF\fragmentco{x}{y}, \bF'\fragmentco{x'}{y'})$.
        Therefore, we can calculate
        \begin{align*}
            \similarity(\bF\fragmentco{x}{y}, \bF'\fragmentco{x'}{y'})
            = \max\nolimits_{(z,z') \in \sB''} \Big\{ \
            \similarity(\bF\fragmentco{x}{z}, \bF'\fragmentco{x'}{z'})
            + \similarity(\bF\fragmentco{z}{y}, \bF'\fragmentco{z'}{y'}) \ \Big\}.
        \end{align*}
        In this last expression,
        we can get the first summand from the input of the $(s,s',q,q')$\text{-}\DISED instance,
        and the second term from \cref{thm:fed} applied on $\bF\fragmentco{\Right(q)}{\Right(r)}$
        and $\bF'\fragmentco{\Right(s')}{\Right(s')}$.
        \unbalanced{In the unweighted setting, observe that entries are bounded by global tree size $\Oh(\D)$, and we use \Cref{thm:unweighted-fed-balance} to compute \FED.}
        \qedhere
    \end{enumerate}
\end{proof}

\sedpatch

\begin{proof}
    Note, as the inputs of the $(r,s',q,q')$\text{-}\DISED instance are a subset of the inputs of the $(s,s',q,q')$\text{-}\DISED instance,
    we can already assume the outputs of the $(r,s',q,q')$\text{-}\DISED at our disposal.
    \unbalanced{In the unweighted setting, we observe that all matrices are monotone and have entries bounded by $\Oh(\D)$.}

    \begin{claim} \label{cl:retrieveinputs1}
        With all the values at our disposal so far,
        we can compute the inputs for the $(s,s',r,q')$\text{-}\DISED instance in time $\Oh(\sT_{\MUL}(m) + m^{2+o(1)})$.

        \unbalanced{Furthermore, if the \DISED instance is unweighted and has global tree size $\D$, then we can do so in time $\Oh(\sT_{\MonMUL}(m, m, m, \D) +  m^{2+o(1)} g(\D))$ where $\sT_{\MonMUL}(m, m, m, \D) = \Oh(f(m) g(\D))$.}
    \end{claim}

    \begin{claimproof}
        We rewrite the indices for Input~\eqref{it:sed:input1} as (see \cref{fig:dised_patch})
        \begin{align}
            \MoveEqLeft \{\ (x,x',y,y') \mid (x,x') \in \sB^{\top}_{\Left}(s,s',r,q'), \ (y,y') \in \sB^{\top}_{\Right}(s,s',r,q')\ \} = \nonumber \\[5pt]
            \vspace{10pt}
            & \{\ (x,x',y,y') \mid (x,x') \in \fragment{\Left(s)}{\Left(r)} \  \times \ \Left(q'), \
                (y,y') \in \fragment{\Right(r)}{\Right(s)} \ \times \ \Right(q') \ \} \label{eq:subinput:1}\\
            \cup \quad
            & \{\ (x,x',y,y') \mid (x,x') \in \fragment{\Left(s)}{\Left(r)} \  \times \ \Left(q'), \
            (y,y') \in \Right(r) \ \times \ \fragment{\Right(q')}{\Right(s')}\ \} \label{eq:subinput:2}\\
            \cup \quad
            & \{\ (x,x',y,y') \mid (x,x') \in \Left(r) \ \times \ \fragment{\Left(s')}{\Left(q')}, \
            (y,y') \in \fragment{\Right(r)}{\Right(s)} \ \times \ \Right(q')\ \} \label{eq:subinput:3}\\
            \cup \quad
            & \{\ (x,x',y,y') \mid (x,x') \in \Left(r) \ \times \ \fragment{\Left(s')}{\Left(q')}, \
            (y,y') \in \Right(r) \ \times \ \fragment{\Right(q')}{\Right(s')}\ \}. \label{eq:subinput:4}
        \end{align}

        Now, we describe how to compute $\similarity(\bF\fragmentco{x}{y}, \bF\fragmentco{x'}{y'})$
        for the index sets~\eqref{eq:subinput:1}, \eqref{eq:subinput:2}, \eqref{eq:subinput:3}, and \eqref{eq:subinput:4}.
        \begin{itemize}
            \item For~\eqref{eq:subinput:1}, note that~\eqref{eq:subinput:1} is a subset of the input of the $(s,s',q,q')$\text{-}\DISED instance.

            \item For~\eqref{eq:subinput:2}, note that~\eqref{eq:subinput:2} is a subset of the outputs of \cref{lem:sed_specialcase2}.

            \item For~\eqref{eq:subinput:3}, observe that~\eqref{eq:subinput:3} is equal to~\eqref{eq:subinput:2}
                under reverse symmetry.%

            \item For~\eqref{eq:subinput:4}, note that~\eqref{eq:subinput:4} is a subset of the output of the $(r,s',q,q')$\text{-}\DISED instance.
        \end{itemize}
        It remains to explain how to retrieve Inputs~\eqref{it:sed:input2}, \eqref{it:sed:input3}, \eqref{it:sed:input4}, \eqref{it:sed:input5}
        of the $(s,s',r,q')$\text{-}\DISED instance.
        Note, Inputs~\eqref{it:sed:input2} and \eqref{it:sed:input3} are a subset of the input of the $(s,s',q,q')$\text{-}\DISED instance.
        Conversely, Input~\eqref{it:sed:input4} can be obtain using \cref{lem:sed_specialcase} on the $(s,s',r,q')$\text{-}\DISED instance, and Input~\eqref{it:sed:input5} equals to Input~\eqref{it:sed:input4} under reverse symmetry.
    \end{claimproof}

    With all the inputs at our disposal, we can call also the subroutine on the  $(s,s',r,q')$\text{-}\DISED instance, and retrieve the ouputs. We can therefore assume that also the outputs of the $(s,s',r,q')$\text{-}\DISED instance are at our disposal.
    We now show how to compute the output of the $(s,s',q,q')$\text{-}\DISED instance.

    \begin{claim} \label{cl:retrieveoutputs}
        With all the values at our disposal so far,
        we can compute the outputs of the $(s,s',q,q')$\text{-}\DISED instance in time $\Oh(\sT_{\MUL}(m) + m^{2+o(1)})$.

        \unbalanced{Furthermore, if the \DISED instance is unweighted and has global tree size $\D$, then we can do so in time $\Oh(\sT_{\MonMUL}(m, m, m, \D) +  m^{2+o(1)} g(\D))$ where $\sT_{\MonMUL}(m, m, m, \D) = \Oh(f(m) g(\D))$.}
    \end{claim}

    \begin{claimproof}
    We rewrite the indices for the output of the $(s,s',q,q')$\text{-}\DISED instance as (see \cref{fig:dised_patch})
    \begin{align}
        \MoveEqLeft \{\ (x,x',y,y') \mid (x,x') \in \sB^{\bot}_{\Left}(s,s',q,q'), \ (y,y') \in \sB^{\bot}_{\Right}(s,s',q,q')\ \} = \nonumber \\[5pt]
        \vspace{10pt}
        & \{\ (x,x',y,y') \mid (x,x') \in \sB^{\bot}_{\Left}(s,s',r,q'), \
            (y,y') \in \sB^{\bot}_{\Right}(s,s',r,q') \ \} \label{eq:output:1}\\
        \cup \quad
        & \{\ (x,x',y,y') \mid (x,x') \in \sB^{\bot}_{\Left}(s,s',r,q'), \
        (y,y') \in  \fragment{\Right(q)}{\Right(r)} \times \Right(s') \ \} \label{eq:output:2}\\
        \cup \quad
        & \{\ (x,x',y,y') \mid (x,x') \in \fragment{\Left(r)}{\Left(q)} \times \Left(s'), \
        (y,y') \in \sB^{\bot}_{\Right}(s,s',r,q') \ \} \label{eq:output:3}\\
        \cup \quad
        & \{\ (x,x',y,y') \mid (x,x') \in \fragment{\Left(r)}{\Left(q)} \times \Left(s'), \
        (y,y') \in \fragment{\Right(q)}{\Right(r)} \times \Right(s') \ \}. \label{eq:output:4}
    \end{align}
    Now, we describe how to compute $\similarity(\bF\fragmentco{x}{y}, \bF\fragmentco{x'}{y'})$
    for the index sets~\eqref{eq:output:1}, \eqref{eq:output:2}, \eqref{eq:output:3}, and \eqref{eq:output:4}.
    \begin{itemize}
        \item For~\eqref{eq:output:1}, note that~\eqref{eq:output:1} is the output of the $(s,s',r,q')$\text{-}\DISED instance.
        \item For~\eqref{eq:output:2}, observe that for such ranges of $x$ and $y$, we can write 
        \[
            \bF\fragmentco{x}{y} = \bF\fragmentco{x}{\Left(r)} + \bF\fragmentco{\Left(r)}{y}.
        \]
        From \cref{rmk:mao} follows that $\Left(r) \times \fragment{\Left(s')}{\Right(s')}$ is an anchor
        set for $\similarity(\bF\fragmentco{x}{y}, \bF\fragmentco{x'}{y'})$.
        Next, we show how to compute $\similarity(\bF\fragmentco{x}{y}, \bF\fragmentco{x'}{y'})$
        in two cases that cover all possible scenarios.
        \begin{enumerate}[(1)]
            \item In the first case there is an anchor $(z,z') \in \Left(r) \times \fragment{\Left(s')}{\Left(q')}$.
            Then we can compute directly
            \begin{align*}
                \similarity(\bF\fragmentco{x}{y}, \bF'\fragmentco{x'}{y'}) =
                \max_{z' \in \fragment{\Left(s')}{\Left(q')} \mid x' \leq z'} \Big\{ \
                    &\similarity(\bF\fragmentco{x}{\Left(r)}, \bF'\fragmentco{x'}{z'}) \\
                    &+ \similarity(\bF\fragmentco{\Left(r)}{y}, \bF'\fragmentco{z'}{y'})
                \ \Big\},
            \end{align*}
            where the first summand can be obtained from \cref{thm:fed} and \Cref{thm:unweighted-fed-balance} on
            $\bF\fragmentco{\Left(s)}{\Left(r)}$ and $\bF'\fragmentco{\Left(s')}{\Left(q')}$
            and the second from the output from the $(r,s',q,q')$\text{-}\DISED instance.

            \item In the second case there is an anchor $(z,z') \in \Left(r) \times \fragment{\Left(q')}{\Right(s')}$. We further distinguish two subcases:
            \begin{enumerate}[(a)]
                \item In the first subcase suppose that $z' \geq \Right(q')$ and that the first case of \cref{prp:align_or_remove} applied on $\sub(s')$, $\bF\fragmentco{\Left(s)}{\Left(r)}$, $\sub(s')$ and $q'$ holds.
                That is, $z' \geq \Right(q')$ and there is $t' \in \bS$ such that $s' \preceq t' \prec q'$ and $\similarity(\bF\fragmentco{x}{z}, \bF\fragmentco{x'}{z'})$ aligns $\sub(t')$ to a subtree of $\bF\fragmentco{\Left(s)}{\Left(r)}$. Note, whenever $z' \geq \Right(q')$, then $\sub(q') \subseteq \bF'\fragmentco{x'}{z'} \subseteq \sub(s')$.
                We derive that 
                \[
                    \similarity(\bF\fragmentco{x}{z}, \bF'\fragmentco{x'}{z'}) = \similarity_{q'}(\bF\fragmentco{x}{z}, \bF\fragmentco{x'}{z'}),
                \]
                where $\similarity_{q'}$ can be obtained from \cref{lem:sim_w_cut}
                on $\sub(s'), \bF\fragmentco{\Left(s)}{\Left(r)}, \bS'$ and $q' \in \bS'$.
                We may therefore write
                \begin{align*}
                        \similarity(\bF\fragmentco{x}{y}, \bF'\fragmentco{x'}{y'}) =
                        \max_{(z,z') \in \Left(r) \times \fragment{\Right(q')}{\Right(s')}} \Big\{ \
                            &\similarity_{q'}(\bF\fragmentco{x}{z}, \bF'\fragmentco{x}{z'}) \\
                            &+ \similarity(\bF\fragmentco{z}{y}, \bF'\fragmentco{z'}{y'})
                        \ \Big\}.
                \end{align*}
                In this last expression, we can obtain the first summand,
                as mentioned before, from \cref{lem:sim_w_cut}
                on $\sub(s'), \bF\fragmentco{\Left(s)}{\Left(r)}, \bS'$ and $q' \in \bS'$,
                and the second summand from \cref{claim:sed_specialcase}\eqref{eq:sed_specialcase:2}.
                
                \item In the second subcase, we have that either $z' < \Right(q')$, or $z' \geq \Right(q')$ and $\similarity(\bF\fragmentco{x}{z}, \bF\fragmentco{x'}{z'}) = \similarity(\bF\fragmentco{x}{z}, \bF'\fragmentco{x'}{\Left(q')} + \bF'\fragmentco{\Left(q')}{z'})$.
                Therefore, for either of the two possible subcases, we can write
                $\similarity(\bF\fragmentco{x}{z}, \bF\fragmentco{x'}{z'}) = \similarity(\bF\fragmentco{x}{z}, \bF'\fragmentco{x'}{\Left(q')} + \bF'\fragmentco{\Left(q')}{z'})$.
                By \cref{rmk:mao}, we have that $\fragment{\Left(s)}{\Left(r)} \times \Left(q')$
                is a anchor set for $\similarity(\bF\fragmentco{x}{z}, \bF'\fragmentco{x'}{z'})$.
                By \cref{rmk:anchor_transform}\eqref{it:anchor_transform:ii}, we get that $\fragment{\Left(s)}{\Left(r)} \times \Left(q')$ is also an anchor set for
                $\similarity(\bF\fragmentco{x}{y}, \bF'\fragmentco{x'}{y'})$.
                Consequently, we may write
                \begin{align*}
                    \similarity(\bF\fragmentco{x}{y}, \bF'\fragmentco{x'}{y'}) =
                    \max_{z \in \fragment{\Left(s)}{\Left(r)} \mid x \leq z} \Big\{ \
                        &\similarity(\bF\fragmentco{x}{z}, \bF'\fragmentco{x'}{\Left(q')}) \\
                        &+ \similarity(\bF\fragmentco{z}{y}, \bF'\fragmentco{\Left(q')}{\Right(s')})
                    \ \Big\}.
                \end{align*}
                In this last expression the first summand can be obtained from
                from \cref{thm:fed} and \Cref{thm:unweighted-fed-balance} on $\bF\fragmentco{\Left(s)}{\Left(r)},\bF'\fragmentco{\Left(s')}{\Left(q')}$,
                and the second summand from \cref{lem:sed_specialcase2}.
            \end{enumerate}
        \end{enumerate}

        \item For~\eqref{eq:output:3}, it suffices to use the reverse symmetry on~\eqref{eq:output:2}.
        \item For~\eqref{eq:output:4}, note that it is a subset of the output of the $(r,s',q,q')$\text{-}\DISED instance.
        \claimqedhere
    \end{itemize}
    \end{claimproof}
    This concludes the proof of \cref{lem:sed_patch}.
\end{proof}

\subsection{Handling the base case}
\label{sec:sed_basecase}

\sedbasecase

\begin{proof}
    It suffices to prove the case where $s$ immediately precedes $q$
    since the case where $s'$ immediately precedes $q'$ equals to the former case under swap symmetry.

    We perform a case distinction on where/whether the spine nodes are mapped by $\similarity(\bF\fragmentco{x}{y}, \bF'\fragmentco{x'}{y'})$.
    \begin{itemize}
        \item If neither $s$ nor any node $r' \in \bS'$ such that $s' \preceq r' \prec q'$ is mapped,
        then $\similarity(\bF\fragmentco{x}{y}, \bF\fragmentco{x'}{y'})
        = \similarity(\bF\fragmentco{x}{\Left(q)} + \bF\fragmentco{\Left(q)}{y}, \bF'\fragmentco{x'}{\Left(q')} + \bF'\fragmentco{\Left(q')}{y'}$,
        as the spine node do not contribute to the similarity of the two forests.
        Via \cref{rmk:paired_anchors}\eqref{it:paired_anchors:b} 
        we compute
        \begin{align*}
            \similarity(\bF\fragmentco{x}{y}, \bF'\fragmentco{x'}{y'})
            = \max\nolimits_{\substack{
                (z,z') \in \sB^{\top}_{\Left}(s,s',q,q') \mid z \geq x, z' \geq x'\\
                (w,w') \in \sB^{\top}_{\Right}(s,s',q,q') \mid w\leq y, w'\leq y'}} \Big\{ \
            &\similarity(\bF\fragmentco{x}{z}, \bF'\fragmentco{x'}{z'}) \\
            &+ \similarity(\bF\fragmentco{z}{w}, \bF'\fragmentco{z'}{w'}) \\
            &+ \similarity(\bF\fragmentco{w}{y}, \bF'\fragmentco{w'}{y'}) \ \Big\},
        \end{align*}
        where the first and third summand can be obtained from \cref{thm:fed} on
        the forests
        $\bF\fragmentco{\Left(s)}{\Left(q)}$, $\bF'\fragmentco{\Left(s')}{\Left(q')}$
        and $\bF\fragmentco{\Right(q)}{\Right(s)}$, $\bF'\fragmentco{\Right(q')}{\Right(s')}$,
        respectively, and
        the second summand from Input~\eqref{it:sed:input1} of the $(s,s',q,q')$\text{-}\DISED instance.
        \unbalanced{Note that in the unweighted setting, all values are at most global tree size $\Oh(\D)$ and we can apply \Cref{thm:unweighted-fed-balance} in place of \Cref{thm:fed}.}

        \item If $s$ is not mapped and a node $r' \in \bS'$ such that $s' \preceq r' \prec q'$ is mapped to a node in $\bF\fragmentco{x}{\Left(q)}$, then we can use \cref{lem:sim_w_cut}.
        
        \item If $s$ is mapped to a node $r' \in \bS'$ such that $s' \preceq r' \prec q'$, then we must have $\bF\fragmentco{x}{y} = \sub(s)$ and we can compute
        \begin{align*}
            \MoveEqLeft \similarity(\sub(s), \bF\fragmentco{x'}{y'}) = \\
            &\max_{r' \in \bS' \mid s' \preceq r' \prec q'} \Big\{ \
            \eta(s,r')
            +\similarity(\bF\fragmentco{x}{\Left(q)} + \bF\fragmentco{\Left(q)}{y}, \bF'\fragmentco{\Left(r')+1}{\Right(r')-1})
            \ \Big\}.
        \end{align*}
        Note that the values \(\similarity(\bF\fragmentco{x}{\Left(q)} + \bF\fragmentco{\Left(q)}{y}, \bF'\fragmentco{\Left(r')+1}{\Right(r')-1})\) were already computed in the first two cases (distinguishing whether another node in \(\bS'\) is mapped or not). Therefore, there is nothing further to compute in this case.

        \item If $\bF\fragmentco{x}{y} = \sub(s)$ and $s$ is mapped to any node contained in $\sub(q')$,
        then $\similarity(\sub(s), \bF\fragmentco{x'}{y'}) = \similarity(\sub(s), \sub(q'))$.
        Observe that $\similarity(\sub(s), \sub(q'))$ is already at our disposal as input of the
        $(s,s',q,q')$\text{-}\DISED instance.

        \item If $\bF\fragmentco{x}{y} = \sub(s)$ and $s$ is mapped to any node contained in $\bF'\fragmentco{x'}{\Left(q')}$,
            then we have that
            $\similarity(\sub(s), \bF'\fragmentco{x'}{y'}) = \similarity(\sub(s), \sub(v'))$
            for some $v' \in \bF'\fragmentco{x'}{\Left(q')}$ and $v' \notin \bS'$.
            Note, all these values are already at our disposal in the \SED Problem.

        \item Lastly, if $\bF\fragmentco{x}{y} = \sub(s)$ and $s$ is mapped to any node contained in $\bF'\fragmentco{\Right(q')}{y'}$,
        then we can apply reverse symmetry obtaining the previous case. \qedhere
    \end{itemize}
\end{proof}

\subsection{Spine Edit Distance on Unbalanced Instances}
\label{sec:sed-unbalance}

In this section, we give an algorithm for Spine Edit Distance on unbalanced instances, where one tree, say $\bF$ is significantly larger than $\bF'$. 
In the following assume $\abs{\bF} = n$, $\abs{\bF'} = n'$, and $n \geq n'$.
For the smaller tree $\bF'$, we let $r'$ denote its root (i.e.\ top-most node of the spine) and $b'$ denote the bottom-most node of the spine.
In this setting, we compute a slightly restricted version of the divide et impera scheme for Spine Edit Distance.

\defproblemwlist
{Unbalanced Divide-et-Impera Spine Edit Distance (\UDISED)}
{$s \prec q \in \bS$, $s' \prec q' \in \bS'$ and
\begin{enumerate}[(i)]
    \item $\similarity(\bF\fragmentco{x}{y}, \bF'\fragmentco{x'}{y'})$ for all of the following indices:
    \begin{align*}
        & \{\ (x,x',y,y') \mid (x,x') \in \fragment{\Left(s)}{\Left(r)} \  \times \ \Left(q'), \
        (y,y') \in \Right(r) \ \times \ \fragment{\Right(q')}{\Right(s')}\ \}\\
        \cup \quad
        & \{\ (x,x',y,y') \mid (x,x') \in \Left(r) \ \times \ \fragment{\Left(s')}{\Left(q')}, \
        (y,y') \in \fragment{\Right(r)}{\Right(s)} \ \times \ \Right(q')\ \}\\
        \cup \quad
        & \{\ (x,x',y,y') \mid (x,x') \in \Left(r) \ \times \ \fragment{\Left(s')}{\Left(q')}, \
        (y,y') \in \Right(r) \ \times \ \fragment{\Right(q')}{\Right(s')}\ \}.
    \end{align*}
    \label{it:u_sed:input1}
    \item $\similarity(\sub(q), \bF'\fragmentco{x'}{y'})$ for all $x',y' \in \fragment{\Left(s')}{\Left(q')}$.
    \label{it:u_sed:input2}
    \item $\similarity(\sub(q), \bF'\fragmentco{x'}{y'})$ for all $x',y' \in \fragment{\Right(q')}{\Right(s')}$.
    \label{it:u_sed:input3}
\end{enumerate}}{
The values
\begin{enumerate}[(i)]
    \item $\similarity(\sub(s), \bF'\fragmentco{x'}{y'})$ for all $x' \in \fragment{\Left(s')}{\Left(q')}, y' \in \fragment{\Right(q')}{\Right(s')}$.
    \label{it:u_sed:output1}
    \item $\similarity(\sub(s), \bF'\fragmentco{x'}{y'})$ for all $x', y' \in \fragment{\Left(s')}{\Left(q')}$.
    \label{it:u_sed:output2}
    \item $\similarity(\sub(s), \bF'\fragmentco{x'}{y'})$ for all $x', y' \in \fragment{\Right(q')}{\Right(s')}$.
    \label{it:u_sed:output3}
\end{enumerate}
}

On unbalanced instances, we will only decompose the larger graph $\bF$, so that we only handle $(s,s',q,q')$\text{-}\UDISED instances with $s' = r'$ and $q' = b'$.
Given such an instance, we define its \emph{size} as $(m, n')$ where
\[
        m = \abs{\ \sub(s) \setminus \sub(q) \ }.
\]

In the following (sub)section,
we present a divide-et-impera algorithm $\mA_{\UDISED }$ solving \UDISED.
The three lemmas, the proof of which we defer for now, that motivate such an approach are the following.

\begin{restatable}{lemma}{usedbalanced}\label{lem:used_balance}
    Consider an $(s,s',q,q')$\text{-}\UDISED instance of size $(m, n')$ such that $m = \Oh(n')$.
    Then, we can solve the $(s,s',q,q')$\text{-}\UDISED instance in time $\Oh(\sT_{\MUL}(n') + n'^{2 + o(1)})$.

    \unbalanced{Furthermore, if the \UDISED instance is unweighted, the algorithm requires time $\Oh(\sT_{\MonMUL}(n') + n'^{2 + o(1)} g(n'))$ where $\sT_{\MonMUL}(n', n', n', \D) = \Oh(f(n') g(\D))$.}
\end{restatable}

\begin{restatable}{lemma}{usedbasecase}\label{lem:used_basecase}
    Consider an $(s,s',q,q')$\text{-}\UDISED instance of size $(m, n')$
    such that $s$ immediately precedes $q$.
    Then, we can solve the $(s,s',q,q')$\text{-}\UDISED instance in time $(m/n')^{1+o(1)} \cdot (\sT_{\MUL}(n') + n'^{2 + o(1)})$.

    \unbalanced{Furthermore, if the \UDISED instance is unweighted, the algorithm requires time $(m/n')^{1+o(1)} \cdot (\sT_{\MonMUL}(n') + n'^{2 + o(1)} g(n'))$ where $\sT_{\MonMUL}(n', n', n', \D) = \Oh(f(n') g(\D))$}.
\end{restatable}

\begin{restatable}{lemma}{usedpatch}\label{lem:used_patch}
    Suppose we are given an $(s,s',q,q')$\text{-}\UDISED instance of size $m$, and let
    $r \in \bS$ be such that $s \prec r \prec q$.
    Then, we can reduce the $(s,s',q,q')$\text{-}\UDISED instance to the
    $(s,s',r,q')$\text{-}\UDISED and $(r,s',q,q')$\text{-}\UDISED instances
    in time $(m/n')^{1+o(1)} \cdot (\sT_{\MUL}(n') + n'^{2 + o(1)})$.

    \unbalanced{Furthermore, if the \UDISED instance is unweighted, the algorithm requires time $(m/n')^{1+o(1)} \cdot (\sT_{\MonMUL}(n') + n'^{2 + o(1)} g(n'))$ where $\sT_{\MonMUL}(n', n', n', \D) = \Oh(f(n') g(\D))$}.
\end{restatable}

We now give an algorithm for the \UDISED problem.

\begin{lemma}
    \label{lem:dised_unbalanced}
    There exists an algorithm $\mA_{\UDISED}$ solving \UDISED
    which runs in time $(m/n')^{1 + o(1)} \cdot \left( \sT_{\MUL}(n') + n'^{2 + o(1)} \right)$ on instances of size $(m, n')$.

    \unbalanced{Furthermore, if the \UDISED instance is unweighted, then the algorithm only requires time $(m/n')^{1+o(1)} \cdot \left( \sT_{\MonMUL}(n') + n'^{2 + o(1)} g(n') \right)$ where $\sT_{\MonMUL}(n', n', n', \D) = \Oh(f(n') g(\D))$.}
\end{lemma}

\begin{proof}
    It suffices to apply recursively (the simplified version of) \cref{lem:dised_divide} with threshold $\Delta(m) = m / \alpha$
    for a constant $\alpha \geq 1$ to be determined later.
    In particular, we decompose only the larger tree $\bF$.

    If $m = \Oh(n')$, we simply apply \cref{lem:used_balance} to the instance.
    Otherwise, we proceed to construct a set $I \subseteq \bS \times \bS$ using \Cref{alg:sed_count} with threshold $\Delta$ yielding a sequence of spine nodes $s = r_1 \prec r_2 \prec \cdots \prec r_{d} = q$.
    The algorithm guarantees for each $i \in \fragmentco{1}{d}$, either $(r_{i}, s', r_{i + 1}, q')$\text{-}\UDISED has size at most $(\Delta, n')$ or $r_i$ immediately precedes $r_{i + 1}$.
    To complete the reduction we recursively solve $(r_{i}, s', r_{i + 1}, q')$\text{-}\UDISED instances, which requires total time $\Oh(d \sT(\Delta, n'))$ in the former case and $(m/n')^{1 + o(1)} \cdot \sT_{\MUL}(n')$ in the latter case, since the combined size of the instances in the latter case is at most $(m, n')$.
    Note that there are at most $d$ instances to recombine, and thus recombining the instances require $\Oh(d (m/n')^{1 + o(1)} \cdot \sT_{\MUL}(n'))$.

    Thus, we obtain the following recurrence by using $d \leq 3m/\Delta$ so that for some constant $C$,
    \begin{align*}
        \sT(m) &= \Oh(3 m/\Delta \sT(\Delta, n') + (m/n')^{1+o(1)} \cdot (\sT_{\MUL}(n') + n'^{2 + o(1)})) \\
        &\leq 3 C \alpha \sT(m/\alpha, n') + (m/n')^{1+o(1)} \cdot (\sT_{\MUL}(n') + n'^{2 + o(1)}).
    \end{align*}
    For any $\epsilon > 0$,
    we can choose a sufficiently large constant $\alpha$ such that
    $\log_{\alpha}(3 C \alpha) < 1 + \epsilon$.
    Similarly as in \Cref{cor:lrbbd_algo} we conclude that $\sT(m, n') = (m/n')^{1+o(1)} \cdot \left( \sT_{\MUL}(n') + n'^{2 + o(1)} \right)$.

    \unbalanced{In the unweighted setting, we instead have the recurrence
    \begin{align*}
        \sT(m) &= \Oh(3 m/\Delta \sT(\Delta, n') + (m/n')^{1+o(1)} \cdot (\sT_{\MonMUL}(n') + n'^{2 + o(1)} g(n')))
    \end{align*}
    which yields the desired result.}
\end{proof}

It remains to prove the necessary lemmas.

\subsubsection{Handling the Recursive Case}
\label{sec:sed_patch_unbalanced}

We begin with the recursive case where we patch together two sub-problems.

\usedpatch*

As in balanced case, we begin with the proof of a useful lemma.

\begin{lemma} 
    \label{lem:used_recurse_input}
    Suppose we are provided with the input of an $(s,s',q,q')$\text{-}\UDISED instance of size $(m, n')$.
    Further, let $r \in \bS$ such that $s \prec r \prec q$.
    Then, we can compute $\similarity(\bF\fragmentco{x}{y}, \bF'\fragmentco{x'}{y'})$
    for all
    $(x,x') \in (\fragmentco{\Left(s)}{\Left(r)} \ \times \ \Left(q'))$
    and
    $(y,y') \in  (\Right(r) \times \fragmentco{\Right(q')}{\Right(s')})$
    in time $(m/n')^{1+o(1)} \cdot (\sT_{\MUL}(n') + n'^{2 + o(1)})$.

    \unbalanced{Furthermore, if the \UDISED instance is unweighted, the algorithm requires time $(m/n')^{1+o(1)} \cdot (\sT_{\MonMUL}(n') + n'^{2 + o(1)} g(n"))$ where $\sT_{\MonMUL}(n', n', n', \D) = \Oh(f(n') g(\D))$}.
\end{lemma}

\begin{proof}
    Observe that 
    $x \in \fragmentco{\Left(s)}{\Left(r)}$ so that 
    \[
        \bF\fragmentco{x}{y} = \bF\fragmentco{x}{\Right(r)} = \bF\fragmentco{x}{\Left(r)} + \sub(r)
    \]
    By \cref{rmk:mao} there exists $z' \in \fragment{\Left(q')}{y'}$ such that $(\Left(r), z')$ is an anchor of
    $\similarity(\bF\fragmentco{x}{\Right(r)}, \bF'\fragmentco{\Left(q')}{y'})$.
    Then we compute for the similarity $\similarity(\bF\fragmentco{x}{\Right(r)}, \bF'\fragmentco{\Left(q')}{y'})$ as
    \begin{align*}
        &\similarity(\bF\fragmentco{x}{\Right(r)}, \bF'\fragmentco{\Left(q')}{y'}) \\
        &= \max\nolimits_{z' \in \fragment{\Left(q')}{y'}} \Big\{ \ \similarity(\bF\fragmentco{x}{\Left(r)}, \bF'\fragmentco{\Left(q')}{z'})
        + \similarity(\bF\fragmentco{\Left(r)}{\Right(r)}, \bF'\fragmentco{z'}{y'}) \ \Big\},
    \end{align*}
    where the first summand can be computed from the output of \Cref{thm:fed-unbalance} applied to $\bF\fragmentco{\Left(s)}{\Left(r)}$ (which contains no spine nodes) and $\bF'\fragmentco{\Left(q')}{\Right(s')}$. 
    Thus, since the inputs are at our disposal, we may apply \Cref{thm:fed-unbalance} in time $\Oh((m/n)^{1 + o(1)} \cdot (\sT_{\MUL}(n') + n'^{2 + o(1)}))$.
    We obtain the second summand in two steps.
    First, note $\similarity(\sub(r), \bF'\fragmentco{\Left(q')}{y'})$ is a subset of Output~\eqref{it:u_sed:output1} of the $(r, s', q, q')$-\UDISED instance.
    To obtain the remaining values, we note that the remaining values $\similarity(\sub(r), \bF'\fragmentco{z'}{y'})$ are given by Output~\eqref{it:u_sed:output3} of the $(r, s', q, q')$-\UDISED instance.
    Finally, the computation can be performed by a max-plus product of an $(m + 1) \times (n' + 1)$ matrix and a $(n' + 1) \times (n' + 1)$ matrix in time $\Oh(\sT_{\MUL}(m, n', n')) = \Oh(m/n' \cdot \sT_{\MUL}(n'))$.

    \unbalanced{Note that if the input instance is unweighted, the latter matrix is monotone, as for each fixed $y'$ we have that similarity must not increase as $z'$ increases (i.e.\ each column is non-increasing). 
    Furthermore, since both matrices have entries corresponding to similarities between forests where the smaller forest has at most $\Oh(n')$ nodes, both matrices have entries bounded by $\Oh(n')$.
    In particular, we can combine them in time $\Oh(\sT_{\MonMUL}(m, n', n', n')) = \Oh(m/n) \cdot \sT_{\MonMUL}(n')$.
    Furthermore, we instead apply \Cref{thm:fed-unbalance-unweighted} which only takes time $(m/n')^{1+o(1)} \cdot (\sT_{\MonMUL}(n') + n'^{2 + o(1)} g(n'))$.}
\end{proof}

We can finally prove \Cref{lem:used_patch}.

\begin{proof}[Proof of \Cref{lem:used_patch}]
    As before, the inputs of the $(r,s',q,q')$\text{-}\UDISED instance are a subset of the inputs of the $(s,s',q,q')$\text{-}\UDISED instance,
    so we can already assume the outputs of the $(r,s',q,q')$\text{-}\UDISED at our disposal.

    \begin{claim} 
        \label{claim:uretrieveinputs1}
        Given the outputs to the $(r,s',q,q')$\text{-}\UDISED instance,
        we can compute the inputs for the $(s,s',r,q')$\text{-}\UDISED instance in time $(m/n')^{1+o(1)} \cdot \sT_{\MUL}(n')$.

        \unbalanced{Furthermore, if the \UDISED instance is unweighted, the algorithm requires time $(m/n')^{1+o(1)} \cdot \sT_{\MonMUL}(n')$}.
    \end{claim}

    \begin{claimproof}
        We rewrite the indices for Input~\eqref{it:u_sed:input1} as
        \begin{align}
            & \{\ (x,x',y,y') \mid (x,x') \in \fragment{\Left(s)}{\Left(r)} \  \times \ \Left(q'), \
            (y,y') \in \Right(r) \ \times \ \fragment{\Right(q')}{\Right(s')}\ \} \label{eq:subinput:2u}\\
            \cup \quad
            & \{\ (x,x',y,y') \mid (x,x') \in \Left(r) \ \times \ \fragment{\Left(s')}{\Left(q')}, \
            (y,y') \in \fragment{\Right(r)}{\Right(s)} \ \times \ \Right(q')\ \} \label{eq:subinput:3u}\\
            \cup \quad
            & \{\ (x,x',y,y') \mid (x,x') \in \Left(r) \ \times \ \fragment{\Left(s')}{\Left(q')}, \
            (y,y') \in \Right(r) \ \times \ \fragment{\Right(q')}{\Right(s')}\ \}. \label{eq:subinput:4u}
        \end{align}

        Now, we describe how to compute $\similarity(\bF\fragmentco{x}{y}, \bF\fragmentco{x'}{y'})$
        for the index sets \eqref{eq:subinput:2u}, \eqref{eq:subinput:3u}, and \eqref{eq:subinput:4u}.
        \begin{itemize}
            \item For~\eqref{eq:subinput:2u}, note that~\eqref{eq:subinput:2u} is computed by \cref{lem:used_recurse_input}.

            \item For~\eqref{eq:subinput:3u}, note that~\eqref{eq:subinput:3u} is equal to~\eqref{eq:subinput:2u}
            under reverse symmetry.
            
            \item For~\eqref{eq:subinput:4u}, note that~\eqref{eq:subinput:4u} is a subset of the output of the $(r,s',q,q')$-\UDISED instance.
        \end{itemize}
        It remains to explain how to retrieve Inputs~\eqref{it:u_sed:input2}, \eqref{it:u_sed:input3}
        of the $(s,s',r,q')$-\DISED instance.
        Note, Inputs~\eqref{it:u_sed:input2} and \eqref{it:u_sed:input3} are Outputs~\eqref{it:u_sed:output2} and \eqref{it:u_sed:output3} of the $(r,s',q,q')$-\UDISED instance.
    \end{claimproof}

    With all the inputs at our disposal, we can call also the subroutine on the  $(s,s',r,q')$-\UDISED instance, and retrieve the outputs.
    Towards this, we observe that the $(s,s',r,q')$-\UDISED provides the outputs to the $(s,s',q,q')$-\UDISED instance.
\end{proof}

\subsubsection{Handling the Base Cases}
\label{sec:sed_basecase_unbalanced}

Recall that $s' = r'$ is the root of $\bF'$ and $q' = b'$ is the bottom node of the spine of $\bF'$.
In particular, $\sub(s') \setminus \sub(q') = \bF' \setminus \{q'\}$. 
First, we show that when $m = \Oh(n')$, we can use our \DISED to compute the \UDISED instance.

\usedbalanced

\begin{proof}
    In this case, we can treat the instance as a $(s, s', q, q')$-\DISED instance and apply \Cref{cor:dised}.
    Note that such a instance is a \DISED instance of size $\max(m, n') = \Oh(n')$ and global tree size $\Oh(n')$ regardless of the size of $\sub(s)$.
    To apply \Cref{cor:dised}, note that we must supply the missing inputs.
    First, we complete Input~\eqref{it:sed:input1} by computing $\similarity(\bF\fragmentco{x}{y}, \sub(q'))$ for $x \in \fragment{\Left(s)}{\Left(q)}$ and $y \in \fragment{\Right(q)}{\Right(s)}$.
    Since $\sub(q')$ is a single node, we can compute these inputs using dynamic programming as in the \DISED algorithm (see the computation of Input~\eqref{eq:sedinput:1} in \Cref{thm:sed}) in time $\Oh(m^2) = \Oh(n'^2)$, noting that we are already given $\similarity(\sub(q), \sub(q'))$ as part of the input to the \UDISED instance.
    Similarly, since $\sub(q')$ is a single node, we can supply the missing Inputs~\eqref{it:sed:input2} and \eqref{it:sed:input3} for \DISED.
    
    Finally, we show that we can obtain the outputs to the \UDISED instance given outputs to the \DISED instance.
    In particular, for Output~\eqref{it:u_sed:output1}, note $(\Left(s), x') \in \sB^{\bot}_{\Left}(s, s', q, q')$ and $(\Right(s), y') \in \sB^{\bot}_{\Right}(s, s', q, q')$.
    Note that computing the \DISED instance requires time $\Oh(\sT_{\MUL}(n') + n'^{2 + o(1)})$ by \Cref{cor:dised}.
    \unbalanced{In the unweighted setting, the \DISED instance with both size and global tree size bounded by $\Oh(n')$ requires time $\Oh(\sT_{\MonMUL}(n') + n'^{2 + o(1)} g(n'))$.}
    
    For Output~\eqref{it:u_sed:output2}, we consider two cases.
    If $t \prec q$ is aligned by $\similarity(\sub(s), \bF'\fragmentco{\Left(s')}{\Left(q')}$, we apply \Cref{lem:sim_w_cut} on $\sub(s), \bF'\fragmentco{\Left(s')}{\Left(q')}, \bS, q$ which requires time $\Oh(\sT_{\MUL}(m) + m^{2 + o(1)}) = \Oh(\sT_{\MUL}(n') + n'^{2 + o(1)})$.
    Observe that all the required inputs are given in the \SED instance since $\bF'\fragmentco{\Left(s')}{\Left(q')}$ has no spine nodes.
    
    On the other hand, if no spine node $t \prec q$ is aligned, then we can write
    \[
        \similarity(\sub(s), \bF'\fragmentco{x'}{y'})
        = \similarity(\bF\fragmentco{\Left(s)}{\Left(q)} + \bF\fragmentco{\Left(q)}{\Right(s)}, \bF'\fragmentco{x'}{y'})
    \]
    as no nodes $t \prec q$ contribute to the similarity of the two forests.
    Via \cref{rmk:paired_anchors}\eqref{it:paired_anchors:a} we compute
    \begin{align*}
        \similarity(\sub(s), \bF\fragmentco{x'}{y'})
        = \max\nolimits_{\substack{
            z' \in \bF'\fragmentco{\Left(s')}{\Left(q')} \mid w' \geq z' \geq x'\\
            w' \in \bF'\fragmentco{\Left(s')}{\Left(q')} \mid z' \leq w' \leq y'}} \Big\{ \
        &\similarity(\bF\fragmentco{\Left(s)}{\Left(q)}, \bF\fragmentco{x'}{z'}) \\
        &+ \similarity(\bF\fragmentco{\Left(q)}{\Right(q)}, \bF\fragmentco{z'}{w'}) \\
        &+ \similarity(\bF\fragmentco{\Right(q)}{\Right(s)}, \bF\fragmentco{w'}{y'}) \ \Big\},
    \end{align*}
    where the first and third summand can be obtained from \cref{thm:fed} and \Cref{thm:unweighted-fed-balance} on
    the forests
    $\bF\fragmentco{\Left(s)}{\Left(q)}$, $\bF'\fragmentco{\Left(s')}{\Left(q')}$
    and $\bF\fragmentco{\Right(q)}{\Right(s)}$, $\bF'\fragmentco{\Left(s')}{\Left(q')}$,
    respectively, and
    the second summand from Input~\eqref{it:sed:input2} of the $(s,s',q,q')$\text{-}\DISED instance.
    The computation can then be performed by max-plus products of three $(n' + 1) \times (n'+ 1)$ matrices.
    Overall, the time to compute Output~\eqref{it:u_sed:output2} is
    \begin{equation*}
        \Oh(\sT_{\MUL}(n') + n'^{2 + o(1)})
    \end{equation*}

    Finally, we note that Output~\eqref{it:u_sed:output3} can be obtained using similar arguments as Output~\eqref{it:u_sed:output2}.
\end{proof}

Next, we consider the case where $s$ immediately precedes $q$.

\usedbasecase

\begin{proof}
    We begin with Output~\eqref{it:u_sed:output1}.
    We perform a case distinction on where/whether $s$ is mapped by $\similarity(\sub(s), \bF'\fragmentco{x'}{y'})$ for all $x' \in \fragment{\Left(s')}{\Left(q')}$ and $y' \in \fragment{\Right(q')}{\Right(s')}$.
    \begin{itemize}
        \item If $s$ is not mapped to any node of $\bF'\fragmentco{x'}{y'}$,
        then
        \[
            \similarity(\sub(s), \bF'\fragmentco{x'}{y'})
            = \similarity(\bF\fragmentco{\Left(s)}{\Left(q)} + \bF\fragmentco{\Left(q)}{\Right(s)}, \bF'\fragmentco{x'}{y'})
        \]
        as the spine node $s$ does not contribute to the similarity of the two forests.
        Via \cref{rmk:paired_anchors}\eqref{it:paired_anchors:a} we compute
        \begin{align*}
            \similarity(\sub(s), \bF\fragmentco{x'}{y'})
            = \max\nolimits_{\substack{
                z' \in \bF'\fragmentco{x'}{y'} \mid w' \geq z' \geq x'\\
                w' \in \bF'\fragmentco{x'}{y'} \mid z' \leq w' \leq y'}} \Big\{ \
            &\similarity(\bF\fragmentco{\Left(s)}{\Left(q)}, \bF\fragmentco{x'}{z'}) \\
            &+ \similarity(\bF\fragmentco{\Left(q)}{\Right(q)}, \bF\fragmentco{z'}{w'}) \\
            &+ \similarity(\bF\fragmentco{\Right(q)}{\Right(s)}, \bF\fragmentco{w'}{y'}) \ \Big\}.
        \end{align*}

        We can combine these values using a max-plus product of three $(n' + 1) \times (n' + 1)$ matrices, which requires time
        $\Oh(\sT_{\MUL}(n'))$.
        \unbalanced{In the unweighted setting, we note that all three matrices are row-monotone or column-monotone and have entries bounded by $O(n')$.
        In particular, the matrix multiplication requires time $\Oh(\sT_{\MonMUL}(n'))$.}
        It remains to describe how to obtain the appropriate inputs.
        To do so, we proceed by case analysis on $z', w'$.

        \begin{itemize}
            \item If $z', w' \in \fragment{\Left(s')}{\Left(q')}$, then the first summand can be obtained from applying \cref{thm:fed-unbalance} to $\bF\fragmentco{\Left(s)}{\Left(q)}$ and $\bF'$ where the required inputs are supplied as the first forest contains no spine nodes.
            The second summand can be obtained from Input~\eqref{it:u_sed:input2} of the $(s, s', q, q')$-\UDISED instance.
            The last summand can be obtained from applying \cref{thm:fed-unbalance} to $\bF\fragmentco{\Right(q)}{\Right(s)}$ and $\bF'$ where again the first forest contains no spine nodes.
    
            \item The case $z', w' \in \fragment{\Left(q') + 1}{\Right(s')}$ can be handled analogously.
            We note that the range $\fragment{\Left(q') + 1}{\Right(s')}$ is equivalent to the range $\fragment{\Right(q')}{\Right(s')}$ with respect to $\similarity$ since $\bF'\fragmentco{\Left(q') + 1}{\Right(s')} = \bF'\fragmentco{\Right(q')}{\Right(s')}$.

            \item If $z' \in \fragment{\Left(s')}{\Left(q')}, w' \in \fragment{\Left(q') + 1}{\Right(s')}$ then the first and third summand can be obtained from \cref{thm:fed-unbalance} and
            the second summand from Input~\eqref{it:u_sed:input1} of the $(s,s',q,q')$\text{-}\UDISED instance.
            For the second summand, we note that $\bF'\fragmentco{z'}{\Left(q') + 1} = \bF'\fragmentco{z'}{\Left(q')}$ so we can retrieve this input from Input~\eqref{it:u_sed:input2} instead.
            For the first and third summand, we compute \UFED on forests $\bF\fragmentco{\Left(s)}{\Left(q)}, \bF'$ and $\bF\fragmentco{\Right(q)}{\Right(s)}, \bF'$, respectively, noting that in both cases the first forest contain no spine nodes so the inputs are at our disposal.
        \end{itemize}

        In all cases, the computational bottleneck is the application of \Cref{thm:fed-unbalance}, which requires time $(m/n')^{1 + o(1)} \cdot (\sT_{\MUL}(n') + n'^{2 + o(1)})$.
        \unbalanced{In the unweighted case, we use \Cref{thm:fed-unbalance-unweighted} which requires time $(m/n')^{1 + o(1)} \cdot (\sT_{\MonMUL}(n') + n'^{2 + o(1)} g(n'))$.}

        \item If $s$ is mapped to a node $r' \in \bS'$ such that $s' \preceq r' \prec q'$, then
        \begin{align*}
            \MoveEqLeft \similarity(\sub(s), \bF\fragmentco{x'}{y'}) = \\
            &\max_{r' \in \bS' \mid s' \preceq r' \prec q'} \Big\{ \
            \eta(s,r')
            +\similarity(\bF\fragmentco{x}{\Left(q)} + \sub(q) + \bF\fragmentco{\Right(q)}{y}, \bF'\fragmentco{\Left(r')+1}{\Right(r')-1})
            \ \Big\}.
        \end{align*}
        Note, we already computed the values
        $\similarity(\bF\fragmentco{x}{\Left(q)} + \sub(q) + \bF\fragmentco{\Right(q)}{y}, \bF'\fragmentco{\Left(r')+1}{\Right(r')-1}$
        in the previous case, so there is nothing more left to calculate in this case.

        \item If $s$ is mapped to any node contained in $q' = \sub(q')$,
        then $\similarity(\sub(s), \bF\fragmentco{x'}{y'}) = \similarity(\sub(s), \sub(q')) = \eta(s, q')$ which can easily be computed.

        \item If $s$ is mapped to any node contained in $\bF'\fragmentco{x'}{\Left(q')}$,
        then we have that
        $\similarity(\sub(s), \bF'\fragmentco{x'}{y'}) = \similarity(\sub(s), \sub(v'))$
        for some $v' \in \bF'\fragmentco{x'}{\Left(q')}$ and $v' \notin \bS'$.
        Note, all these values are already at our disposal in the \SED Problem.

        \item Lastly, if $s$ is mapped to any node contained in $\bF'\fragmentco{\Right(q')}{y'}$,
        then we can apply reverse symmetry obtaining the previous case.
    \end{itemize}

    We now discuss Outputs~\eqref{it:u_sed:output2} and \eqref{it:u_sed:output3}.
    We will explicitly describe the algorithm for computing Output~\eqref{it:u_sed:output2}, noting that Output~\eqref{it:u_sed:output3} can be computed analogously.
    If $s \prec q$ is aligned by $\similarity(\sub(s), \bF'\fragmentco{x'}{y'})$, then we apply \Cref{lem:sim_w_cut} to $\sub(s), \bF', \bS, q$ as in \Cref{lem:used_balance} noting that on our unbalanced instance, the algorithm requires time $(m/n')^{1 + o(1)} \cdot (\sT_{\MUL}(n') + n'^{2 + o(1)})$.
    On the other hand, if $s$ is not aligned, we argue similarly as in \Cref{lem:used_balance}.
    We can write
    \[
        \similarity(\sub(s), \bF'\fragmentco{x'}{y'})
        = \similarity(\bF\fragmentco{\Left(s)}{\Left(q)} + \bF\fragmentco{\Left(q)}{\Right(s)}, \bF'\fragmentco{x'}{y'})
    \]
    as $s$ does not contribute to the similarity.
    Via \cref{rmk:paired_anchors}\eqref{it:paired_anchors:a} we compute
    \begin{align*}
        \similarity(\sub(s), \bF\fragmentco{x'}{y'})
        = \max\nolimits_{\substack{
            z' \in \bF'\fragmentco{\Left(s')}{\Left(q')} \mid w' \geq z' \geq x'\\
            w' \in \bF'\fragmentco{\Left(s')}{\Left(q')} \mid z' \leq w' \leq y'}} \Big\{ \
        &\similarity(\bF\fragmentco{\Left(s)}{\Left(q)}, \bF\fragmentco{x'}{z'}) \\
        &+ \similarity(\bF\fragmentco{\Left(q)}{\Right(q)}, \bF\fragmentco{z'}{w'}) \\
        &+ \similarity(\bF\fragmentco{\Right(q)}{\Right(s)}, \bF\fragmentco{w'}{y'}) \ \Big\},
    \end{align*}
    where the first and third summand can be obtained from \cref{thm:fed} and \Cref{thm:unweighted-fed-balance} on
    the forests
    $\bF\fragmentco{\Left(s)}{\Left(q)}$, $\bF'\fragmentco{\Left(s')}{\Left(q')}$
    and $\bF\fragmentco{\Right(q)}{\Right(s)}$, $\bF'\fragmentco{\Left(s')}{\Left(q')}$,
    respectively, and
    the second summand from Input~\eqref{it:sed:input2} of the $(s,s',q,q')$\text{-}\DISED instance.
    The computation can then be performed by max-plus products of three $(n' + 1) \times (n'+ 1)$ matrices.
    
    To bound the overall running time, we note that obtaining the outputs of \Cref{thm:fed-unbalance} takes time $(m/n')^{1+o(1)} \cdot (\sT_{\MUL}(n') + n'^{2 + o(1)})$.
    \unbalanced{In the unweighted setting, we instead obtain $(m/n')^{1+o(1)} \cdot (\sT_{\MonMUL}(n') + n'^{2 + o(1)}g(n'))$.}
\end{proof}

\subsubsection{Computing \SED on Unbalanced Instances with \UDISED}

We give the result for both weighted and unweighted \SED.

\begin{restatable}{theorem}{used}
    \label{thm:used}
    There is an $(n/n')^{1+o(1)} \cdot \left(\sT_{\MUL}(n') + n'^{2 + o(1)} \right)$ time algorithm for \SED, where $n = \abs{\bF}, n' = \abs{\bF'}$ and $n \geq n'$. 
\end{restatable}

\unweightused*

Since the proof is essentially identical, we prove both results, pointing out modifications where necessary.

\begin{proof}[\lipicsStart Proof of \cref{thm:used} and \cref{thm:unweighted-used}]
    As in \Cref{thm:sed}, we fix $s, s'$ as the roots of $\bF, \bF'$ and $q, q'$ as the last nodes in the spines of $\bS, \bS'$.
    The algorithm then runs $\mA_{\UDISED}$ on the $(s, s', q, q')$\text{-}\UDISED instance and thus takes time $(n/n')^{1+o(1)} \cdot \left( \sT_{\MUL}(n') + n'^{2 + o(1)} \right)$.
    \unbalanced{In the unweighted setting, \UDISED only takes time $(n/n')^{1+o(1)} \cdot \left( \sT_{\MonMUL}(n') + n'^{2 + o(1)} g(n') \right)$.}

    We claim that this computes all required outputs $\similarity(\sub(v), \sub(v'))$ for $(v, v') \in \bS \times \bS'$.
    Fix a pair $(v, v') \in \bS \times \bS'$. 
    Assume we always recurse on the instance satisfying $u \preceq v \prec w$ and $u' \preceq v' \prec w'$ until we reach one of the following:
    \begin{itemize}
        \item A $(u, u', w, w')$\text{-}\UDISED where $u$ immediately precedes $w$ and $u' \preceq v' \prec w'$.
        In this case $\similarity(\sub(v), \sub(v'))$ is among the outputs of the \UDISED instance.

        \item A $(u, u', w, w')$\text{-}\DISED. 
        In this case, the outputs are computed following identical arguments as \Cref{thm:sed}.
    \end{itemize}

    To conclude the proof, we bound the run-time of obtaining the inputs to the $(s, s', q, q')$\text{-}\UDISED instance.
    We begin with Input~\eqref{it:u_sed:input1} and rewrite the input indices as
    \begin{align}
        & \{\ (x,x',y,y') \mid (x,x') \in \fragment{\Left(s)}{\Left(q)} \times \Left(q'),
        (y,y') \in \Right(q) \times \fragment{\Right(q')}{\Right(s')}\ \} \label{eq:usedinput:2}\\
        \cup \quad
        & \{\ (x,x',y,y') \mid (x,x') \in \Left(q) \times \fragment{\Left(s')}{\Left(q')}, \
        (y,y') \in \fragment{\Right(q)}{\Right(s)} \times \Right(q')\ \} \label{eq:usedinput:3}\\
        \cup \quad
        & \{\ (x,x',y,y') \mid (x,x') \in \Left(q) \times \fragment{\Left(s')}{\Left(q')}, \
        (y,y') \in \Right(q) \times \fragment{\Right(q')}{\Right(s')}\ \}. \label{eq:usedinput:4}
    \end{align}
    We can compute the similarity for the various subsets of indices \eqref{eq:usedinput:2}, \eqref{eq:usedinput:3} and \eqref{eq:usedinput:4} as follows.
    \begin{itemize}
        \item For~\eqref{eq:usedinput:2}, we apply \Cref{thm:fed-unbalance} to the forests $\bF\fragmentco{\Left(s)}{\Right(q)}$ and $\bF'\fragmentco{\Left(q')}{\Right(s')}$.
        Note that both forests do not contain any spine nodes so we are given as inputs to the \SED instance all required inputs to the \UFED instance, and thus we may compute $\similarity(\bF\fragmentco{x}{\Right(q)}, \bF'\fragmentco{\Left(q')}{\Right(s')})$ in time $(n/n')^{1+o(1)} \cdot \left( \sT_{\MUL}(n') + n'^{2 + o(1)} \right)$.

        \item For~\eqref{eq:usedinput:3}, we likewise apply \Cref{thm:fed-unbalance} but this time to the forests $\bF\fragmentco{\Left(q)}{\Right(s)}$ and $\bF'\fragmentco{\Left(s')}{\Right(q')}$.
        Following similar arguments as \eqref{eq:usedinput:2}, we note that the required similarities $\similarity(\bF\fragmentco{\Left(q)}{y}, \bF'\fragmentco{x'}{\Right(q')})$ are given by Output~\eqref{it:ufed:output2} in time $(n/n')^{1+o(1)} \cdot \sT_{\MUL}(n')$.

        \item For~\eqref{eq:usedinput:4}, we observe that $\sub(q) = q$ is a single node so we can compute the necessary entries using dynamic programming in $O(n'^2)$ time.
    \end{itemize}
    Finally, for Inputs~\eqref{it:u_sed:input2},~\eqref{it:u_sed:input3} we again note that $\sub(q) = q$ is a single node so we can obtain the necessary entries in $O(n'^2)$ time.
    
    \unbalanced{In the unweighted setting, computing \UFED takes time $(n/n')^{1+o(1)} \cdot \left( \sT_{\MonMUL}(n') + n'^{2 + o(1)} g(n') \right)$.}
\end{proof}

\section{Reduction from Spine Edit Distance to Tree Edit Distance}
\label{sec:ted}

In this section, we prove \cref{lem:sed_to_ted}, i.e.,
we reduce \TED on \SED.\footnote{Although \cite{BGHS19} briefly mentions the existence of this reduction, it does so only in a footnote.}

\sedtoted*

We will first reduce the following variant of \SED to \SED: 

\defproblem{Spine to All-Subtree Tree Edit Distance ($\SASED$)}{Two forests $\bF,\bF'$, a spine $\bS \subseteq \bF$, and $\similarity(\sub(v),\sub(v'))$ for all $(v,v') \in (\bF \setminus \bS) \times \bF'$}{$\similarity(\sub(v),\sub(v'))$ for all $(v,v') \in \bF \times \bF'$.}

Utilizing the above reduction, we further reduce \ETED to \SED, which we define as a more general problem than \TED.

\defproblem{All Subtrees Tree Edit Distance ($\ETED)$}{Two forests $\bF,\bF'$.}{$\similarity(\sub(v),\sub(v'))$ for all $(v,v') \in \bF \times \bF'$.}

Our reduction is based on a tree decomposition
similar to the well-known \emph{heavy path decomposition} introduced by Harel and Tarjan \cite{HeavyLight}.
For our purposes, we can formulate this decomposition as follows.

\begin{proposition}\label{prp:heavy-light}
    Let $\bF$ be a forest. Then, there exists a spine $\bS\subseteq \bF$ ending at leaf $v$ satisfying
    \[
        \abs{\bF\fragmentco{0}{\Left(v)}} \leq \abs{\bF}/2
        \quad
        \text{and}
        \quad
        \abs{\bF\fragmentco{\Right(v)}{2\abs{\bF}+1}} \leq \abs{\bF}/2.
    \]
    Here, $\bF\fragmentco{0}{\Left(v)}$ and $\bF\fragmentco{\Right(v)}{2\abs{\bF}+1}$ are exactly the nodes right and left w.r.t.~$\bS$, respectively.
\end{proposition}
\begin{proof}
    Attach to $\bF$ a new root such that $\bF$ becomes a tree.
    We start by walking from the (new) root of \(\bF\) to a leaf.
    At each node \(u\), we continue the walk to the rightmost child \(w\) of \(u\)
    that satisfies \(\abs{\bF\fragmentco{0}{\Left(w)}} \leq \abs{\bF}/2\).

    We aim to maintain the invariant \(\abs{\bF\fragmentco{0}{\Left(u)}} \leq \abs{\bF}/2\)
    and \(\abs{\bF\fragmentco{\Right(u)}{2\abs{\bF}+1}} \leq \abs{\bF}/2\).
    This invariant clearly holds at the beginning, and ultimately,
    the walk traces out a spine \(\bS\) with the properties stated.

    For the inductive step, observe that, by the definition of this walk,
    moving from \(u\) to the next node \(w\) ensures \(\abs{\bF\fragmentco{0}{\Left(w)}} \leq \abs{\bF}/2\).
    We must show that \(\abs{\bF\fragmentco{\Right(w)}{2\abs{\bF}+1}} \leq \abs{\bF}/2\).
    To do this, we consider two cases.
    If \(w\) has a right sibling \(w'\), then \(\abs{\bF\fragmentco{0}{\Right(w)}} = \abs{\bF\fragmentco{0}{\Left(w')}} \geq \abs{\bF}/2\),
    which leads to \(\abs{\bF\fragmentco{\Right(w)}{2\abs{\bF}+1}} = \abs{\bF} - \abs{\bF\fragmentco{0}{\Right(w)}} \leq \abs{\bF}/2\).
    On the other hand, if \(w\) is the rightmost child of \(u\),
    then \(\abs{\bF\fragmentco{\Right(w)}{2\abs{\bF}+1}} = \abs{\bF\fragmentco{\Right(u)}{2\abs{\bF}+1}} \leq \abs{\bF}/2\).
\end{proof}

\begin{definition}\label{def:light_right}
    For each forest $\bF$,
    fix an arbitrary spine $\bS$ ending in a leaf $v$ satisfying \cref{prp:heavy-light}.
    Define $\bL(\bF) \coloneqq \bF\fragmentco{0}{\Left(v)}$ and $\bR(\bF) \coloneqq \bF\fragmentco{\Right(v)}{2\abs{\bF}+1}$,
    for the subforest containing the nodes on the left and right w.r.t.~the fixed spine $\bS$, respectively.
    If $\bF=\emptyset$ is empty, then, we set $\bL(\bF)=\bR(\bF)=\emptyset$.
\end{definition}

We are now ready to show the reduction from \SASED to \SED.

\begin{lemma}
\label{lem:sed_to_sased}
Suppose there exists an algorithm for \SED on two forests $\bH,\bH'$
running in time $\sT_{\SED}(m, m') = \Oh(f(m) g(m'))$,
where $m = |\bH|, m' = |\bH'|$ and $f(m) = \Omega(m), g(m') = \Omega(m')$ are some functions.
Then, there is an algorithm for \SASED on two forests
$\bF,\bF'$ running in time $\Oh(f(n) g(n')\log n')$,
where $n = |\bF|, n' = |\bF'|$.
\end{lemma}

\begin{proof}
In the base case, if $n' = 0$, we can trivially return an empty set. 

In general, we find a spine $\bS' \subseteq \bF'$ using \cref{prp:heavy-light}, and obtain $\bL(\bF')$ and $\bR(\bF')$ according to \cref{def:light_right}. Then we recursively solve \SASED on $(\bF, \bL(\bF'))$ and $(\bF, \bR(\bF'))$. 

Note that with the input of the \SASED instance $(\bF, \bF')$, and the outputs of the \SASED instances $(\bF, \bL(\bF'))$ and $(\bF, \bR(\bF'))$, we have collected all inputs for \SED between $(\bF, \bF')$, so we can run the \SED algorithm to finish the algorithm in $\Oh(f(n) g(n'))$ time. 

The running time $\sT(n, n')$ of the algorithm can be formulated as the following formula:
\[
\sT(n, n') = 2\sT(n, n' / 2) + \Oh(f(n) g(n')),
\]
which can be upper bounded by $\sT(n, n') = \Oh(f(n) g(n') \log n')$. 
\end{proof}

Now we proceed to show the reduction from \ETED to \SED.

\begin{lemma}\label{lem:sed_to_eted}
Suppose there exists an algorithm for \SED on two forests $\bH,\bH'$
running in time $\sT_{\SED}(m, m') = \Oh(f(m) g(m'))$,
where $m = |\bH|, m' = |\bH'|$ and $f(m) = \Omega(m), g(m') = \Omega(m')$ are some functions.
Then, there is an algorithm for \ETED on two forests
$\bF,\bF'$ running in time $\Oh(f(n) g(n') \log^2 \max(n', n))$,
where $n = \abs{\bF}, n' = \abs{\bF'}$.
\end{lemma}
\begin{proof}
In the base case, if $n = 0$ or $n'=0$, we can simply return an empty set. 

In general, given $\bF, \bF'$, we use \cref{prp:heavy-light} and \cref{def:light_right} to prepare $\bS, \bL(\bF), \bR(\bF), \bS', \bL(\bF'), \bR(\bF')$. We recursively solve \ETED in $4$ instances $(\bL(\bF), \bL(\bF')), (\bL(\bF), \bR(\bF')), (\bR(\bF), \bL(\bF')), (\bR(\bF), \bR(\bF'))$. 

Next, we apply \cref{lem:sed_to_sased} on inputs $(\bF, \bL(\bF')), (\bF, \bR(\bF')), (\bF', \bL(\bF)), (\bF', \bR(\bF))$. Note that applying \cref{lem:sed_to_sased} on the first two instances clearly takes time $\Oh(f(n) g(n') \log n')$. Applying \cref{lem:sed_to_sased} directly on the last two instances would take time $\Oh(f(n') g(n) \log n)$. However, notice that by symmetry, if there exists a \SED algorithm with running time $\Oh(f(m) g(m'))$, there is also a \SED algorithm with running time $\Oh(g(m) f(m'))$ by swapping the two input forests. Applying \cref{lem:sed_to_sased} on the last two instances assuming the existence of a \SED algorithm with running time $\Oh(g(m) f(m'))$ gives us a running time $\Oh(f(n) g(n') \log n)$. Hence, the overall running time of applying \cref{lem:sed_to_sased} is $\Oh(f(n) g(n') \log \max(n', n))$. 

Finally, at this point we have collected all required inputs for \SED on input $(\bF, \bF')$, so we can apply the assume \SED algorithm in time $\Oh(f(n) g(n'))$. 

Let $\sT(n, n')$ be the running time of this algorithm, then it can be written as 
\[
\sT(n, n') = 4\sT(n/2, n'/2) + \Oh(f(n) g(n') \log \max(n', n)), 
\]
which can be upper bounded by $\sT(n, n') = \Oh(f(n) g(n') \log^2 \max(n', n))$. 
\end{proof}

Indeed, \cref{lem:sed_to_ted} follows directly from \cref{lem:sed_to_eted}.

\sedtoted

\begin{proof}
   We attach new roots \(v\) and \(v'\) to the two forests \(\bF\) and \(\bF'\) in the \TED instance, resulting in the trees \(\bT\) and \(\bT'\)
   (if the forests are already trees, the new root becomes the parent of the original root).
   We set weight such that we enforce that the two roots are both deleted.
   This ensures that \(\similarity(\bF,\bF') = \similarity(\bT,\bT') = \similarity(\sub(v),\sub(v'))\),
   which can be computed using \cref{lem:sed_to_eted}.
\end{proof}

As corollaries, we obtain the following results on tree edit distances, by applying \Cref{lem:sed_to_ted} to \Cref{thm:used} or \Cref{thm:unweighted-used}.

\TEDtheorem*

\unweightedTED*

\bibliographystyle{alpha}
\bibliography{main}

\newcommand{\etalchar}[1]{$^{#1}$}
\begin{thebibliography}{MTWZU09}

\bibitem[ACS08]{ACS08}
Carlos Eduardo~Rodrigues Alves, E.~N. C{\'{a}}ceres, and Siang~Wun Song.
\newblock An all-substrings common subsequence algorithm.
\newblock {\em Discret. Appl. Math.}, 156(7):1025--1035, 2008.

\bibitem[ADV{\etalchar{+}}25]{ADVXXZ24}
Josh Alman, Ran Duan, Virginia {Vassilevska Williams}, Yinzhan Xu, Zixuan Xu, and Renfei Zhou.
\newblock More asymmetry yields faster matrix multiplication.
\newblock In {\em Proceedings of the 2025 Annual ACM-SIAM Symposium on Discrete Algorithms (SODA)}, page to appear, 2025.

\bibitem[AGM97]{AGM97}
Noga Alon, Zvi Galil, and Oded Margalit.
\newblock On the exponent of the all pairs shortest path problem.
\newblock {\em J. Comput. Syst. Sci.}, 54(2):255--262, 1997.

\bibitem[AHVW16]{LB3StringED15}
Amir Abboud, Thomas~Dueholm Hansen, Virginia {Vassilevska Williams}, and Ryan Williams.
\newblock Simulating branching programs with edit distance and friends: or: a polylog shaved is a lower bound made.
\newblock In {\em Proceedings of the 48th Annual ACM Symposium on Theory of Computing (STOC)}, page 375–388, 2016.

\bibitem[AJ21]{AkmalJin21}
Shyan Akmal and Ce~Jin.
\newblock {Faster Algorithms for Bounded Tree Edit Distance}.
\newblock In {\em Proceedings of the 48th International Colloquium on Automata, Languages, and Programming (ICALP)}, pages 12:1--12:15, 2021.

\bibitem[Aku99]{akutsu1999RNA}
Tatsuya Akutsu.
\newblock Approximation and exact algorithms for rna secondary structure prediction and recognition of stochastic context-free languages.
\newblock {\em J. Comb. Optim.}, 3:321--336, 1999.

\bibitem[AP72]{aho1972minimum}
Alfred~V. Aho and Thomas~G Peterson.
\newblock A minimum distance error-correcting parser for context-free languages.
\newblock {\em SIAM J. Comput.}, 1(4):305--312, 1972.

\bibitem[AV14]{popular14}
Amir Abboud and Virginia {Vassilevska Williams}.
\newblock Popular conjectures imply strong lower bounds for dynamic problems.
\newblock In {\em Proceedings of the 55th Annual Symposium on Foundations of Computer Science (FOCS)}, 2014.

\bibitem[BCH{\etalchar{+}}07]{BCHMRWZ07}
Rolf Backofen, Shihyen Chen, Danny Hermelin, Gad~M. Landau, Mikhail~A. Roytberg, Oren Weimann, and Kaizhong Zhang.
\newblock Locality and gaps in rna comparison.
\newblock {\em Journal of Computational Biology}, 14(8):1074--1087, 2007.
\newblock PMID: 17985988.

\bibitem[BGHS19]{BGHS19}
Mahdi Boroujeni, Mohammad Ghodsi, MohammadTaghi Hajiaghayi, and Saeed Seddighin.
\newblock $1+\epsilon$ approximation of tree edit distance in quadratic time.
\newblock In {\em Proceedings of the 51st Annual ACM SIGACT Symposium on Theory of Computing (STOC)}, pages 709--720, 2019.

\bibitem[BGK03]{KochBG03}
Peter Buneman, Martin Grohe, and Christoph Koch.
\newblock Path queries on compressed {XML}.
\newblock In {\em Proceedings of the 29th International Conference on Very Large Data Bases (VLDB)}, pages 141--152, 2003.

\bibitem[BGMW20]{BGMW20}
Karl Bringmann, Pawe\l{} Gawrychowski, Shay Mozes, and Oren Weimann.
\newblock Tree edit distance cannot be computed in strongly subcubic time (unless {APSP} can).
\newblock {\em ACM Trans. Algorithms}, 16(4), jul 2020.

\bibitem[BGSV19]{BGSV19}
Karl Bringmann, Fabrizio Grandoni, Barna Saha, and Virginia {Vassilevska Williams}.
\newblock Truly subcubic algorithms for language edit distance and rna folding via fast bounded-difference min-plus product.
\newblock {\em SIAM J. Comput.}, 48(2):481--512, 2019.

\bibitem[BI15]{LBStringED15}
Arturs Backurs and Piotr Indyk.
\newblock Edit distance cannot be computed in strongly subquadratic time (unless {SETH} is false).
\newblock In {\em Proceedings of the 47th Annual ACM Symposium on Theory of Computing (STOC)}, page 51–58, 2015.

\bibitem[BK99]{BellandoK99}
John Bellando and Ravi Kothari.
\newblock Region-based modeling and tree edit distance as a basis for gesture recognition.
\newblock In {\em Proceedings of the 10th International Conference on Image Analysis and Processing (ICIAP)}, pages 698--703, 1999.

\bibitem[CDXZ22]{CDXZ22}
Shucheng Chi, Ran Duan, Tianle Xie, and Tianyi Zhang.
\newblock Faster min-plus product for monotone instances.
\newblock In {\em Proceedings of the 54th Annual ACM SIGACT Symposium on Theory of Computing (STOC)}, pages 1529--1542, 2022.

\bibitem[Cha99]{Chawathe99}
Sudarshan~S. Chawathe.
\newblock Comparing hierarchical data in external memory.
\newblock In {\em Proceedings of the 25th International Conference on Very Large Data Bases (VLDB)}, pages 90--101, 1999.

\bibitem[Che01]{CHEN01}
Weimin Chen.
\newblock New algorithm for ordered tree-to-tree correction problem.
\newblock {\em J. Algorithms}, 40(2):135--158, 2001.

\bibitem[CKM20]{CKM20}
Panagiotis Charalampopoulos, Tomasz Kociumaka, and Shay Mozes.
\newblock {Dynamic String Alignment}.
\newblock In {\em Proceedings of the 31st Annual Symposium on Combinatorial Pattern Matching (CPM)}, pages 9:1--9:13, 2020.

\bibitem[CKW23]{CKW23}
Alejandro Cassis, Tomasz Kociumaka, and Philip Wellnitz.
\newblock Optimal algorithms for bounded weighted edit distance.
\newblock In {\em Proceedings of the 2023 IEEE 64th Annual Symposium on Foundations of Computer Science (FOCS)}, pages 2177--2187, 2023.

\bibitem[DGH{\etalchar{+}}22]{Kociumaka22}
Debarati Das, Jacob Gilbert, MohammadTaghi Hajiaghayi, Tomasz Kociumaka, Barna Saha, and Hamed Saleh.
\newblock Õ(n+poly(k))-time algorithm for bounded tree edit distance.
\newblock In {\em 2022 IEEE 63rd Annual Symposium on Foundations of Computer Science (FOCS)}, pages 686--697, 2022.

\bibitem[DGH{\etalchar{+}}23]{Kociumaka23}
Debarati Das, Jacob Gilbert, MohammadTaghi Hajiaghayi, Tomasz Kociumaka, and Barna Saha.
\newblock Weighted edit distance computation: Strings, trees, and dyck.
\newblock In {\em Proceedings of the 55th Annual ACM Symposium on Theory of Computing (STOC)}, pages 377--390, 2023.

\bibitem[DJW19]{DuanJW19}
Ran Duan, Ce~Jin, and Hongxun Wu.
\newblock Faster algorithms for all pairs non-decreasing paths problem.
\newblock In {\em Proceedings of the 46th International Colloquium on Automata, Languages, and Programming (ICALP)}, 2019.

\bibitem[DKS22]{das2021improved}
Debarati Das, Tomasz Kociumaka, and Barna Saha.
\newblock Improved approximation algorithms for dyck edit distance and {RNA} folding.
\newblock In {\em Proceedings of the 49th International Colloquium on Automata, Languages, and Programming (ICALP)}, pages 49:1--49:20, 2022.

\bibitem[DMRW10]{DMRW10}
Erik~D. Demaine, Shay Mozes, Benjamin Rossman, and Oren Weimann.
\newblock An optimal decomposition algorithm for tree edit distance.
\newblock {\em ACM Trans. Algorithms}, 6(1), dec 2010.

\bibitem[DP09]{DuanP09}
Ran Duan and Seth Pettie.
\newblock Fast algorithms for (max, min)-matrix multiplication and bottleneck shortest paths.
\newblock In {\em Proceedings of the 20th Annual {ACM-SIAM} Symposium on Discrete Algorithms (SODA)}, 2009.

\bibitem[DT03]{DT03}
Serge Dulucq and H{\'e}l{\`e}ne Touzet.
\newblock Analysis of tree edit distance algorithms.
\newblock In {\em Proceedings of the 14th Annual Symposium on Combinatorial Pattern Matching (CPM)}, pages 83--95, 2003.

\bibitem[DT05]{DT05}
Serge Dulucq and Hélène Touzet.
\newblock Decomposition algorithms for the tree edit distance problem.
\newblock {\em J. Discrete Algorithms}, 3(2):448--471, 2005.

\bibitem[D{\"{u}}r23]{Durr23}
Anita D{\"{u}}rr.
\newblock Improved bounds for rectangular monotone min-plus product and applications.
\newblock {\em Inf. Process. Lett.}, 181(C), March 2023.

\bibitem[FGK{\etalchar{+}}24]{fried2024improved}
Dvir Fried, Shay Golan, Tomasz Kociumaka, Tsvi Kopelowitz, Ely Porat, and Tatiana Starikovskaya.
\newblock An improved algorithm for the k-dyck edit distance problem.
\newblock {\em ACM Trans. Algorithms}, 20(3):1--25, 2024.

\bibitem[FLMM09]{FerraginaLMM09}
Paolo Ferragina, Fabrizio Luccio, Giovanni Manzini, and S.~Muthukrishnan.
\newblock Compressing and indexing labeled trees, with applications.
\newblock {\em J. {ACM}}, 57(1):4:1--4:33, 2009.

\bibitem[FM71]{fischermeyer}
Michael~J. Fischer and Albert~R. Meyer.
\newblock Boolean matrix multiplication and transitive closure.
\newblock In {\em Proceedings of the 12th Annual Symposium on Switching and Automata Theory (SWAT)}, pages 129--131, 1971.

\bibitem[GJKT24]{GJKT24}
Daniel Gibney, Ce~Jin, Tomasz Kociumaka, and Sharma~V. Thankachan.
\newblock Near-optimal quantum algorithms for bounded edit distance and lempel-ziv factorization.
\newblock In {\em Proceedings of the 2024 {ACM-SIAM} Symposium on Discrete Algorithms (SODA)}, pages 3302--3332, 2024.

\bibitem[GK24]{GK24}
Egor Gorbachev and Tomasz Kociumaka.
\newblock Bounded edit distance: Optimal static and dynamic algorithms for small integer weights, 2024.

\bibitem[GKS19]{GKS19}
Elazar Goldenberg, Robert Krauthgamer, and Barna Saha.
\newblock Sublinear algorithms for gap edit distance.
\newblock In {\em Proceedings of the 2019 IEEE 60th Annual Symposium on Foundations of Computer Science (FOCS)}, pages 1101--1120, 2019.

\bibitem[GPVX21]{DBLP:conf/icalp/Gu0WX21}
Yuzhou Gu, Adam Polak, Virginia {Vassilevska Williams}, and Yinzhan Xu.
\newblock Faster monotone min-plus product, range mode, and single source replacement paths.
\newblock In {\em Proceedings of the 48th International Colloquium on Automata, Languages, and Programming (ICALP)}, pages 75:1--75:20, 2021.

\bibitem[Gus97]{gusfield_1997}
Dan Gusfield.
\newblock {\em Algorithms on Strings, Trees, and Sequences: Computer Science and Computational Biology}.
\newblock Cambridge University Press, 1997.

\bibitem[HKNS15]{HenzingerKNS15}
Monika Henzinger, Sebastian Krinninger, Danupon Nanongkai, and Thatchaphol Saranurak.
\newblock Unifying and strengthening hardness for dynamic problems via the online matrix-vector multiplication conjecture.
\newblock In {\em Proceedings of the 47th Annual {ACM} on Symposium on Theory of Computing (STOC)}, 2015.

\bibitem[HT84]{HeavyLight}
Dov Harel and Robert~Endre Tarjan.
\newblock Fast algorithms for finding nearest common ancestors.
\newblock {\em SIAM J. Comput.}, 13(2):338--355, 1984.

\bibitem[HTGK03]{HochsmannTGK03}
Matthias H{\"{o}}chsmann, Thomas T{\"{o}}ller, Robert Giegerich, and Stefan Kurtz.
\newblock Local similarity in {RNA} secondary structures.
\newblock In {\em Proceedings of 2nd {IEEE} Computer Society Bioinformatics Conference (CSB)}, pages 159--168, 2003.

\bibitem[Kle98]{Klein98}
Philip~N. Klein.
\newblock Computing the edit-distance between unrooted ordered trees.
\newblock In {\em Proceedings of the 6th Annual European Symposium on Algorithms (ESA)}, pages 91--102, 1998.

\bibitem[Kle05]{K05}
Philip~N. Klein.
\newblock Multiple-source shortest paths in planar graphs.
\newblock In {\em Proceedings of the 16th Annual ACM-SIAM Symposium on Discrete Algorithms (SODA)}, pages 146--155, 2005.

\bibitem[KNW24]{KNW24}
Tomasz Kociumaka, Jakob Nogler, and Philip Wellnitz.
\newblock On the communication complexity of approximate pattern matching.
\newblock In {\em Proceedings of the 56th Annual ACM Symposium on Theory of Computing (STOC)}, pages 1758--1768, 2024.

\bibitem[KSK01]{KleinSK01}
Philip~N. Klein, Thomas~B. Sebastian, and Benjamin~B. Kimia.
\newblock Shape matching using edit-distance: an implementation.
\newblock In {\em Proceedings of the 12th Annual Symposium on Discrete Algorithms (SODA)}, pages 781--790, 2001.

\bibitem[KTSK00]{KleinTSK00}
Philip~N. Klein, Srikanta Tirthapura, Daniel Sharvit, and Benjamin~B. Kimia.
\newblock A tree-edit-distance algorithm for comparing simple, closed shapes.
\newblock In {\em Proceedings of the 11th Annual {ACM-SIAM} Symposium on Discrete Algorithms (SODA)}, pages 696--704, 2000.

\bibitem[LMS98]{LMS98}
Gad~M. Landau, Eugene~W. Myers, and Jeanette~P. Schmidt.
\newblock Incremental string comparison.
\newblock {\em SIAM J. Comput.}, 27(2):557--582, 1998.

\bibitem[LV88]{LV88}
Gad~M. Landau and Uzi Vishkin.
\newblock Fast string matching with k-differences.
\newblock {\em J. Comput. Syst. Sci.}, 37(1):63--78, August 1988.

\bibitem[Mao22]{M22}
Xiao Mao.
\newblock Breaking the cubic barrier for (unweighted) tree edit distance.
\newblock In {\em Proceedings of the 2021 IEEE 62nd Annual Symposium on Foundations of Computer Science (FOCS)}, pages 792--803, 2022.

\bibitem[MTWZU09]{MTWZ09}
Shay Mozes, Dekel Tsur, Oren Weimann, and Michal Ziv-Ukelson.
\newblock Fast algorithms for computing tree lcs.
\newblock {\em Theor. Comput. Sci.}, 410(43):4303–4314, October 2009.

\bibitem[NJ80]{NJ80}
Ruth Nussinov and Ann~B. Jacobson.
\newblock Fast algorithm for predicting the secondary structure of single-stranded rna.
\newblock {\em Proc. Natl. Acad. Sci. U.S.A.}, 77(11):6309--6313, 1980.

\bibitem[Sah14]{s14}
Barna Saha.
\newblock The dyck language edit distance problem in near-linear time.
\newblock In {\em Proceedings of the 2014 IEEE 55th Annual Symposium on Foundations of Computer Science (FOCS)}, pages 611--620, 2014.

\bibitem[Sah17]{s17}
Barna Saha.
\newblock Fast \& space-efficient approximations of language edit distance and rna folding: An amnesic dynamic programming approach.
\newblock In {\em Proceedings of the 2017 IEEE 58th Annual Symposium on Foundations of Computer Science (FOCS)}, pages 295--306. IEEE, 2017.

\bibitem[Sch95]{S95}
J.P. Schmidt.
\newblock All shortest paths in weighted grid graphs and its application to finding all approximate repeats in strings.
\newblock In {\em Proceedings Third Israel Symposium on the Theory of Computing and Systems}, pages 67--77, 1995.

\bibitem[Sel77]{Selkow77}
Stanley~M. Selkow.
\newblock The tree-to-tree editing problem.
\newblock {\em Inf. Process. Lett.}, 6(6):184--186, 1977.

\bibitem[SKK04]{SebastianKK04}
Thomas~B. Sebastian, Philip~N. Klein, and Benjamin~B. Kimia.
\newblock Recognition of shapes by editing their shock graphs.
\newblock {\em {IEEE} Trans. Pattern Anal. Mach. Intell.}, 26(5):550--571, 2004.

\bibitem[SPA17]{SPA17}
Stefan Schwarz, Mateusz Pawlik, and Nikolaus Augsten.
\newblock A new perspective on the tree edit distance.
\newblock In {\em Proceedings of the 10th International Conference on Similarity Search and Applications (SISAP)}, pages 156--170, 2017.

\bibitem[SS22]{Seddighin22}
Masoud Seddighin and Saeed Seddighin.
\newblock {$3+\epsilon$} approximation of tree edit distance in truly subquadratic time.
\newblock In {\em Proceedings of the 13th Innovations in Theoretical Computer Science Conference (ITCS)}, pages 115:1--115:22, 2022.

\bibitem[SY24]{SY24}
Barna Saha and Christopher Ye.
\newblock Faster approximate all pairs shortest paths.
\newblock In {\em Proceedings of the 2024 Annual ACM-SIAM Symposium on Discrete Algorithms (SODA)}, pages 4758--4827, 2024.

\bibitem[SZ89]{SZ89}
Dennis Shasha and Kaizhong Zhang.
\newblock Simple fast algorithms for the editing distance between trees and related problems.
\newblock {\em SIAM J. Comput.}, 18(6):1245--1262, 1989.

\bibitem[SZ90]{10.1093/bioinformatics/6.4.309}
Bruce~A. Shapiro and Kaizhong Zhang.
\newblock {Comparing multiple RNA secondary structures using tree comparisons}.
\newblock {\em Bioinformatics}, 6(4):309--318, 10 1990.

\bibitem[Tai79]{Tai79}
Kuo-Chung Tai.
\newblock The tree-to-tree correction problem.
\newblock {\em J. ACM}, 26(3):42--433, jul 1979.

\bibitem[Tis06]{T06}
Alexander Tiskin.
\newblock All semi-local longest common subsequences in subquadratic time.
\newblock In {\em Proceedings of the 1st International Computer Science Conference on Theory and Applications (CSR)}, pages 352--363, 2006.

\bibitem[Tou05]{T05}
H{\'e}l{\`e}ne Touzet.
\newblock A linear tree edit distance algorithm for similar ordered trees.
\newblock In Alberto Apostolico, Maxime Crochemore, and Kunsoo Park, editors, {\em Combinatorial Pattern Matching}, pages 334--345, Berlin, Heidelberg, 2005. Springer Berlin Heidelberg.

\bibitem[{Vas}10]{nondecpaths}
Virginia {Vassilevska Williams}.
\newblock Nondecreasing paths in a weighted graph or: How to optimally read a train schedule.
\newblock {\em {ACM} Trans. Algorithms}, 6(4), 2010.

\bibitem[{Vas}18]{vsurvey}
Virginia {Vassilevska Williams}.
\newblock On some fine-grained questions in algorithms and complexity.
\newblock In {\em Proceedings of the ICM}, volume~3, pages 3431--3472. World Scientific, 2018.

\bibitem[VGF14]{venkatachalam2014RNA}
Balaji Venkatachalam, Dan Gusfield, and Yelena Frid.
\newblock Faster algorithms for rna-folding using the four-russians method.
\newblock {\em Algorithms Mol. Biol.}, 9:1--12, 2014.

\bibitem[VW18]{VW18}
Virginia {Vassilevska Williams} and R.~Ryan Williams.
\newblock Subcubic equivalences between path, matrix, and triangle problems.
\newblock {\em J. ACM}, 65(5):1--38, 2018.

\bibitem[VWY09]{apbp}
Virginia Vassilevska, Ryan Williams, and Raphael Yuster.
\newblock All pairs bottleneck paths and max-min matrix products in truly subcubic time.
\newblock {\em Theory Comput.}, 5(1):173--189, 2009.

\bibitem[VX20]{DBLP:conf/soda/WilliamsX20}
Virginia {Vassilevska Williams} and Yinzhan Xu.
\newblock Truly subcubic min-plus product for less structured matrices, with applications.
\newblock In {\em Proceedings of the 2020 {ACM-SIAM} Symposium on Discrete Algorithms (SODA)}, pages 12--29, 2020.

\bibitem[Wat95]{waterman1995introduction}
Michael~S. Waterman.
\newblock {\em Introduction to computational biology: maps, sequences and genomes}.
\newblock CRC Press, 1995.

\bibitem[Wil18]{DBLP:journals/siamcomp/Williams18}
R.~Ryan Williams.
\newblock Faster all-pairs shortest paths via circuit complexity.
\newblock {\em {SIAM} J. Comput.}, 47(5):1965--1985, 2018.

\bibitem[ZS89]{ShashaZhang89}
Kaizhong Zhang and Dennis Shasha.
\newblock Simple fast algorithms for the editing distance between trees and related problems.
\newblock {\em SIAM J. Comput.}, 18(6):1245--1262, 1989.

\bibitem[ZTZU11]{zakov2011RNA}
Shay Zakov, Dekel Tsur, and Michal Ziv-Ukelson.
\newblock Reducing the worst case running times of a family of {RNA} and {CFG} problems, using {Valiant}’s approach.
\newblock {\em Algorithms Mol. Biol.}, 6:1--22, 2011.

\end{thebibliography}

\end{document}